 \pdfoutput=1
\documentclass[acmsmall,screen]{acmart}

\usepackage{amsfonts}
\usepackage{algorithm}
\usepackage{algorithmic}
\usepackage{enumerate}
\usepackage{subfigure}
\usepackage{amsmath}
\usepackage{courier}
\usepackage{setspace}
\usepackage{multirow}
\usepackage{booktabs}
\usepackage{soul}
\usepackage{bm}
\usepackage{threeparttable}
\usepackage[capitalize,noabbrev]{cleveref}

\theoremstyle{plain}
\newtheorem{theorem}{Theorem}
\newtheorem{proposition}{Proposition}
\newtheorem{lemma}{Lemma}
\newtheorem{corollary}[theorem]{Corollary}
\theoremstyle{definition}

\theoremstyle{remark}
\newtheorem{remark}{Remark}

\newcommand{\atn}[1]{\textcolor{black}{#1}}
\newcommand{\atnn}[1]{\textcolor{black}{#1}}
\newcommand{\notice}[1]{\textcolor{black}{#1}}

\AtBeginDocument{%
	\providecommand\BibTeX{{%
			\normalfont B\kern-0.5em{\scshape i\kern-0.25em b}\kern-0.8em\TeX}}}

\begin{document}
	
	\title{Algorithm xxx: Faster Randomized SVD with Dynamic Shifts}
\thanks{This work was supported by the National Natural Science Foundation of China 61872206. W. Yu is the corresponding author.}	
	\author{Xu Feng}
	\email{fx17@mails.tsinghua.edu.cn}
	\author{Wenjian Yu}
	\email{yu-wj@tsinghua.edu.cn}
	\author{Yuyang Xie}
	\email{xyy18@mails.tsinghua.edu.cn}
	\author{Jie Tang}
	\email{jietang@tsinghua.edu.cn}
	\affiliation{%
		\institution{Department of Computer Science and Technology, BNRist, Tsinghua University}
		\city{Beijing}
		\postcode{100084}
		\country{China}
	}

	\renewcommand{\shortauthors}{X. Feng, Y. Xie, W. Yu, and J. Tang}
	
	\begin{abstract}
Aiming to provide a faster and convenient truncated SVD algorithm for large sparse matrices from real applications (i.e. for computing a few of largest singular values and the corresponding singular vectors), a dynamically shifted power iteration technique is applied to improve the accuracy of the randomized SVD method. This results in a \underline{d}yn\underline{a}mic \underline{sh}ifts based randomized \underline{SVD} (dashSVD) algorithm, which also collaborates with the skills for handling sparse matrices. An accuracy-control mechanism is included in the dashSVD algorithm to approximately monitor the per vector error bound of computed singular vectors with negligible overhead. Experiments on real-world data validate that the dashSVD algorithm largely improves the accuracy of randomized SVD algorithm or attains same accuracy with fewer passes over the matrix, and provides an efficient accuracy-control mechanism to the randomized SVD computation, while demonstrating the advantages on runtime  and parallel efficiency. A bound of the approximation error of the randomized SVD with the shifted power iteration is also proved.
	\end{abstract}

	

	\maketitle

\section{Introduction}\label{sec1}
In machine learning and data mining, truncated singular value decomposition (SVD), i.e. computing a few of largest singular values and corresponding singular vectors, is widely used in dimension reduction, information retrieval, matrix completion, etc \cite{musco2015,pmlr-v95-feng18a,ding2020efficient}. The truncated SVD computes optimal low-rank approximation and principal component analysis (PCA). Specifically, the top singular vector $\mathbf{u}_1$ of matrix $\mathbf{A}$ provides the top principal component, which reflects the direction of the largest variance within $\mathbf{A}$. The $i$-th left singular vector $\mathbf{u}_i$ provides the $i$-th principal component, which reflects the direction of the largest variance orthogonal to all higher principal components.

However, for large and high-dimensional input data from social network analysis, natural language processing and recommender system,  computing truncated SVD of the corresponding matrix often consumes tremendous time and memory usage. For large data matrices which are usually very sparse, the preferred method of computing truncated SVD is using \texttt{svds} in Matlab \cite{baglama2005augmented}. \texttt{svds} is still the most robust and efficient in general cases compared with the variant algorithms of \texttt{svds} like lansvd~\cite{propack}. To tackle the difficulty of computing truncated SVD for large sparse data, various algorithms have been proposed~\cite{Halko2011Finding,musco2015,LazySVD,rsvdpack,alg971,pmlr-v95-feng18a,wu2017primme_svds}. They consume less computational resource, but usually induce a little accuracy loss. Among these algorithms, the randomized method based on random embedding through multiplying a random matrix \cite{martinsson2020randomized} gains  a lot of attention. The randomized method produces a near-optimal truncated SVD of the matrix, while exhibiting performance advantages over the classical methods (like less runtime, fewer passes over the matrix and better parallelizability, etc). Therefore, the randomized method is very favorable for the scenarios where only low-accuracy SVD is required, like in the field of machine learning.
More analysis on relevant techniques and theories can be found in \cite{Halko2011Finding,martinsson2020randomized}.

In this work, we aim to provide an algorithm for computing the truncated SVD of large sparse matrices efficiently on a multi-core computer. Inspired by the shift technique in the power method~\cite{matrix2012}, we develop a dynamic shifts based randomized SVD (called dashSVD) algorithm to improve the computational efficiency of the randomized SVD. An efficient accuracy-control scheme is also developed and integrated into the dashSVD algorithm. 
Our major contributions and results are as follows.
\begin{itemize}
	\setlength{\itemsep}{1pt}
	\setlength{\parsep}{1pt}
	\setlength{\parskip}{1pt}
	\item A dynamic scheme for setting the shift values in the shifted power iteration is developed to accelerate the randomized SVD algorithm. It improves the accuracy of the result or reduces the number of power iterations for attaining same accuracy. Combining the scheme with the accelerating skills for handling sparse matrices, we have further developed a dynamic shifts based randomized SVD (dashSVD) algorithm.
	\item 
	Based on the per vector error (PVE) criterion \cite{musco2015}, an efficient accuracy-control mechanism is developed and integrated into the dashSVD algorithm. It resolves the difficulty of setting a suitable power parameter $p$, and enables automatic termination of the power iteration  according to the PVE bound based accuracy criterion.
	\item Experiments on real-world data have validated the efficiency of these techniques, demonstrating that dashSVD runs 
	faster than the state-of-the-art algorithms for attaining a not very high accuracy (e.g. with PVE error $\epsilon_{\textrm{PVE}}\ge 10^{-2}$), with  comparable memory usage. On dataset uk-2005, dashSVD runs
	3.2X faster than the \texttt{LanczosBD} algortihm in \texttt{svds} for attaining the accuracy corresponding to PVE error $\epsilon_{\textrm{PVE}}=10^{-1}$ with serial computing, and runs 4.0X faster than PRIMME\_SVDS \cite{wu2017primme_svds} with parallel computing employing 8 threads. The experiments also reveal that dashSVD is more robust than the existing fast SVD algorithms \cite{propack,wu2017primme_svds}.
\end{itemize}

It should be noted that, in this work we just consider the efficient implementation of the randomized SVD on a shared-memory machine. Its efficient implementation in a distributed-memory parallelism context is also of great interest \cite{ICML2023}. The codes of dashSVD are shared on GitHub (\url{https://github.com/THU-numbda/dashSVD})

The rest of this paper is structured as follows. Section 2 introduces the basics of the randomized SVD via random embedding and the state-of-the-art truncated SVD algorithms. In Section 3, we present the dashSVD algorithm with the dynamic scheme for setting the shift values in the shifted power iteration, and the accuracy-control mechanism based on PVE criterion. 
Numerical experiments for performance analysis are presented in \notice{Section 4}. Finally, we draw the conclusions. More theoretical proofs and experimental results are given in the Appendix.

\subsection{Related Work}\label{sec1.1}
A basic randomized SVD algorithm was presented in \cite{Halko2011Finding}, 
where a power iteration technique was employed to improve accuracy while sacrificing computational time. A couple of works were later proposed to accelerate it. 
The amount of orthonormalization in the power iteration is reduced and the eigenvalue decomposition (EVD) is used to compute economic SVD faster but sacrificing little numerical stability \cite{rsvdpack}.  H. Li et al. proposed to replace the QR factorization for orthonormalization in power iteration with LU factorization \cite{alg971}, to reduce the computational time. N. B. Erichson et al. employed a variant of the randomized SVD algorithm with several acceleration tricks in image and video processing problems \cite{Erichson_2017_ICCV}. Later, an algorithm called frPCA \cite{pmlr-v95-feng18a} was proposed which collaborates the skills in \cite{rsvdpack,Erichson_2017_ICCV,alg971} and is specialized to accelerate the truncated (top-$k$) SVD computation of large sparse matries.

As an alternative of power iteration, the technique of block Krylov iteration (BKI) also improves the accuracy of the randomized SVD. It consumes less CPU time than the algorithm with power iteration to attain the same accuracy~\cite{musco2015}, but consumes much larger memory. SLEPc in \cite{andez2008robust} is an efficient library to compute both SVD and generalized SVD in both distributed and shared memory environment. PRIMME\_SVDS in \cite{wu2017primme_svds} is based on eigenvalue package PRIMME \cite{2010primme} and could efficiently compute both the largest and smallest singular values, which is more efficient than SLEPc for computing truncated SVD.
LazySVD in \cite{LazySVD} is based on the Lanczos process. It lacks accuracy control and only computes the left singular vectors.

\section{Preliminary}\label{sec2}
	We use the conventions from Matlab in this paper to specify indices of matrices and functions.
	
	\subsection{Basics of Truncated SVD}\label{sec2.1}
	
	The economic SVD of a matrix $\mathbf{A}\in\mathbb{R}^{m\times n}$ ($m\ge n$) is
    \begin{equation}
    \label{esvd}
     \mathbf{A} = \mathbf{U\Sigma V}^\mathrm{T}
	,
    \end{equation}
	where $\mathbf{U}=[\mathbf{u}_1,\mathbf{u}_2,\cdots,\mathbf{u}_n]$ and $\mathbf{V}=[\mathbf{v}_1,\mathbf{v}_2,\cdots,\mathbf{v}_n]$ are matrices containing the left and right singular vectors of $\mathbf{A}$, respectively. And, $\mathbf{\Sigma}$ is an $n\times n$ diagonal matrix  containing the singular values ($\sigma_1, \sigma_2, \cdots,\sigma_n$) of $\mathbf{A}$ in descending order. \{$\mathbf{u}_i, \mathbf{v}_i,\sigma_i$\} is called the $i$-th singular triplet, and $\sigma_i(\cdot)$ also denotes the $i$-th largest singular value. From (\ref{esvd}), one can obtain the truncated SVD of $\mathbf{A}$:
	\begin{equation}
		\mathbf{A}_k =  \mathbf{U}_k\mathbf{\Sigma}_k\mathbf{V}_k^\mathrm{T}, ~ ~ k<\min(m,n),\end{equation}
	where $\mathbf{U}_k$ and $\mathbf{V}_k$ are matrices with the first $k$ columns of $\mathbf{U}$ and $\mathbf{V}$ respectively, and the diagonal matrix $\mathbf{\Sigma}_k$ is the $k\times k$ upper-left submatrix of $\mathbf{\Sigma}$. Notice that $\mathbf{A}_k$ is the best rank-$k$ approximation of $\mathbf{A}$ in both spectral norm and Frobenius norm \citep{eckart1936}. \texttt{svds} in Matlab is based on the Lanczos bidiagonalization process and an augmented restarting scheme which reduces the memory cost and ensures accuracy \citep{baglama2005augmented}. \texttt{svds} has been widely regarded as the standard tool for computing the truncated SVD.

\subsection{Randomized SVD Algorithm with Power Iteration}\label{sec2.2}
	
	The basic randomized SVD algorithm \citep{Halko2011Finding} can be described as Algorithm~1, where $\mathbf{\Omega}$ is a Gaussian i.i.d random matrix, and the orthonormalization operation ``orth($\cdot$)'' can be implemented with a call to a packaged QR factorization. The power iteration in Step 3 through 5 is for improving the accuracy of 
the result, where the orthonormalization alleviating the round-off error in floating-point computation is performed. Here we employ the skill of performing the orthonormalization after every other matrix-matrix multiplication to save computational cost with little accuracy loss \citep{rsvdpack,pmlr-v95-feng18a}, as compared to doing orthonormalization after every matrix-matrix multiplication \cite{Halko2011Finding}.

	\begin{algorithm}[ht]
		\caption{Basic randomized SVD with power iteration}
		\label{alg1}
		\begin{algorithmic}[1]
			\REQUIRE $\mathbf{A}\in\mathbb{R}^{m\times n}$, rank parameter $k$, oversampling parameter $s$, power parameter $p$
			\ENSURE$\mathbf{U}\in\mathbb{R}^{m\times k}$, $\mathbf{S}\in\mathbb{R}^{k\times k}$, $\mathbf{V}\in\mathbb{R}^{n\times k}$
			\STATE $l=k+s$, $\mathbf{\Omega} = \mathrm{randn}(n, l)$
			\STATE $\mathbf{Q} = \mathrm{orth}(\mathbf{A\Omega})$
			\FOR {$j=1, 2, \cdots, p$}
			\STATE $\mathbf{Q} = \mathrm{orth}(\mathbf{AA}^\mathrm{T}\mathbf{Q})$
			\ENDFOR
			\STATE $\mathbf{B}  = \mathbf{Q}^{\mathrm{T}}\mathbf{A}$
			\STATE $[\mathbf{U}, \mathbf{S}, \mathbf{V}] = \mathrm{svd}(\mathbf{B}, \mathrm{'econ'})$
			\STATE $\mathbf{U} = \mathbf{Q}\mathbf{U}(:, 1\!:\!k), \mathbf{S} \!=\!\mathbf{S}(1\!:\!k, 1\!:\!k), \mathbf{V}\! =\! \mathbf{V}(:, 1\!:\!k)$
		\end{algorithmic}
	\end{algorithm}
	
	If there is no power iteration, the $m\times l$ orthonormal matrix $\mathbf{Q}=\mathrm{orth}(\mathbf{A\Omega})$ is an  approximation of the basis of dominant subspace of $range(\mathbf{A})$, i.e.,  $span\{\mathbf{u}_1, \mathbf{u}_2, \cdots, \mathbf{u}_l\}$. 
	Therefore, $\mathbf{A}\approx\mathbf{QQ}^\mathrm{T}\mathbf{A}=\mathbf{QB}$ according to Step 6.
	When the economic SVD is performed  on the short-and-fat $l\!\times\! n$ matrix $\mathbf{B}$, the approximate truncated SVD of $\mathbf{A}$ is finally obtained. 
	Employing the power iteration, one obtains $\mathbf{Q}=\mathrm{orth}((\mathbf{AA}^{\mathrm{T}})^p\mathbf{A}\mathbf{\Omega})$, if the intermediate orthonormalization steps are ignored. This makes $\mathbf{Q}$ better approximate the basis of dominant subspace of $range((\mathbf{AA}^{\mathrm{T}})^p\mathbf{A})$, same as that of $range(\mathbf{A})$, because $(\mathbf{AA}^{\mathrm{T}})^p\mathbf{A}$'s singular values decay more quickly than those of $\mathbf{A}$ \citep{Halko2011Finding}. So, the computed singular triplets are more accurate, and the larger $p$ 
 leads to
 more accurate results and more computational cost as well. According to the following Lemma, i.e. (3.3.17) and (3.3.18) in \cite{horn1991topics}, we can derive \cref{proposition:2} which shows the relationship between the singular values of $\mathbf{A}$ and the singular values computed by Alg.~1. 
	
	\begin{lemma}
	\label{lemma:3}
	Suppose $\mathbf{A},\mathbf{C}\in\mathbb{R}^{m\times n}$. The following inequalities hold for the decreasingly ordered singular values of $\mathbf{A}$, $\mathbf{C}$ and $\mathbf{AC}^\mathrm{T}$ ($1\le i, j, i+j -1  \le \min(m, n)$) 
	\begin{equation}
		\label{lemma3:1}
		\sigma_{i+j-1}(\mathbf{AC}^{\mathrm{T}}) \le \sigma_i(\mathbf{A})\sigma_j(\mathbf{C}) ~,
	\end{equation} 
	and
	\begin{equation}
		\label{lemma3:2}
		\sigma_{i+j-1}(\mathbf{A}+\mathbf{C}) \le \sigma_i(\mathbf{A})+\sigma_j(\mathbf{C}) ~.
	\end{equation}
\end{lemma}

\begin{proposition}
	\label{proposition:2}
	Suppose $\mathbf{A}\in\mathbb{R}^{m\times n}$ and $\mathbf{Q}\in\mathbb{R}^{n\times l}~(l\le \min(m, n)) $ is an orthonormal matrix. Then,
	\begin{equation}
		\label{proposition2:1}
		\sigma_{i}(\mathbf{AQ}) \le \sigma_i(\mathbf{A})~,~\textrm{for any}~i\le l~.
	\end{equation}
\end{proposition}

\begin{proof}
	We append zero columns to $\mathbf{Q}$ to get an $n\times m$ matrix $\mathbf{C}^\mathrm{T}=[\mathbf{Q}, \mathbf{0}]\in\mathbb{R}^{n\times m}$. Since $\mathbf{Q}$ is an orthonormal matrix, $\sigma_1(\mathbf{C})=1$. According to (\ref{lemma3:1}) in \cref{lemma:3},
	\begin{equation}
		\label{p2:1}
		\sigma_i(\mathbf{AC}^\mathrm{T}) \le \sigma_i(\mathbf{A})\sigma_1(\mathbf{C}) = \sigma_i(\mathbf{A}).
	\end{equation}
	Because $\mathbf{AC}^\mathrm{T}\!=\![\mathbf{AQ}, \mathbf{0}]$, for any $i\le l$, $\sigma_i(\mathbf{AQ}) = \sigma_i(\mathbf{AC}^\mathrm{T})$. Then, combining (\ref{p2:1}) we can prove (\ref{proposition2:1}).
\end{proof}

\begin{remark}
	Because the singular values of $\mathbf{B}$ in Step 6 of Alg. 1 equal to those of $\mathbf{B}^\mathrm{T}=\mathbf{A}^\mathrm{T}\mathbf{Q}$%
	, \cref{proposition:2} infers that the singular value computed with Alg. 1 is \emph{equal to or less than} its accurate value, i.e. $\sigma_i(\mathbf{B})\le \sigma_i(\mathbf{A})$. This explains that the singular value curve computed with the randomized SVD algorithm is always \emph{underneath} the accurate curve of singular value, as shown in literature.
\end{remark}
	For a sparse matrix $\mathbf{A}$, the power iteration costs a larger portion of the total time as it involves the manipulation of dense matrices. So, accelerating the computations in power iteration becomes important. Besides reducing the amount of orthonormalization, the skills of using LU factorization to replace QR factorization and eigenvalue decomposition (EVD) based SVD to do orthonormalization have been employed. Along with the techniques for efficiently handling the matrix with more columns than rows and allowing  odd number of passes
	over $\mathbf{A}$, an algorithm called frPCA was developed in \cite{pmlr-v95-feng18a} for faster truncated SVD of sparse matrices. However, for most real-world sparse matrices whose singular values decay slowly, a large number of  power iteration steps are required 
	to attain satisfied accuracy, and there is a lack of efficient mechanism to determine the power parameter $p$. 
 
 Below we use the floating-point operation (flop) count to analyze the computational cost. Suppose $C_{\mathrm{mul}}$, $C_{\mathrm{qr}}$ and $C_{\mathrm{svd}}$ represent the constants in the flop counts of matrix-martrix multiplication, QR factorization and economic SVD, respectively. And, $\mathrm{nnz}(\mathbf{A})$ is the number of nonzero elements in $\mathbf{A}$.
 According to the procedure, the flop count of Alg.~1 is:
    \begin{equation}
    \mathrm{FC}_{1} = (2p+2)C_{mul}\mathrm{nnz}(\mathbf{A})l+(p+1)C_{qr}ml^2+C_{svd}nl^2+C_{mul}mlk,
    \end{equation}
where $(2p+2)C_{mul}\mathrm{nnz}(\mathbf{A})l$ reflects the matrix-matrix multiplication on $\mathbf{A}$, $(p+1)C_{qr}ml^2$ reflects the QR factorization in Step 2 and 4, $C_{svd}nl^2$ reflects the economic SVD in Step 7 and $C_{mul}mlk$ reflects the matrix-matrix multiplication in Step 8.

\subsection{Algorithms in \texttt{svds} and PRIMME\_SVDS}\label{sec2.3}
	
	\texttt{svds} is the most robust and well-known tool to compute truncated SVD based on the Lanczos bidiagonalization process with augmented restarting scheme in \cite{baglama2005augmented}. In \texttt{svds}, the relative  residual of singular vectors 
	\begin{equation}
		\label{err:svds}
		\max\limits_{i\le k}\frac{\Vert\mathbf{A}^\mathrm{T}\mathbf{\hat{u}}_i-\hat{\sigma}_i\mathbf{\hat{v}}_i\Vert_2}{\hat{\sigma}_i}
	\end{equation}
	is calculated to control the accuracy, where $\{\mathbf{\hat{u}}_i, \hat{\sigma}_i, \mathbf{\hat{v}}_i\}$ is the computed $i$-th singular triplet. The Lanczos bidiagonalization process with augmented restarting scheme is called \texttt{LanczosBD} in Matlab. \texttt{svds} first runs \texttt{LanczosBD} until the first $k$ singular triplets make (\ref{err:svds}) smaller than a preset tolerance. Then, one or more extra invocations of \texttt{LanczosBD} for computing the first $k+1$ singular triples are executed to ensure the robustness of \texttt{svds}. Therefore, \texttt{svds} can produce accurate truncated SVD in almost all scenarios.

Recently, a high-performance parallel SVD solver called PRIMME\_SVDS \citep{wu2017primme_svds} was developed. For a matrix $\mathbf{A}$, PRIMME\_SVDS first uses the state-of-the-art truncated EVD library PRIMME \citep{2010primme} to compute EVD of $\mathbf{A}^\mathrm{T}\mathbf{A}$. Then, the obtained results are used as input vectors to solve the truncated EVD of $\left[\begin{smallmatrix}\boldsymbol{0} & \mathbf{A}^{\mathrm{T}}\\\mathbf{A} & \boldsymbol{0}\end{smallmatrix}\right]$ with PRIMME, to ensure the accuracy of computed singular vectors. The relative residuals of eigenvectors and eigenvalues are also used in PRIMME for accuracy control, which is similar to (\ref{err:svds}) in \texttt{svds}.
	
\section{Randomized SVD with Dynamic Shifts} \label{sec3}
In this section, we first introduce how to improve the accuracy of the randomized SVD with the shifted power iteration and
a dynamic scheme of setting the shift value. Then, we discuss the termination criteria for accuracy control and devise an efficient algorithm based on PVE criterion.

\subsection{Shifted Power Iteration and Setting Dynamic Shifts}\label{sec3.1}

The computation $\mathbf{Q}=\mathbf{AA}^\mathrm{T}\mathbf{Q}$ in the power iteration of Alg.~1 is the same as that in the power method for computing the largest eigenvalue and corresponding eigenvector of $\mathbf{AA}^\mathrm{T}$. The shift skill can be used to accelerate the convergence of the power method because of the reduction of the ratio between the second largest eigenvalue and the largest one \citep{matrix2012}. This inspires the idea of using shifts in the power iteration of randomized SVD algorithm. We first give two Lemmas \citep{matrix2012} to derive the power iteration with the shift value.
\begin{lemma}
	\label{lemma:1}
	For a symmetric real-valued matrix $\mathbf{A}$, its singular values are the absolute values of its eigenvalues. For any eigenvalue $\lambda$ of $\mathbf{A}$, the left singular vector corresponding to singular value $\vert\lambda\vert$ is the normalized eigenvector for $\lambda$.
\end{lemma}

\begin{lemma}
	\label{lemma:2}
	Suppose matrix $\mathbf{A} \in \mathbb{R}^{n\times n}$, and a shift $\alpha \in \mathbb{R}$. For any eigenvalue $\lambda$ of $\mathbf{A}$,  $\lambda-\alpha$ is an eigenvalue of  $\mathbf{A}-\alpha\mathbf{I}$, where $\mathbf{I}$ is the identity matrix. And, the eigenspace of $\mathbf{A}$ for $\lambda$ is the same as the eigenspace of $\mathbf{A}-\alpha\mathbf{I}$ for $\lambda-\alpha$.
\end{lemma}

For any matrix $\mathbf{A}$, $\mathbf{AA}^\mathrm{T}$ is a symmetric positive semi-definite matrix, so its singular value is its eigenvalue according to \cref{lemma:1}.  \cref{lemma:2} shows that $\sigma_i(\mathbf{AA}^\mathrm{T})-\alpha$ is the eigenvalue of $\mathbf{AA}^\mathrm{T}\!-\!\alpha\mathbf{I}$. And, $\vert\sigma_i(\mathbf{AA}^\mathrm{T})\!-\!\alpha\vert$ is the singular value of $\mathbf{AA}^\mathrm{T}\!-\!\alpha\mathbf{I}$ according to \cref{lemma:1}.
In Alg. 1, the decay trend of the 
largest $l$ singular values of handled matrix affects the accuracy of resulted SVD. When  $\sigma_i(\mathbf{AA}^\mathrm{T})-\alpha >0$ for any $i \le l$, and they are the $l$ largest singular values of $\mathbf{AA}^\mathrm{T}-\alpha\mathbf{I}$, they obviously exhibit faster decay. The following Proposition states when these conditions are satisfied.

\begin{proposition}
	\label{proposition:1}
	Suppose $0<\alpha \!\le \! \sigma_l(\mathbf{A}\mathbf{A}^\mathrm{T})/2$ and $i\le l$. Then, $\sigma_i(\mathbf{AA}^\mathrm{T}-\alpha \mathbf{I}) = \sigma_i(\mathbf{AA}^\mathrm{T})-\alpha$. Moreover, if $\sigma_i(\mathbf{AA}^\mathrm{T}-\alpha \mathbf{I})\neq \sigma_{l+1}(\mathbf{AA}^\mathrm{T}-\alpha \mathbf{I})$, the left singular vector corresponding to the $i$-th largest singular value of $\mathbf{A}\mathbf{A}^\mathrm{T}\!-\!\alpha\mathbf{I}$ 
	is
	the left singular vector corresponding to the $i$-th largest singular value of $\mathbf{A}\mathbf{A}^\mathrm{T}$, and vice versa.
\end{proposition}
\begin{proof}
When $0<\alpha \!\le\! \sigma_l(\mathbf{A}\mathbf{A}^\mathrm{T})/2$, we can derive $2\alpha\le\sigma_l(\mathbf{AA}^\mathrm{T})$ and $0<\alpha\le\sigma_l(\mathbf{AA}^\mathrm{T})-\alpha$. Therefore, $\alpha-\sigma_i(\mathbf{AA}^\mathrm{T})\le\alpha\le\sigma_l(\mathbf{AA}^\mathrm{T})-\alpha$.
	For any $i>l$
	\begin{equation}
		\label{p1:3}		\sigma_i(\mathbf{AA}^\mathrm{T})-\alpha\le\sigma_l(\mathbf{AA}^\mathrm{T})-\alpha.
	\end{equation}
	Then, we can derive
	\begin{equation}
		\label{p1:4}
		\vert\sigma_i(\mathbf{AA}^\mathrm{T})-\alpha\vert\le\sigma_l(\mathbf{AA}^\mathrm{T})-\alpha~,~\textrm{for any}~i>l.
	\end{equation}
	Besides, for any $i\le l$ we can derive
\begin{equation}
		\label{p1:5}
		\vert\sigma_i(\mathbf{AA}^\mathrm{T})-\alpha\vert = \sigma_i(\mathbf{AA}^\mathrm{T})-\alpha\ge\sigma_l(\mathbf{AA}^\mathrm{T})-\alpha>0, ~\textrm{for any}~i \le l.
\end{equation}
Notice that $\sigma_i(\mathbf{AA}^\mathrm{T})$ is also an eigenvalue of $\mathbf{AA}^\mathrm{T}$, and $\vert\sigma_i(\mathbf{AA}^\mathrm{T})-\alpha\vert$ is a singular value of $\mathbf{AA}^\mathrm{T}-\alpha \mathbf{I}$ according to Lemma 2 and 3. So, combining (\ref{p1:4}) and (\ref{p1:5}) we can see that
 $\vert\sigma_i(\mathbf{AA}^\mathrm{T})-\alpha\vert$, $1\le i\le l$, are the $l$ largest singular values of $\mathbf{AA}^\mathrm{T}-\alpha\mathbf{I}$, and 
	\begin{equation}
		\label{p1:6}
		\sigma_i(\mathbf{AA}^\mathrm{T}-\alpha\mathbf{I})=\sigma_i(\mathbf{AA}^\mathrm{T})-\alpha~,~\textrm{for any}~i\le l.
	\end{equation}
	\cref{lemma:1} shows that, for any $i\le l$ the left singular vector corresponding to $\sigma_i(\mathbf{AA}^\mathrm{T})$ is the same as the normalized eigenvector for eigenvalue $\sigma_i(\mathbf{AA}^\mathrm{T})$ of $\mathbf{AA}^\mathrm{T}$, and the left singular vector corresponding to $\sigma_i(\mathbf{AA}^\mathrm{T})-\alpha$ is the same as the normalized eigenvector for eigenvalue $\sigma_i(\mathbf{AA}^\mathrm{T})-\alpha$ of $\mathbf{AA}^\mathrm{T}-\alpha\mathbf{I}$. 
Notice that $\sigma_i(\mathbf{AA}^\mathrm{T}-\alpha \mathbf{I})\neq \sigma_{l+1}(\mathbf{AA}^\mathrm{T}-\alpha \mathbf{I})$ ensures that $\sigma_i(\mathbf{AA}^\mathrm{T}-\alpha \mathbf{I})=\sigma_i(\mathbf{AA}^\mathrm{T})-\alpha$ is a value only existing among the largest $l$ singular values  $\sigma_1(\mathbf{AA}^\mathrm{T}-\alpha \mathbf{I}), \cdots, \sigma_l(\mathbf{AA}^\mathrm{T}-\alpha \mathbf{I})$.	
	According to \cref{lemma:2}, the eigenspace of $\mathbf{AA}^\mathrm{T}$ for $\sigma_i(\mathbf{AA}^\mathrm{T})$ is the same as the eigenspace of $\mathbf{AA}^\mathrm{T}-\alpha\mathbf{I}$ for $\sigma_i(\mathbf{AA}^\mathrm{T})-\alpha$ for any $i\le l$.
 Therefore, combining (\ref{p1:6})
 yields that for any $i\le l$, the left singular vector corresponding to the $i$-th largest singular value $\sigma_i(\mathbf{AA}^\mathrm{T})$ of $\mathbf{AA}^\mathrm{T}$ 
 is the left singular vector corresponding to the $i$-th largest singular value $\sigma_i(\mathbf{AA}^\mathrm{T})-\alpha$ of $\mathbf{AA}^\mathrm{T}-\alpha\mathbf{I}$ , and vice versa. 
\end{proof}

\cref{proposition:1} shows that, if we choose a shift  $0<\alpha \!\le\! \sigma_l(\mathbf{A}\mathbf{A}^\mathrm{T})/2$, we can change the computation $\mathbf{Q}=\mathbf{AA}^\mathrm{T}\mathbf{Q}$ to $\mathbf{Q}=(\mathbf{AA}^\mathrm{T}-\alpha\mathbf{I})\mathbf{Q}$ in the power iteration, with the same approximated dominant subspace. We called this \emph{shifted power iteration}. For each step of \atn{the} shifted power iteration, this makes $\mathbf{Q}$ approximate the basis of dominant subspace of $range(\mathbf{A})$ to a larger extent than executing an original power iteration step, because the singular values of $\mathbf{AA}^\mathrm{T}\!-\!\alpha\mathbf{I}$ decay faster than those of $\mathbf{AA}^\mathrm{T}$.
Therefore, the shifted power iteration can improve the accuracy of the randomized SVD algorithm with same power parameter $p$. Then, how to set the shift $\alpha$ properly is the remaining problem.  

Consider the change of ratio of singular values from $\frac{\sigma_i(\mathbf{AA}^\mathrm{T})}{ \sigma_j(\mathbf{AA}^\mathrm{T})}$ to $\frac{\sigma_i(\mathbf{AA}^\mathrm{T}\!-\!\alpha\mathbf{I}) }{ \sigma_j(\mathbf{AA}^\mathrm{T}\!-\!\alpha\mathbf{I})}$, for $j<i\le l$ and $\sigma_j(\mathbf{AA}^\mathrm{T})>\sigma_i(\mathbf{AA}^\mathrm{T})$. It is easy to see $\frac{\sigma_i(\mathbf{AA}^\mathrm{T}\!-\!\alpha\mathbf{I}) }{ \sigma_j(\mathbf{AA}^\mathrm{T}\!-\!\alpha\mathbf{I})}\!<\!\frac{\sigma_i(\mathbf{AA}^\mathrm{T})}{ \sigma_j(\mathbf{AA}^\mathrm{T})}$
if the assumption of $\alpha$ in \cref{proposition:1} holds. 
And, the larger value of $\alpha$, the smaller the ratio $\frac{\sigma_i(\mathbf{AA}^\mathrm{T}\!-\!\alpha\mathbf{I}) }{ \sigma_j(\mathbf{AA}^\mathrm{T}\!-\!\alpha\mathbf{I})}$, reflecting faster decay of singular \atn{values}.  Therefore, to maximize the effect of \atn{the} shifted power iteration on improving the accuracy, we should choose the shift $\alpha$ as large as possible  while satisfying $\alpha\le\sigma_l(\mathbf{AA}^\mathrm{T})/2$. Notice that  calculating $\sigma_l(\mathbf{A}\mathbf{A}^\mathrm{T})$ directly is very difficult. Our idea is to use the singular value of $\mathbf{A}\mathbf{A}^\mathrm{T}\mathbf{Q}$ 
in the power iteration to approximate $\sigma_l(\mathbf{A}\mathbf{A}^\mathrm{T})$ and set the shift $\alpha$. Suppose $\mathbf{Q}\in\mathbb{R}^{m\times l}$ is the orthonormal matrix in \atn{the} power iteration of Alg.~1.
According to \cref{proposition:2},
\begin{equation}
	\label{proposition2:2}
	\sigma_{i}(\mathbf{A}\mathbf{A}^\mathrm{T}\mathbf{Q}) \le \sigma_i(\mathbf{A}\mathbf{A}^\mathrm{T})~,~\textrm{for any}~i\le l,
\end{equation}
which means that we can set $\alpha=\sigma_l(\mathbf{A}\mathbf{A}^\mathrm{T}\mathbf{Q})/2$ to guarantee the requirement of $\alpha$ in \cref{proposition:1} for performing the shifted power iteration. In order to do the orthonormalization for alleviating round-off error and calculate $\sigma_l(\mathbf{A}\mathbf{A}^\mathrm{T}\mathbf{Q})$, we implement ``orth($\cdot$)'' with the economic SVD. This has similar computational cost as using QR factorization, and the resulted matrix of left singular vectors includes the orthonormal basis of same subspace.

So far, we can compute the value of $\alpha$ at the first step of \atn{the} power iteration, and then we perform $\mathbf{Q}\!=\!(\mathbf{AA}^\mathrm{T}\!-\!\alpha\mathbf{I})\mathbf{Q}$ in the following iteration steps. Notice computing the singular values of $\!(\mathbf{AA}^\mathrm{T}\!-\!\alpha\mathbf{I})\mathbf{Q}$ is convenient. Then, according \atn{to \cref{proposition:1,proposition:2}}, we can derive with $0<\alpha\le\sigma_l(\mathbf{AA}^\mathrm{T})/2$
\begin{equation}
		\label{p3:1}
		\begin{aligned}
			\sigma_i((\mathbf{A}\mathbf{A}^\mathrm{T}\!-\alpha\mathbf{I})\mathbf{Q}) + \alpha \le  \sigma_i(\mathbf{A}\mathbf{A}^\mathrm{T}\!-\alpha\mathbf{I}) + \alpha
			= \sigma_i(\mathbf{A}\mathbf{A}^\mathrm{T})~,~\textrm{for any}~i\le l, 
		\end{aligned}
\end{equation}
which states how to use the singular values of $\!(\mathbf{AA}^\mathrm{T}\!-\!\alpha\mathbf{I})\mathbf{Q}$ to approximate $\sigma_i(\mathbf{A}\mathbf{A}^\mathrm{T})$. Therefore, in each iteration step we can obtain a valid value of shift and update $\alpha$ with it if we have a larger $\alpha$, for faster decay of singular values.

\begin{algorithm}[!b]
	\caption{Randomized SVD with dynamic shifts}
	\label{alg2}
	\begin{algorithmic}[1]
		\REQUIRE $\mathbf{A}\in\mathbb{R}^{m\times n}$, parameters $k$, $s$, $p$
		\ENSURE $\mathbf{U}\in\mathbb{R}^{m\times k}$, $\mathbf{S}\in\mathbb{R}^{k\times k}$, $\mathbf{V}\in\mathbb{R}^{n\times k}$
		\STATE $l=k+s$, $\mathbf{\Omega} = \mathrm{randn}(n, l)$
		\STATE $\mathbf{Q} = \mathrm{orth}(\mathbf{A}\mathbf{\Omega})$,  $\alpha=0$
		\FOR {$j=1, 2, \cdots, p$}
		\STATE  $[\mathbf{Q}, \mathbf{\hat{S}}, \sim] = \mathrm{svd}(\mathbf{A}\mathbf{A}^\mathrm{T}\mathbf{Q}-\alpha\mathbf{Q}, \mathrm{'econ'})$
		\STATE  \textbf{if} $\mathbf{\hat{S}}(l,l) > \alpha$ \textbf{then} $\alpha = \frac{\mathbf{\hat{S}}(l,l)+\alpha}{2}$
		\ENDFOR
		\STATE  $\mathbf{B}  = \mathbf{Q}^{\mathrm{T}}\mathbf{A}$
		\STATE  $[\mathbf{U}, \mathbf{S}, \mathbf{V}] = \mathrm{svd}(\mathbf{B}, \mathrm{'econ'})$
		\STATE  $\mathbf{U} = \mathbf{Q}\mathbf{U}(:, 1\!:\!k), \mathbf{S} \!=\!\mathbf{S}(1\!:\!k, 1\!:\!k), \mathbf{V}\! =\! \mathbf{V}(:, 1\!:\!k)$
	\end{algorithmic}
\end{algorithm}

Combining the shifted power iteration with the dynamic updating scheme of $\alpha$, we derive a randomized SVD algorithm based on \atn{the} dynamically shifted power iteration as Algorithm~2.
In Step 5, it checks if $\frac{\sigma_l((\mathbf{A}\mathbf{A}^\mathrm{T}\!-\alpha\mathbf{I})\mathbf{Q}) + \alpha}{2}$ is larger than $\alpha$.
For same setting of power parameter $p$, Alg. 2 has similar computational time as Alg. 1 (with just QR for a tall-and-skinny matrix
replaced with SVD), but produces results with much better accuracy. 
For matrix $\mathbf{Q}$ computed with Alg. 2, we derive a bound of $\Vert\mathbf{Q}\mathbf{Q}^\mathrm{T}\mathbf{A}-\mathbf{A}\Vert$, which reflects how close the computed truncated SVD is to optimal.
\begin{theorem}
	\label{theorem:1}
	Suppose $\mathbf{A}\in\mathbb{R}^{m\times n}~(m\ge n)$, and $k$, $s$ and $p$ are the parameters in Alg.~2. $l=k+s\le n-k$, and $\mathbf{Q}$ in size $m\times l$ is the obtained orthonormal matrix with Alg.~2. If positive integer $j<k$, and real numbers $\beta$, $\gamma>1$ satisfying $\phi\!=\!$ $\frac{1}{\sqrt{2\pi(l\!-\!j\!+\!1)}}(\frac{e}{(l\!-\!j\!+\!1)\beta})^{l\!-\!j\!+\!1} \!+\!\frac{1}{4\gamma(\gamma^2\!-\!1)} (\frac{1}{\sqrt{\pi(n\!-\!k)}}(\frac{2\gamma^2}{e^{\gamma^2\!-\!1}})^{n\!-\!k}\! +\!\frac{1}{\sqrt{\pi l}}(\frac{2\gamma^2}{e^{\gamma^2\!-\!1}})^{l}) \!\le \! 1$, then
	\begin{equation}
		\label{theorem1:1}
		\begin{aligned}
			\Vert\mathbf{Q}\mathbf{Q}^\mathrm{T}\mathbf{A}\!-\!\mathbf{A}\Vert \le & 2\sqrt{\left(2l^2\beta^2\gamma^2\prod_{c=1}^{p}\left(\frac{\sigma_{j+1}^2\!-\!\alpha_c}{\sigma_j^2-\alpha_c}\right)^{2}\!+\!1\right)
			\sigma_{j+1}^2 +\left(2l(n\!-\!k)\beta^2\gamma^2\prod_{c=1}^{p}\left(\frac{\sigma_{k+1}^2\!-\!\alpha_c}{\sigma_{j}^2-\alpha_c}\right)^{2}\!+\!1\right)\sigma_{k+1}^2},
		\end{aligned}
	\end{equation}
	with probability not less than $1-\phi$, where $\alpha_c$ denotes the shift value in the $c$-th shifted power iteration in Alg.~2 and $\sigma_j$ denotes the $j$-th \atn{largest} singular value of $\mathbf{A}$.
\end{theorem}
\begin{remark}
	$\phi$ can be a small value, and the bound (\ref{theorem1:1}) is smaller than that derived in \cite{rokhlin2010randomized} for the $\mathbf{Q}$ computed with the basic randomized SVD algorithm (see Appendix~\ref{secA1}).
\end{remark}

Alg. 2 can collaborate with more skills to achieve further acceleration while handling sparse matrices. Firstly, we can use EVD to compute the economic SVD or the orthonormal basis of a tall-and-skinny matrix in Step 2, 4 and 8. This resorts to the eigSVD algorithm in \cite{pmlr-v95-feng18a}, which is described in Alg.~3. \atnn{The algorithm} is similar to that in \cite{pmlr-v95-feng18a}, but outputs singular values in the correct order. ``eig($\cdot$)'' computes the eigenvalue decomposition. The ``rot90($\mathbf{S}$, 2)'' in Step 4 is the function to rotate $\mathbf{S}$ 180 degrees, and ``$\mathrm{fliplr}(\cdot)$'' in Step 4 is the function to flip array or matrix \atnn{left} to right. Therefore, the singular values are in descending order, while the singular vectors corresponding to them are \atnn{at} correct positions. Secondly, for $\mathbf{A}\in\mathbb{R}^{m\times n}$ with $m>n$, it is more efficient to computing SVD of $\mathbf{A}^\mathrm{T}$ with Alg. 2. Therefore, we can derive two efficient versions of \atnn{randomized SVD} algorithm, for the cases with $m\ge n$ and $n\ge m$ respectively. The \atnn{version} for the matrices with $m\ge n$ is described in Algorithm 4, and all the previous theoretic analysis can be applied\atnn{, if we replace} $\mathbf{A}$ with $\mathbf{A}^\mathrm{T}$.

\begin{algorithm}[h]
	\caption{eigSVD}
	\label{alg5}
	\begin{algorithmic}[1]
		\REQUIRE a dense matrix $\mathbf{C}\in\mathbb{R}^{m\times n}$ ($m\ge n$)
		\ENSURE $\mathbf{U}\in\mathbb{R}^{m\times n}$, ~ $\mathbf{S}\in\mathbb{R}^{n\times n}$, ~ $\mathbf{V}\in\mathbb{R}^{n\times n}$
		\STATE \atn{$[\mathbf{V}, \mathbf{D}] = \mathrm{eig}(\mathbf{C}^{\mathrm{T}}\mathbf{C})$}
		\STATE $\mathbf{S} = \mathrm{sqrt}(\mathbf{D})$
		\STATE $\mathbf{U} = \mathbf{C}\mathbf{V}\mathbf{S}^{-1}$
		\STATE $\mathbf{U} = \mathrm{fliplr}(\mathbf{U})$, ~ $\mathbf{V} = \mathrm{fliplr}(\mathbf{V})$, ~ $\mathbf{S} = \mathrm{rot90}(\mathbf{S},2)$
	\end{algorithmic}
\end{algorithm}

\begin{algorithm}[h]
		\caption{Dynamic shifts based randomized SVD for matrices with $m\ge n$}
		\label{alg3}
		\begin{algorithmic}[1]
				\REQUIRE $\mathbf{A}\in\mathbb{R}^{m\times n}$, parameters $k$, $s$, $p$
				\ENSURE $\mathbf{U}\in\mathbb{R}^{m\times k}$, $\mathbf{S}\in\mathbb{R}^{k\times k}$, $\mathbf{V}\in\mathbb{R}^{n\times k}$
				\STATE $l=k+s$, $\mathbf{\Omega} = \mathrm{randn}(m, l)$
				\STATE $[\mathbf{Q},\sim, \sim] = \mathrm{eigSVD}(\mathbf{A}^\mathrm{T}\mathbf{\Omega})$,  $\alpha=0$
				\FOR {$j=1, 2, \cdots, p$}
				\STATE $[\mathbf{Q}, \mathbf{\hat{S}}, \sim] = \mathrm{eigSVD}(\mathbf{A}^\mathrm{T}(\mathbf{A}\mathbf{Q})-\alpha\mathbf{Q})$
				\STATE \textbf{if} $\mathbf{\hat{S}}(l, l) > \alpha$ \textbf{then} $\alpha = \frac{\mathbf{\hat{S}}(l, l)+\alpha}{2}$
				\ENDFOR
				\STATE $\mathbf{B}  = \mathbf{AQ}$
				\STATE $[\mathbf{U}, \mathbf{S}, \mathbf{V}] = \mathrm{eigSVD}(\mathbf{B})$
				\STATE $\mathbf{U} = \mathbf{U}(:, 1\!:\!k), \mathbf{S} \!=\!\mathbf{S}(1\!:\!k, 1\!:\!k), \mathbf{V}\! =\! \mathbf{Q}\mathbf{V}(:, 1\!:\!k)$
			\end{algorithmic}
	\end{algorithm}

Suppose $p$ power iterations are executed to attain a certain accuracy\atnn{, and $C_{\mathrm{eig}}$ represents the constant in the flop count of EVD.} The flop count of Alg.~4 \atnn{is}
\begin{equation}
	\begin{aligned}
		\mathrm{FC}_{4}=(2p\!+\!2)C_{\mathrm{mul}}\textrm{nnz}(\mathbf{A})l\!+\!(p\!+\!1)(2C_{\mathrm{mul}}nl^2\!+\!C_{\mathrm{eig}}l^3)\!+2C_{\mathrm{mul}}ml^2\!+\!C_{\mathrm{eig}}l^3\!+\!pC_{\mathrm{mul}}nl\!+\!C_{\mathrm{mul}}nlk,
	\end{aligned}
\end{equation}
where $(2p\!+\!2)C_{\mathrm{mul}}\textrm{nnz}(\mathbf{A})l$ reflects \atnn{the} $2p\!+\!2$ times matrix-matrix multiplications on $\mathbf{A}$, $(p\!+\!1)(2C_{\mathrm{mul}}nl^2\!+\!C_{\mathrm{eig}}l^3)$ reflects \atnn{executing  eigSVD on an $n\times l$ matrix for $p\!+\!1$ times}, $2C_{\mathrm{mul}}ml^2\!+\!C_{\mathrm{eig}}l^3$ reflects once eigSVD on matrix $\mathbf{B}$, $pC_{\mathrm{mul}}nl$ reflects \atnn{executing the operation ``$-\alpha\mathbf{Q}$'' for $p$ times}, and $C_{\mathrm{mul}}nlk$ reflects the matrix-matrix multiplication to generate $\mathbf{V}$. From this, we see that the time complexity of \atnn{Alg. 4} is $O(pl\textrm{nnz}(\mathbf{A})+pnl^2)$ for $m\times n$ matrix $\mathbf{A}$ with $m\ge n$. The space complexity of \atnn{Alg. 4} is about $O((2m+n)l+\textrm{nnz}(\mathbf{A}))$, as the peak memory usage occurs at Step 8. 
Because $C_{qr}$ and $C_{svd}$ are multiple times larger than $C_{mul}$, we can clearly see that $\mathrm{FC}_4<\mathrm{FC}_1$ when $m\ge n$ and $n\gg l$, which reflects the efficiency of \atnn{Alg. 4}.

\subsection{Accuracy Control Based on PVE Bound}\label{sec3.2}
Theoretical research has revealed that the  randomized SVD with power iteration produces the rank-$k$ approximation close to optimal. Under spectral norm (or Frobenius norm) the computational \atnn{results} ($\hat{\mathbf{U}}$, $\hat{\mathbf{\Sigma}}$ and $\hat{\mathbf{V}}$) \atnn{satisfy} the following multiplicative guarantee with high probability:
\begin{equation}
	\label{err:relative}
	\Vert\mathbf{A}- \hat{\mathbf{U}}\hat{\mathbf{\Sigma}}\hat{\mathbf{V}}^\mathrm{T} \Vert \le (1+\epsilon ) \Vert\mathbf{A}- \mathbf{A}_k \Vert,
\end{equation}
\atn{where $\epsilon \in (0, 1)$ is a distortion parameter.}
Another guarantee proposed in \cite{musco2015}, which is more meaningful in machine learning problems, is:
\begin{equation}
	\label{err:pve}
	\forall i\le k, ~ ~ \vert\mathbf{u}_i^\mathrm{T}\mathbf{A}\mathbf{A}^\mathrm{T}\mathbf{u}_i-  \hat{\mathbf{u}}_i^\mathrm{T}\mathbf{A}\mathbf{A}^\mathrm{T}\hat{\mathbf{u}}_i\vert \le \epsilon \sigma_{k+1}(\mathbf{A})^2,
\end{equation}
where $\mathbf{u}_i$ is the $i$-th left singular vector of $\mathbf{A}$, and $\hat{\mathbf{u}}_i$ is the computed $i$-th left singular vector. This is called \emph{per vector error} (PVE)  bound for singular vectors.
In \cite{musco2015}, it is demonstrated that the $(1+\epsilon)$ error bound in (\ref{err:relative}) may not guarantee any accuracy in the computed singular vectors. 
In contrary, the per vector guarantee (\ref{err:pve}) requires each computed singular vector to capture nearly as much variance as the corresponding accurate singular vector, \atnn{thus guaranteeing} the accuracy of the principal components in PCA. 
\atn{Notice that the smaller $\epsilon$ in (17) and (18) is, the more power iterations in the randomized SVD algorithms are needed to meet the corresponding accuracy \atnn{guarantee}.}

In existing work on randomized SVD, the power parameter $p$ (or pass parameter) is assumed known in prior. 
\atnn{For} attaining certain accuracy, \atnn{multiple} examinations may be performed to find \atn{a} suitable \atnn{and not too large $p$. This brings a lot of computation cost.} 
Therefore, how to adaptively and efficiently determine \atnn{how many power iterations should be executed} to ensure certain accuracy of the \atnn{computed SVD} 
is of concern.

There are different criteria to evaluate the accuracy of SVD, like the relative residual in (\ref{err:svds}) and \atnn{that derived form the PVE bound}. Below we find out that if collaborating with the PVE \atnn{criterion} we can develop an efficient randomized SVD algorithm with accuracy control. 
The PVE \atnn{bound} in (\ref{err:pve}) \atnn{can be approximately expressed as}:
\begin{equation}
	\label{err:pve_approx}
	\forall i\le k, ~ ~ \frac{\vert\hat{\sigma_i}^{(j)}(\mathbf{A})^2-  \hat{\sigma_i}^{(j-1)}(\mathbf{A})^2\vert}{\hat{\sigma}_{k+1}^{(j)}(\mathbf{A})^2} \le \mathrm{tol} ,
\end{equation}
where $\hat{\sigma}_i(\mathbf{A})$ denotes the computed $i$-th singular value and superscript $(j)$ denotes the computed value with $p=j$. ``tol'' stands for an error tolerance. \atnn{This formula (\ref{err:pve_approx}) is an approximate PVE criterion for accuracy control, and ``tol'' is supposed to be specified by user.}
Setting ``tol'' suitably often guarantees the \atnn{PVE bound (\ref{err:pve}) with $\epsilon$ comparable to ``tol''}, as shown in our experiments. In contrast, $p$ gives
no insight on the error and varies largely for different cases.
The approximate criterion \atnn{in (\ref{err:pve_approx})} is based on estimating the error of $\hat{\sigma_i}^{(j-1)}(\mathbf{A})^2$ by regarding $\hat{\sigma_i}^{(j)}(\mathbf{A})^2$ as the accurate result.   
This skill is a traditional method to control the accuracy in computing truncated eigenvalue decomposition \citep{matrix2012}, and has been used in 
other numerical methods with iterative steps \citep{heath2018scientific}. Notice that \atnn{constructing an accuracy} criterion based on (\ref{err:svds}) or (\ref{err:relative}) \atnn{leads to difficulties in implementation}, 
\atnn{as it involves} expensive computations for large matrices. 

The above idea of PVE criterion derives a randomized SVD algorithm with accuracy control. A straightforward implementation moves the generation and SVD of matrix $\mathbf{B}=\mathbf{Q}^\mathrm{T}\mathbf{A}$ into the power iteration, and then imposes the check of termination criterion. 
Below we develop a more efficient way to realize the accuracy control of PVE criterion. 
\begin{proposition}
	\label{theorem:2}
	Suppose matrix $\mathbf{A}\!\in\!\mathbb{R}^{m\times n}$, orthonormal matrix $\mathbf{Q}\in\mathbb{R}^{{m\times l}}$ and positive value $\alpha$ are the quantities before executing Step 4 in Alg.~2. Then, for any $i\le l$\,,
	\begin{equation}
		\label{theorem2:1}
		\sigma_i(\mathbf{Q}^\mathrm{T}\mathbf{A})\le \sqrt{\sigma_i(\mathbf{AA}^\mathrm{T}\mathbf{Q}-\alpha\mathbf{Q})+\alpha}\le\sigma_i(\mathbf{A}).
	\end{equation}
\end{proposition}

\begin{proof}
	For any $i\le l$, according to (\ref{lemma3:1}) in \cref{lemma:3}, 
	\begin{equation}\label{p4:1}
		\begin{aligned}
			\sigma_i(\mathbf{Q}^\mathrm{T}\mathbf{A})^2 =\sigma_i(\mathbf{Q}^\mathrm{T}\mathbf{A}\mathbf{A}^\mathrm{T}\mathbf{Q}) &\le \sigma_i(\mathbf{Q}^\mathrm{T}\mathbf{A}\mathbf{A}^\mathrm{T})\sigma_1(\mathbf{Q})=\sigma_i(\mathbf{AA}^\mathrm{T}\mathbf{Q}).\\
		\end{aligned}
	\end{equation}
	According to (\ref{lemma3:2}) in \cref{lemma:3}, we can derive
	\begin{equation}\label{p4:2}
		\begin{aligned}
			\sigma_i(\mathbf{A}\mathbf{A}^\mathrm{T}\mathbf{Q})&=\sigma_i(\mathbf{A}\mathbf{A}^\mathrm{T}\mathbf{Q}-\alpha\mathbf{Q}+\alpha\mathbf{Q})\le \sigma_i(\mathbf{AA}^\mathrm{T}\mathbf{Q}-\alpha\mathbf{Q}) + \alpha,
		\end{aligned}
	\end{equation}
	and according to (\ref{p3:1}) we derive
	\begin{equation}\label{p4:3}
		\begin{aligned}
			\sqrt{\sigma_i(\mathbf{AA}^\mathrm{T}\mathbf{Q}-\alpha\mathbf{Q}) + \alpha}
			\!\le\! \sqrt{\sigma_i(\mathbf{AA}^\mathrm{T})} \!=\! \sigma_i(\mathbf{A}).
		\end{aligned}  
	\end{equation}
	Therefore, combining (\ref{p4:1})$\sim$(\ref{p4:3}) yields \cref{theorem:2}. 
\end{proof}

According to \cref{theorem:2}, $\sqrt{\sigma_i(\mathbf{AA}^\mathrm{T}\mathbf{Q}-\alpha\mathbf{Q})+\alpha}$ is closer to $\sigma_i(\mathbf{A})$ than $\sigma_i(\mathbf{Q}^\mathrm{T}\mathbf{A})$ which is the computed singular value with the power iteration with one fewer step. Therefore, $\sigma_i(\mathbf{AA}^\mathrm{T}\mathbf{Q}-\alpha\mathbf{Q})+\alpha$ can be regarded as a convenient surrogate of $\hat{\sigma_i}(\mathbf{A})^2$, as the former is readily available after Step 4 of Alg.~2. Substituting it to (\ref{err:pve_approx}) we can derive \atnn{the dashSVD algorithm, whose version} for $m\ge n$ is described as Algorithm 5.
\begin{algorithm}[!b]
	\caption{The dashSVD with PVE accuracy control
	}
	\label{alg4}
	\begin{algorithmic}[1]
		\REQUIRE $\mathbf{A}\in\mathbb{R}^{m\times n}$ ($m\ge n$), parameters $k$, $s$, $p_{max}$, an error tolerance tol
		\ENSURE $\mathbf{U}\in\mathbb{R}^{m\times k}$, $\mathbf{S}\in\mathbb{R}^{k\times k}$, $\mathbf{V}\in\mathbb{R}^{n\times k}$
		\STATE $l=k+s$, ~ $\mathbf{\Omega} = \mathrm{randn}(m, l)$
		\STATE $[\mathbf{Q},\sim,\sim] \!=\!\mathrm{eigSVD}(\mathbf{A}^\mathrm{T}\mathbf{\Omega})$, $\alpha\!=\!\alpha^\prime\!=\!0$,  ~ $\hat{\mathbf{S}}^{\prime}\!=\!\mathrm{zeros}(l,1)$
		\FOR {$j=1, 2, 3, \cdots, p_{max}$}
		\STATE $[\mathbf{Q}, \mathbf{\hat{S}},\sim] = \mathrm{eigSVD}(\mathbf{A}^\mathrm{T}(\mathbf{A}\mathbf{Q})-\alpha\mathbf{Q})$
		\STATE \textbf{if} $\forall i\le k,~ \frac{\vert\hat{\mathbf{S}}^\prime(i,i)+\alpha^\prime-\hat{\mathbf{S}}(i,i)-\alpha\vert}{\hat{\mathbf{S}}(k+1,k+1)+\alpha}\le \mathrm{tol}$ \textbf{then break}
		\STATE \textbf{if} $\mathbf{\hat{S}}(l, l)>\alpha$ \textbf{then} $\alpha = \frac{\mathbf{\hat{S}}(l, l)+\alpha}{2}$
		\STATE $\hat{\mathbf{S}}^{\prime}=\hat{\mathbf{S}}, ~ \alpha^\prime=\alpha$
		\ENDFOR
		\STATE $\mathbf{B}=\mathbf{AQ}$
		\STATE $[\mathbf{U}, \mathbf{S}, \mathbf{V}] = \mathrm{eigSVD}(\mathbf{B})$
		\STATE $\mathbf{U} = \mathbf{U}(:, 1:k)$, ~ $\mathbf{S} = \mathbf{S}(1:k, 1:k)$, ~ $\mathbf{V} = \mathbf{Q}\mathbf{V}(:,1:k)$
	\end{algorithmic}
\end{algorithm}	
Here, Step 5 for accuracy control costs negligible time, $p_{max}$ is 
 a parameter denoting the upper bound of \atnn{the number of power iterations (i.e. parameter $p$)}, and \atnn{the symbols with ``$\prime$'' are} the quantities in last iteration step. \atnn{Notice that, Alg. 4 is also a version of dashSVD algorithm; it just does not include the accuracy control. The analysis of computational complexity at the end of last subsection is also valid for the dashSVD algorithm (Alg. 5).}

\section{Performance Analysis}\label{sec4}
The developed algorithms are implemented in Matlab, and also in C with MKL~\citep{Intel} and OpenMP~\citep{openMP} directives for multi-thread parallel computing. \notice{Some of the codes follow the implementation in \cite{martinsson2016randomized2} \cite{randqb-code}.}
Numerical experiments are conducted to compare them with existing truncated SVD algorithms, frPCA in \cite{pmlr-v95-feng18a} written in C with MKL \cite{frPCA-code}, the \texttt{LanczosBD} in \texttt{svds}, and the PRIMME\_SVDS written in C with MKL \atn{which also invokes PETSc~\citep{petsc-web-page}} and \atn{MPI} for parallel computing \cite{primme-code}.
We first validate the performance of the shifted power iteration. Then, we compare the dynamic shifts based SVD algorithm (Alg.~4) with state-of-the-art algorithms. Finally, we validate the proposed \atnn{accuracy-control mechanism in dashSVD}. The accuracy and efficiency of these algorithms are evaluated. 
Four error metrics for evaluating accuracy are considered.
Based on the PVE bound (\ref{err:pve}), the first error is 
$    \epsilon_{\textrm{PVE}} = \max_{i\le k}\frac{\vert\mathbf{u}_i^\mathrm{T}\mathbf{A}\mathbf{A}^\mathrm{T}\mathbf{u}_i-  \hat{\mathbf{u}}_i^\mathrm{T}\mathbf{A}\mathbf{A}^\mathrm{T}\hat{\mathbf{u}}_i\vert}{\sigma_{k+1}(\mathbf{A})^2}$. 
It is called ``rayleigh(last)'' metric in \cite{LazySVD}. Based on relative residual error in (\ref{err:svds}), we define $\epsilon_{\textrm{res}} = \max_{i\le k} \frac{\Vert\mathbf{A}^\mathrm{T}\mathbf{\hat{u}}_i-\hat{\sigma}_i\mathbf{\hat{v}_i}\Vert_2}{\sigma_i(\mathbf{A})}$ which measures the accuracy of singular triplets. Notice that when the computed results are the singular triplets other than the leading singular triplets, $\epsilon_{\textrm{res}}$ can also be very small. This implies $\epsilon_{\textrm{res}}$ may not be a good criterion. \texttt{svds} avoids this issue of $\epsilon_{\textrm{res}}$ by executing extra invocations of \texttt{LanczosBD} for $k+1$ singular triplets.
The third metric is $\epsilon_\textrm{spec} = \frac{\Vert\mathbf{A}-\mathbf{\hat{U}\hat{\Sigma}\hat{V}}^\mathrm{T}\Vert_2-\Vert\mathbf{A}-\mathbf{A}_k\Vert_2}{\Vert\mathbf{A}-\mathbf{A}_k\Vert_2}$, which represents the spectral relative error based on (\ref{err:relative}). Lastly, we use $\epsilon_{\sigma} = \max_{i\le k} \frac{\vert\sigma_i(\mathbf{A})-\hat{\sigma}_i\vert}{\sigma_i{(\mathbf{A}})}$ to measure the accuracy of computed singular values. In these metrics, the accurate truncated SVD results are computed with \texttt{svds} in Matlab 2020b. Notice there is a little run-to-run variability for randomized SVD algorithms, but that is not a serious issue \cite{mahoney2011}. In Appendix B, we show \atnn{the experimental} result \atnn{demonstrating} this variability increases with \atnn{the number of power iterations in dashSVD (parameter $p$)}. While for small $p$ corresponding to computing low-accuracy SVD, the standard deviation of error metric is mostly no more than $5\%$ of \atnn{the corresponding} mean value. Therefore, except for that experiment we present the results of dashSVD in an arbitrary run.

The following sparse matrices from real applications are tested, which are listed in \cref{tab:table0}. 
Three of them were used in \cite{pmlr-v95-feng18a}: a social network matrix from SNAP~\citep{snapnets}, a matrix from MovieLens dataset \citep{movielens} and a person-keyword matrix from the information retrieval application “AMiner”~\citep{Aminer}. 
Rucci1 is a sparse matrix with singular values decaying very slowly, obtained from SuiteSparse matrix collection~\citep{davis2011university}.
The last two sparse matrices uk-2005 and sk-2005 are from larger web graphs~\citep{BoVWFI,BRSLLP}, also downloaded from SuiteSparse matrix collection. The average number of nonzeros per row ranges from 3.9 to 237.

\begin{table}[b]
     \centering
        \setlength{\abovecaptionskip}{0.05 cm}
        \caption{Test Cases.}
        \label{tab:table0}
        \centering
        {
            \begin{spacing}{0.9}
                \renewcommand{\multirowsetup}{\centering}
                \begin{tabular}{ccc} 
                    \toprule
                    Matrix & Dimension & \#Nonzero elements\\
                    \midrule
                    SNAP \atn{(soc-Slashdot0922)} & 82,168$\times$82,168 & 948,464\\
                    MovieLens & 270,896$\times$45,115 & 26,024,289\\
                    Rucci1 & 1,977,885$\times$109,900 & 7,791,168\\
                    Aminer & 12,869,521$\times$323,899 & 206,501,139\\
                    uk-2005 & 39,459,925$\times$39,459,925 & 936,364,282\\
                    sk-2005 & 50,639,401$\times$50,639,401 & 1,949,412,601\\
                    \bottomrule 
                \end{tabular}
            \end{spacing}
        }
\end{table}

In our experiments, the rank parameter $k$ is set to 100 by default. For other values of $k$, the effectiveness of the proposed dashSVD is also \atnn{validated}. The relevant experimental results are shown in Appendix B, which even imply that with larger $k$ the advantage of dashSVD would be slightly more \atnn{prominent}. The oversampling parameter $s$ is fixed to $k/2$.  For fair comparison with respect to memory cost, the subspace dimensionalities in \texttt{LanczosBD} and PRIMME\_SVDS are fixed to $1.5k$.
All experiments are carried out on a computer with two 8-core Intel \atnn{Xeon E5-2620 v4 CPUs}  (at 2.10 GHz) and 512 GB RAM. We compile \atnn{these} algorithms with \atnn{GCC version 9.4.0 and the ``-O3'' option}. \atnn{The  MKL is of version 2019.4.0, while PETSc is of version 3.17.1}. 

\subsection{Validation of the Shifted Power Iteration}

\atnn{In} order to validate the power iteration scheme with dynamic shifts, we set different power parameter $p$ and perform the basic randomized SVD (Alg. 1) and the proposed dynamic shifts based randomized SVD (Alg. 2) on test matrices. With the computational results, the corresponding error metric $\epsilon_{\textrm{PVE}}$ is calculated to evaluate the accuracy. The curves of  $\epsilon_{\textrm{PVE}}$ are drawn in Fig. \ref{fig:1}, where Alg. 2$^*$ denotes 
\begin{figure}[h]
	\setlength{\abovecaptionskip}{0 cm}
	\setlength{\belowcaptionskip}{0 cm}
	\centering
	\subfigure[Dense1] { \label{fig1:1} 
		\includegraphics[width=3.1cm, trim=103 265 115 273,clip]{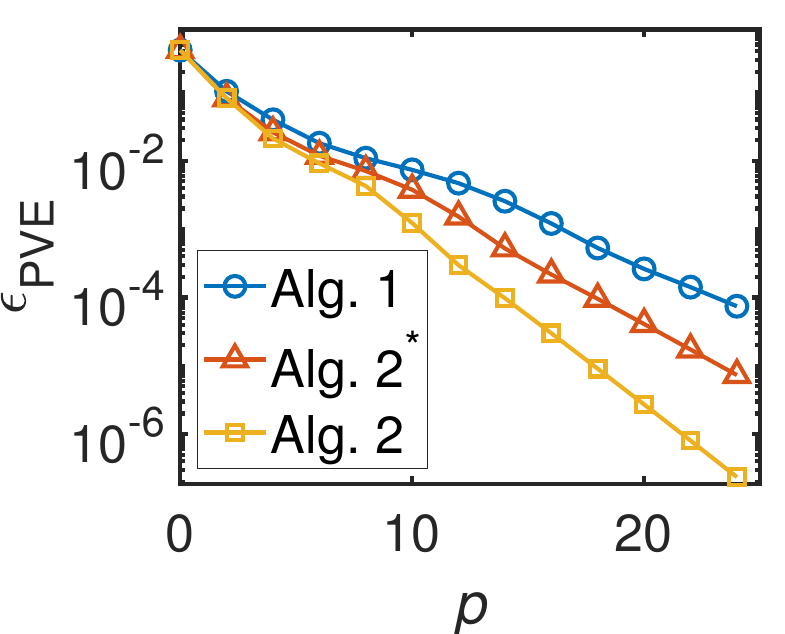} 
	}
	\subfigure[Dense2] { \label{fig1:2} 
		\includegraphics[width=3.1cm,trim=103 265 115 273,clip]{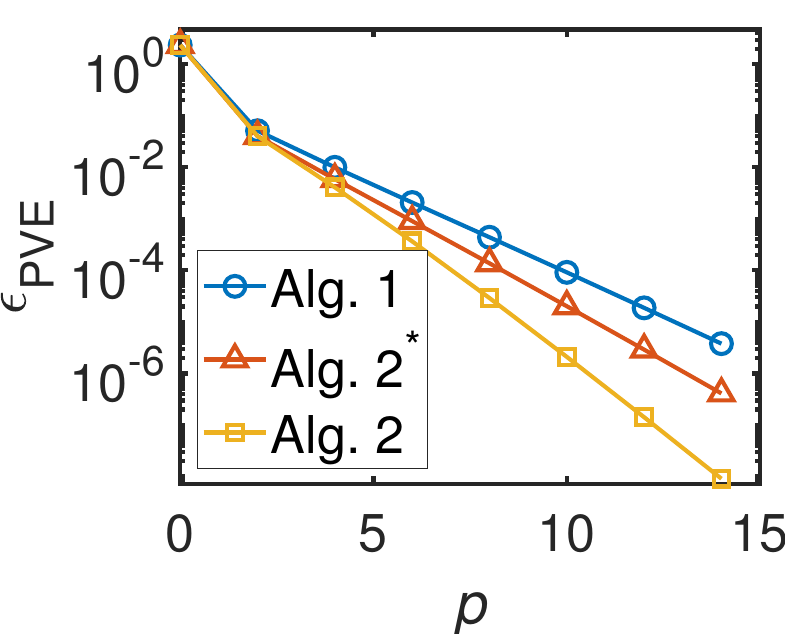} 
	}
	\subfigure[SNAP] { \label{fig1:3} 
		\includegraphics[width=3.1cm,trim=103 265 115 273,clip]{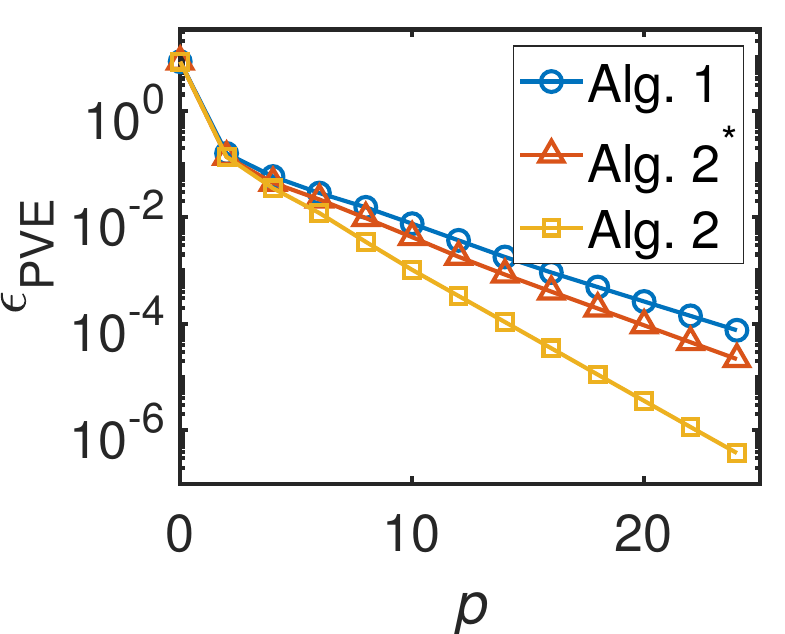} 
	} 
	\subfigure[MovieLens] { \label{fig1:4} 
		\includegraphics[width=3.1cm,trim=103 265 115 273,clip]{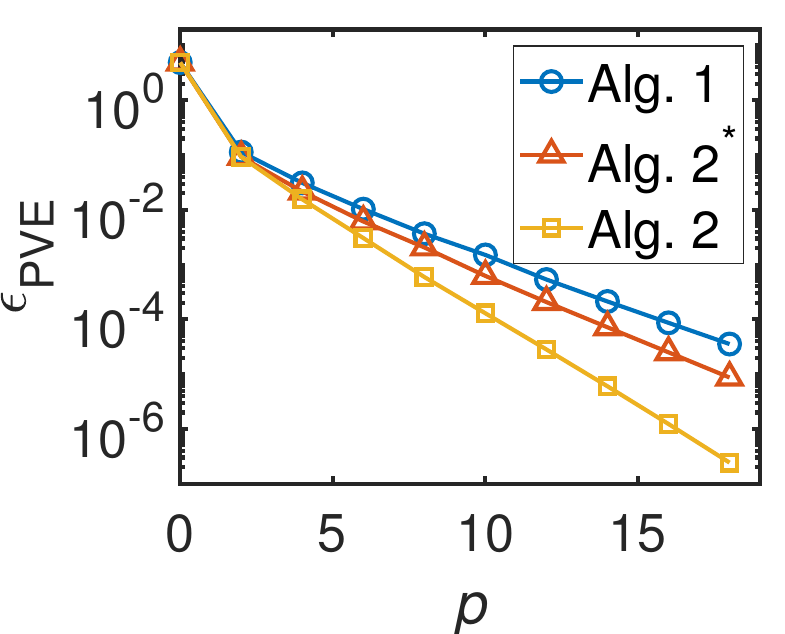} 
	}\\[-1ex]
	\subfigure[Rucci1] { \label{fig1:5} 
		\includegraphics[width=3.1cm,trim=103 265 115 273,clip]{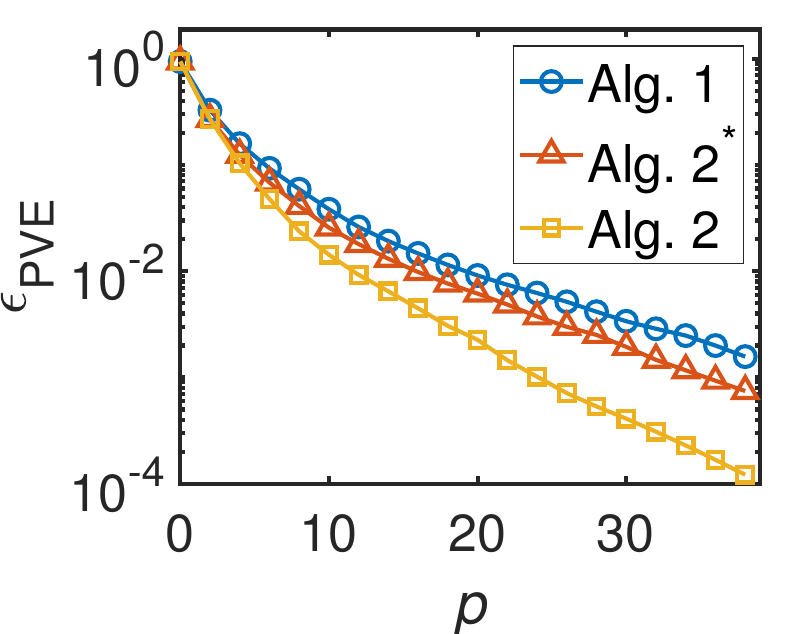} 
	}
	\subfigure[Aminer] { \label{fig1:6} 
		\includegraphics[width=3.1cm,trim=103 265 115 273,clip]{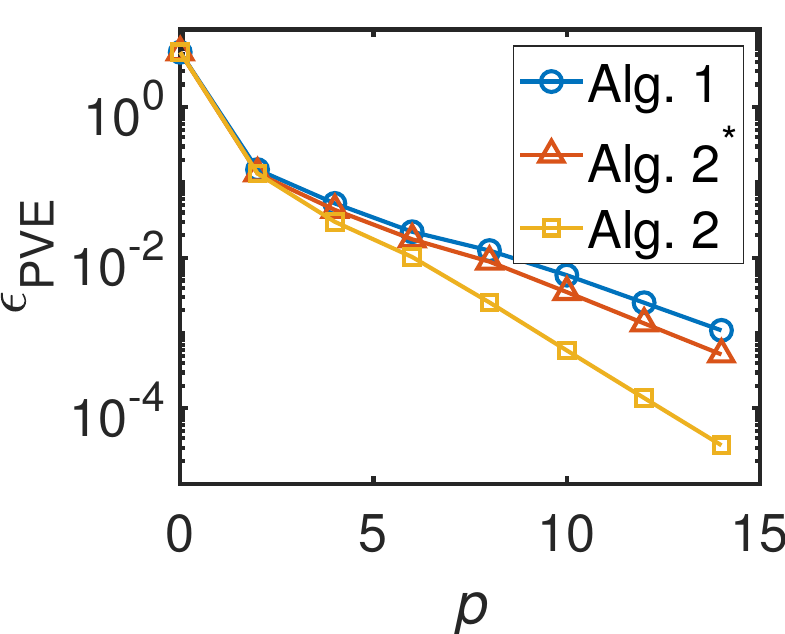} 
	}
	\subfigure[uk-2005] { \label{fig1:7} 
		\includegraphics[width=3.1cm,trim=103 265 115 273,clip]{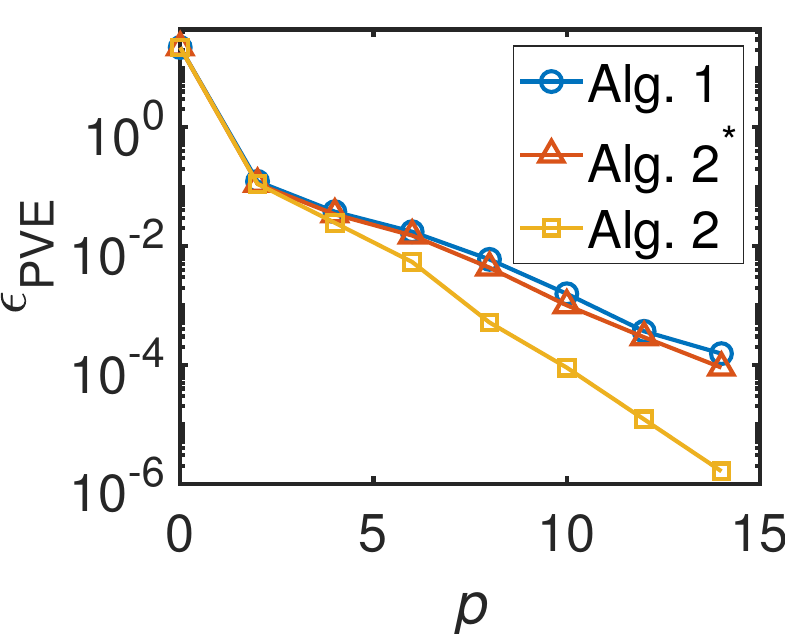} 
	}
	\subfigure[sk-2005] { \label{fig1:8} 
		\includegraphics[width=3.1cm,trim=103 265 115 273,clip]{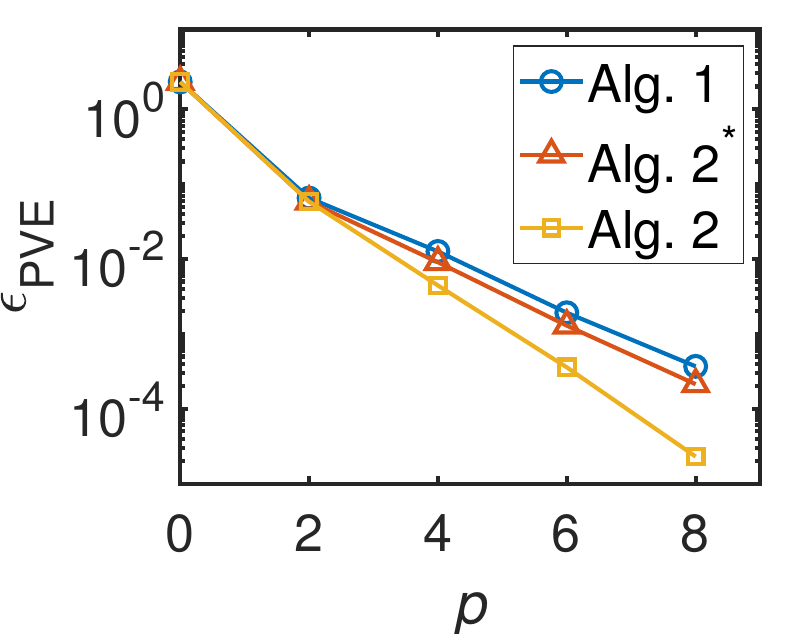} 
	}
	\caption{\notice{The PVE error vs. power parameter curves of three randomized SVD algorithms ($k=100$).}}
	\label{fig:1} 
	\centering
\end{figure} the algorithm with the fixed shift value $\alpha= \sigma_l(\mathbf{A}\mathbf{A}^\mathrm{T}\mathbf{Q})/2$ set at the first step of iteration. Two additional $1,000\times 1,000$ random matrices generated by $\texttt{randn}$ in Matlab (“Dense1”) and with the $i$-th singular value following $\sigma_i = 1/\sqrt{i}$ (“Dense2”) are also tested. 
From Fig.~\ref{fig:1} we see that, 
the dynamic scheme for setting the shift consistently produces more accurate results than the scheme with the fixed shift. And, the proposed \atnn{Alg. 2} results in much better accuracy (and efficiency as well) than the basic randomized SVD with power iteration. The reduction of \atn{the} result's error increases with $p$, and is up to \textbf{100X}.

\atnn{The effectiveness of the shifted power iteration can also be validated by comparing dashSVD with frPCA \cite{pmlr-v95-feng18a}. The major difference between them is that dashSVD employs the shifted power iteration. The comparison results are given in Appendix B.}

\subsection{Validation of the dashSVD Algorithm}\label{sec4.1}

\begin{figure}[b]
	\setlength{\abovecaptionskip}{0 cm}
	\setlength{\belowcaptionskip}{0 cm}
	\centering
	\subfigure[SNAP \atn{($\Delta p$ for dashSVD is 4)}] {
		\begin{minipage}{14cm}
			\centering
			\includegraphics[width=3.4cm, trim=103 265 115 273,clip]{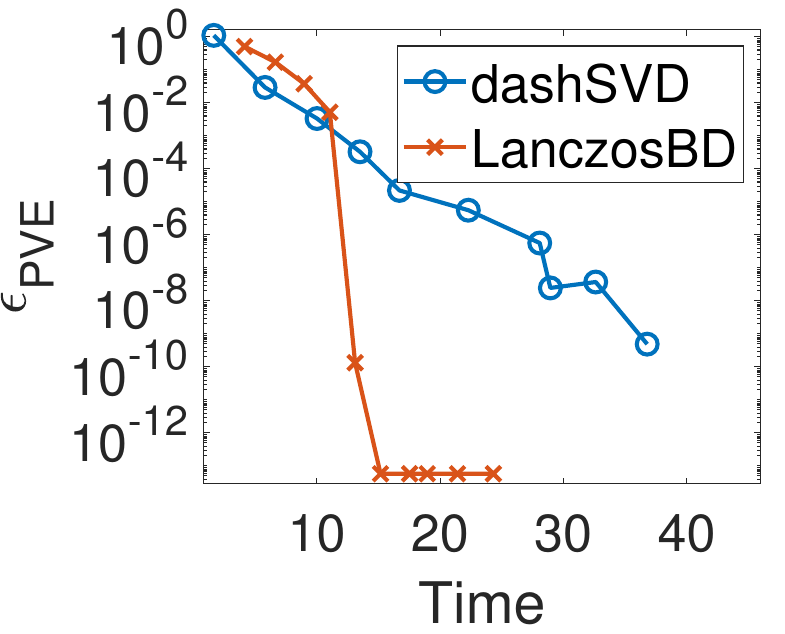}
			\includegraphics[width=3.4cm, trim=103 265 115 273,clip]{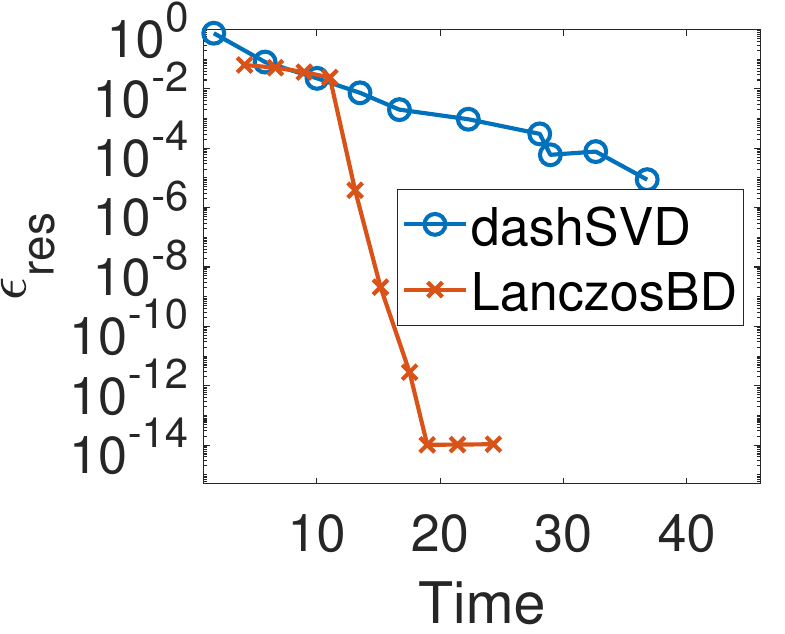}
			\includegraphics[width=3.4cm, trim=103 265 115 273,clip]{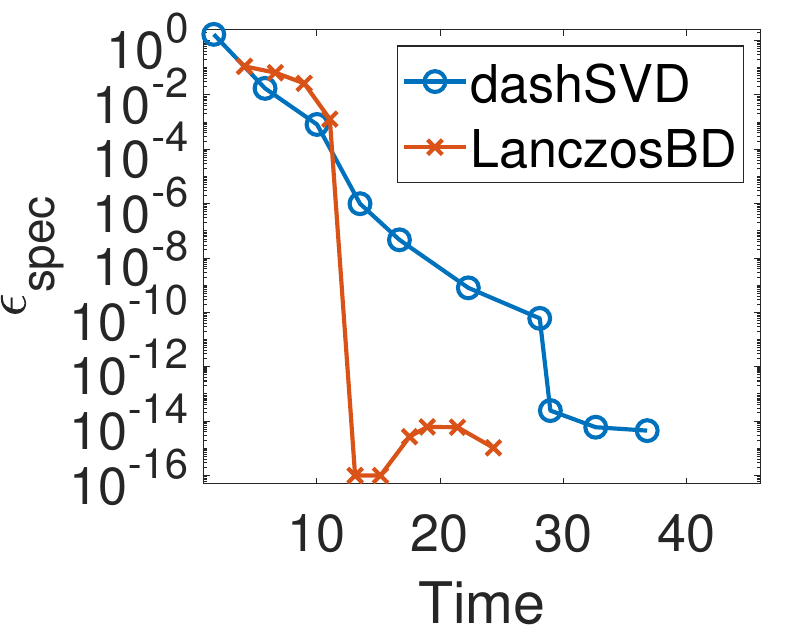}
			\includegraphics[width=3.4cm, trim=103 265 115 273,clip]{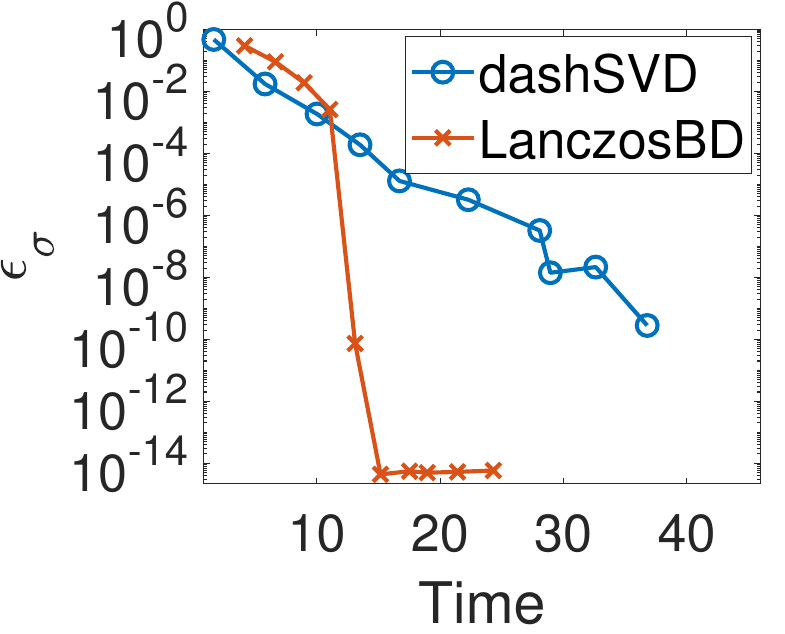}
		\end{minipage}
	}\\[-1ex]
	\subfigure[uk-2005 \atn{($\Delta p$ for dashSVD is 2)}] {
		\begin{minipage}{14cm}
			\centering
			\includegraphics[width=3.4cm, trim=103 265 115 273,clip]{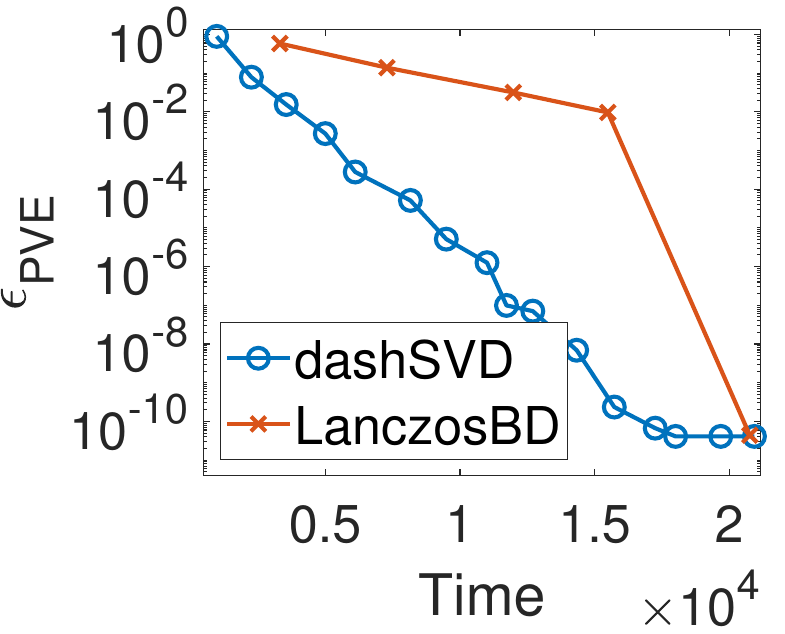}
			\includegraphics[width=3.4cm, trim=103 265 115 273,clip]{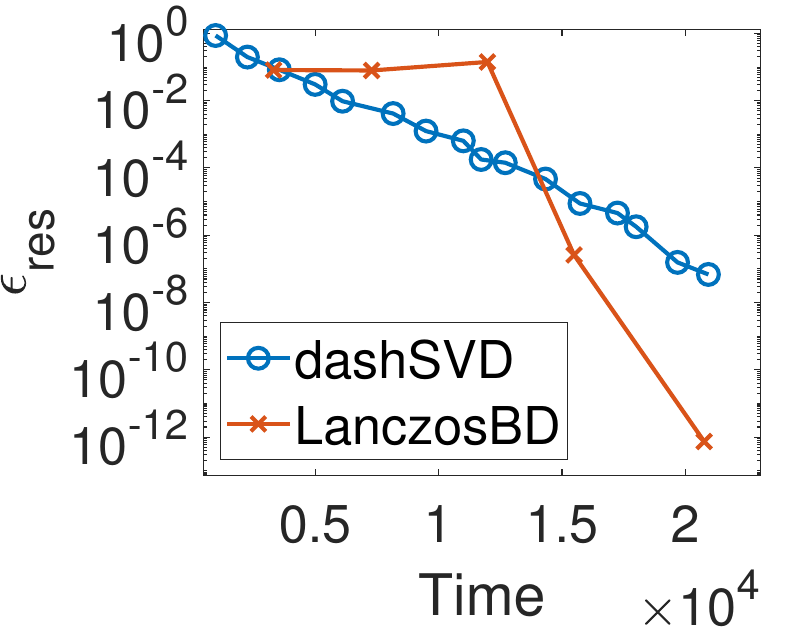}
			\includegraphics[width=3.4cm, trim=103 265 115 273,clip]{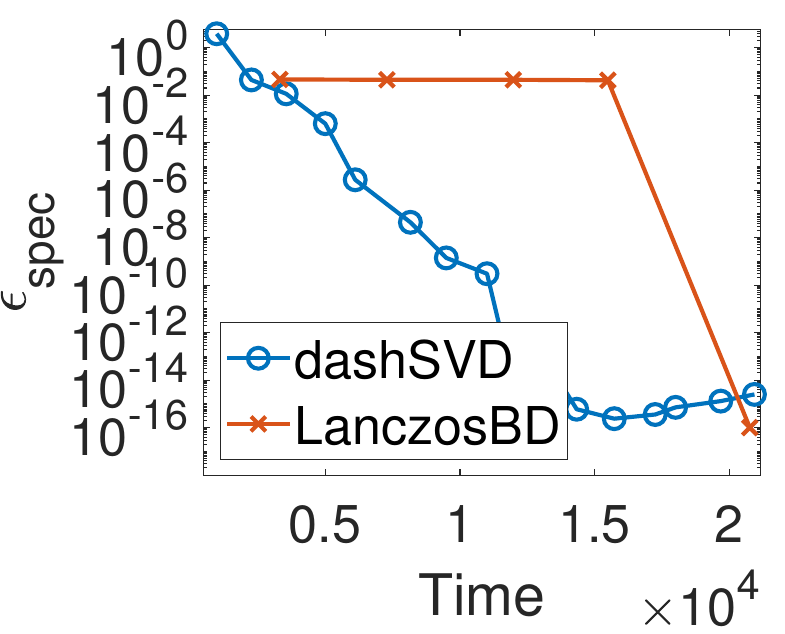}
			\includegraphics[width=3.4cm, trim=103 265 115 273,clip]{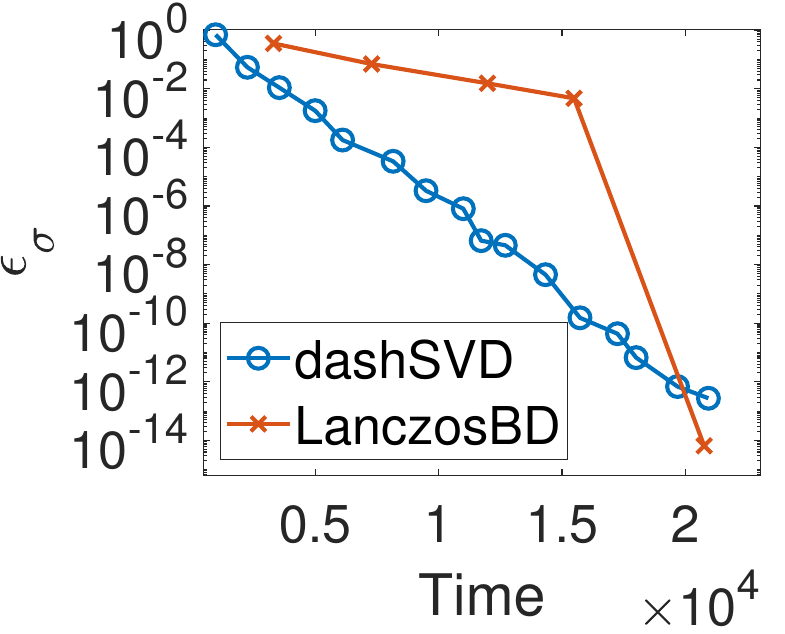}
		\end{minipage}
	}
	\caption{\notice{The error vs. time curves of dashSVD and \texttt{LanczosBD} in \texttt{svds} with single-thread computing ($k=100$). The unit of time is second. }}
	\label{fig:2} 
	\centering
\end{figure}
\atnn{For computing truncated SVD, most existing methods can trade off time against accuracy. In this subsection, we compare the efficiency of dashSVD with other methods at various accuracy levels. 
For this aim, we also implement the dashSVD without accuracy control, i.e. Alg. 4, where $p$ can be altered to trade off time against accuracy. Notice it has same runtime as dashSVD (Alg. 5) because the step for accuracy control costs negligible time.}  

Firstly, we compare dashSVD and \texttt{svds} both in Matlab. For fair comparison, we run them in single thread by \notice{executing command  “\texttt{maxNumCompThreads(1)}” in Matlab a prior}, and the \texttt{LanczosBD} in \texttt{svds} is actually compared. We varies the $p$ in dashSVD \atnn{(Alg. 4)} and the number of restartings in \texttt{LanczosBD} to plot the curves of runtime vs. error. The results of SNAP and uk-2005 in the four error metrics are shown in Fig. \ref{fig:2}, with the results of other matrices in Appendix \ref{secA3}.
The results reveal that dashSVD always costs  comparable  or less  time than \texttt{LanczosBD} for achieving a not very high accuracy, while for some cases (like uk-2005) \texttt{LanczosBD} needs a considerable number of restarting to converge to accurate result.  
For the metrics $\epsilon_{\textrm{PVE}}$, $\epsilon_\textrm{spec}$ and $\epsilon_{\sigma}$, dashSVD has remarkable advantage on computing time {(\atn{3.2X} and \atn{3.4X} for attaining error $10^{-1}$ in $\epsilon_{\textrm{PVE}}$ and  $\epsilon_{\sigma}$ for uk-2005 respectively)}. 
\begin{figure}[b]
	\setlength{\abovecaptionskip}{0 cm}
	\setlength{\belowcaptionskip}{0 cm}
	\centering
	\subfigure[SNAP] { \label{fig3:1}
		\begin{minipage} {14cm}
			\centering
			\includegraphics[width=3.4cm, trim=103 265 115 273,clip]{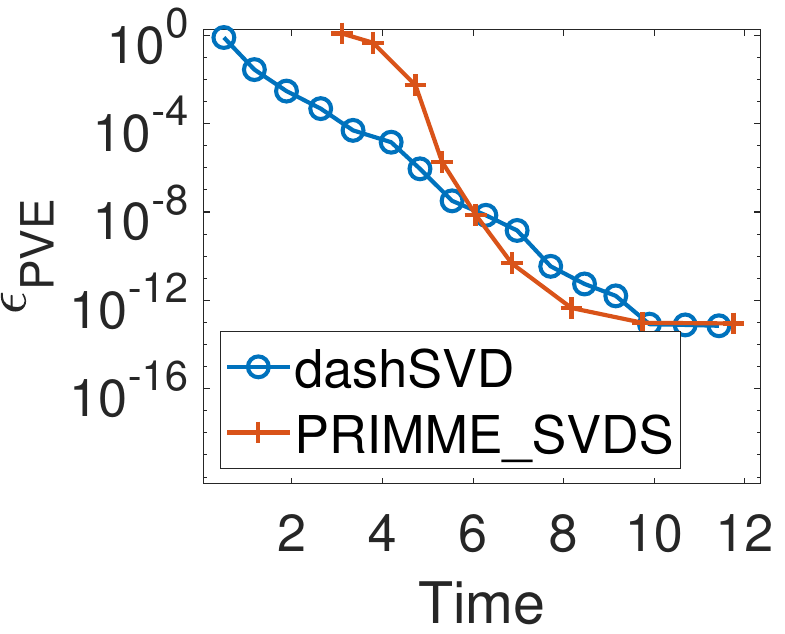} 
			\includegraphics[width=3.4cm, trim=103 265 115 273,clip]{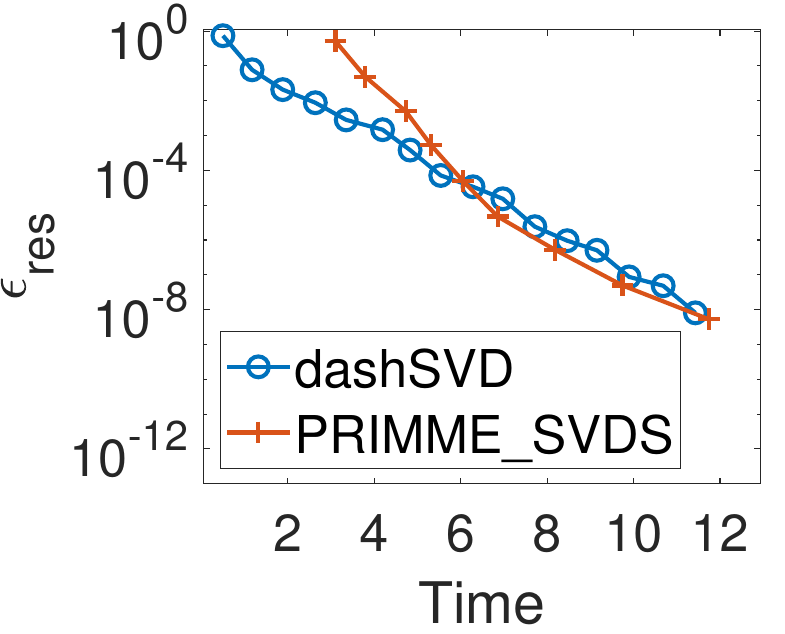} 
			\includegraphics[width=3.4cm, trim=103 265 115 273,clip]{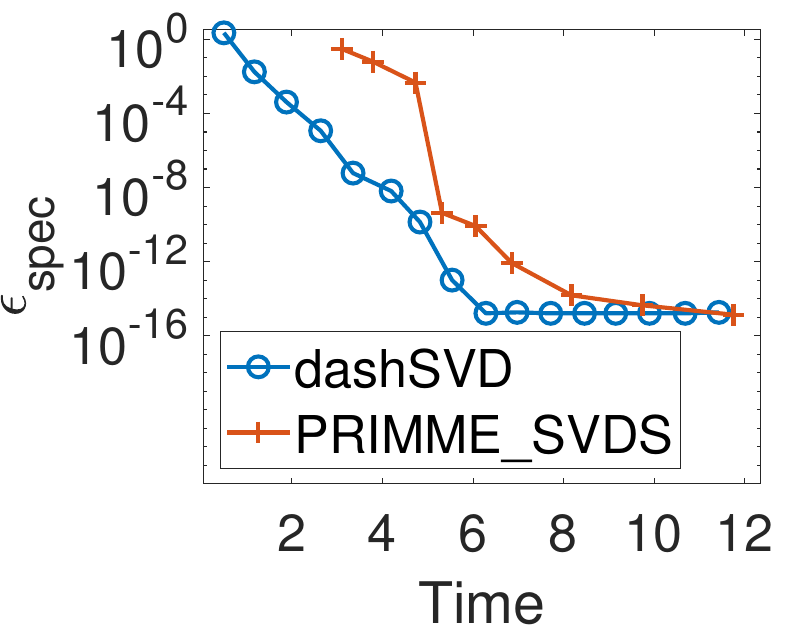} 
			\includegraphics[width=3.4cm, trim=103 265 115 273,clip]{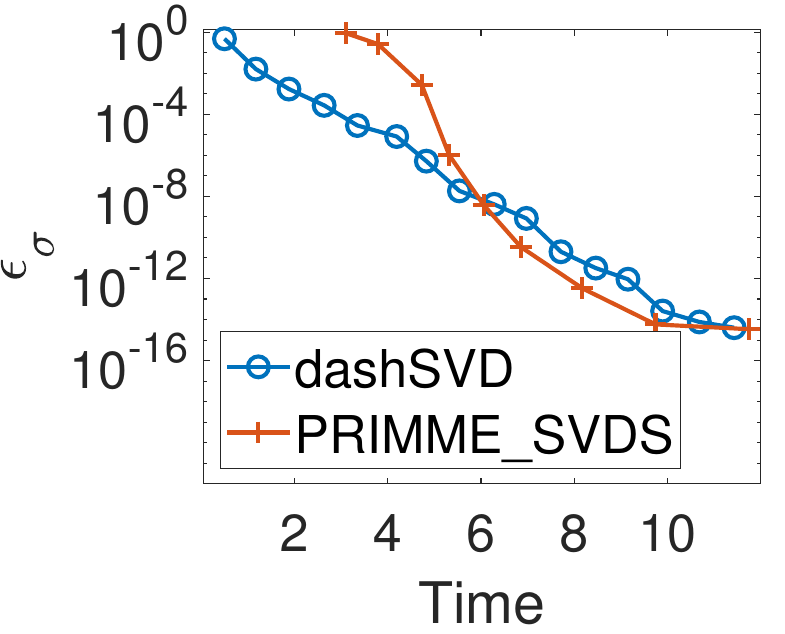} 
		\end{minipage}
	}\\[-1ex]
	\subfigure[uk-2005] { \label{fig3:2}
		\begin{minipage} {14cm}
			\centering
			\includegraphics[width=3.4cm, trim=103 265 115 273,clip]{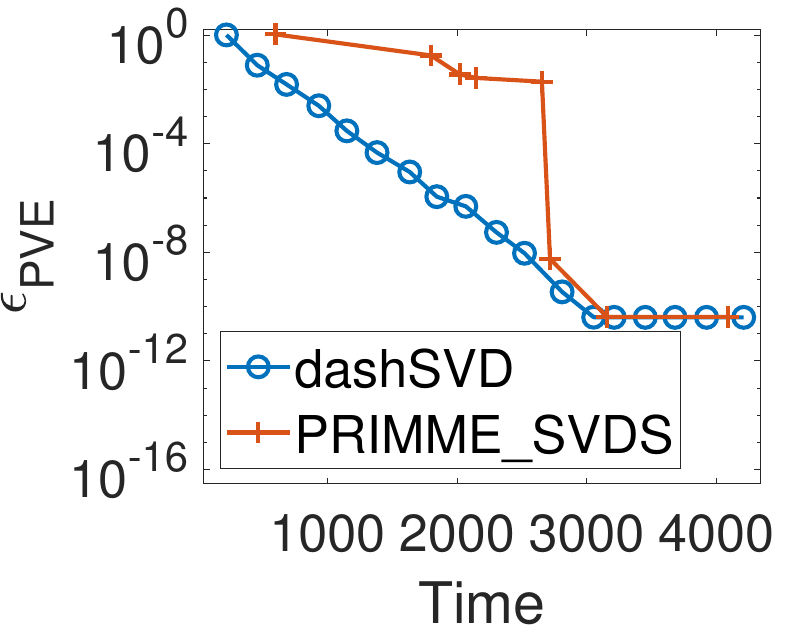} 
			\includegraphics[width=3.4cm, trim=103 265 115 273,clip]{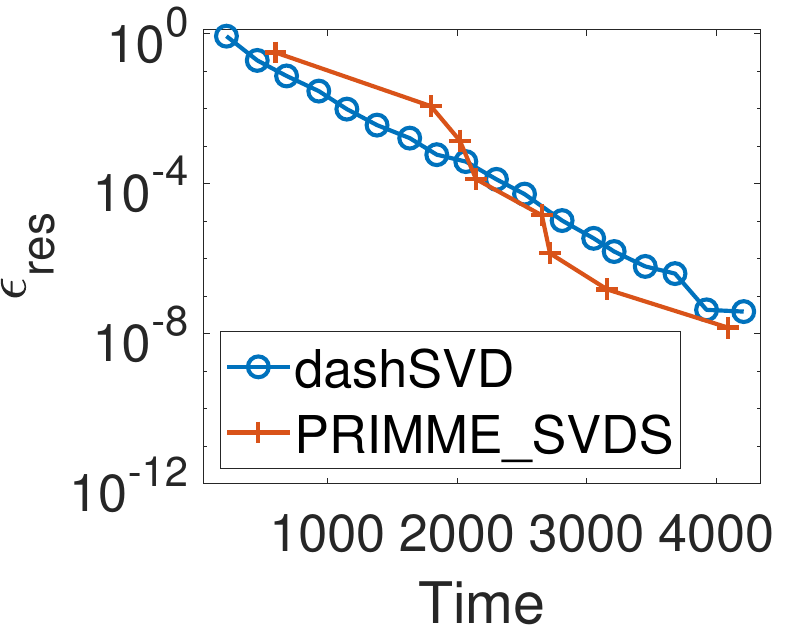} 
			\includegraphics[width=3.4cm, trim=103 265 115 273,clip]{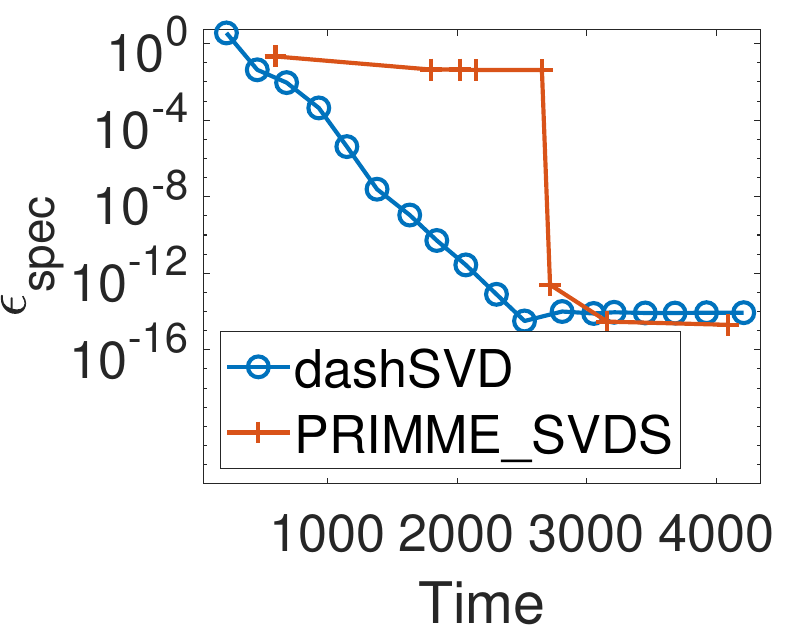} 
			\includegraphics[width=3.4cm, trim=103 265 115 273,clip]{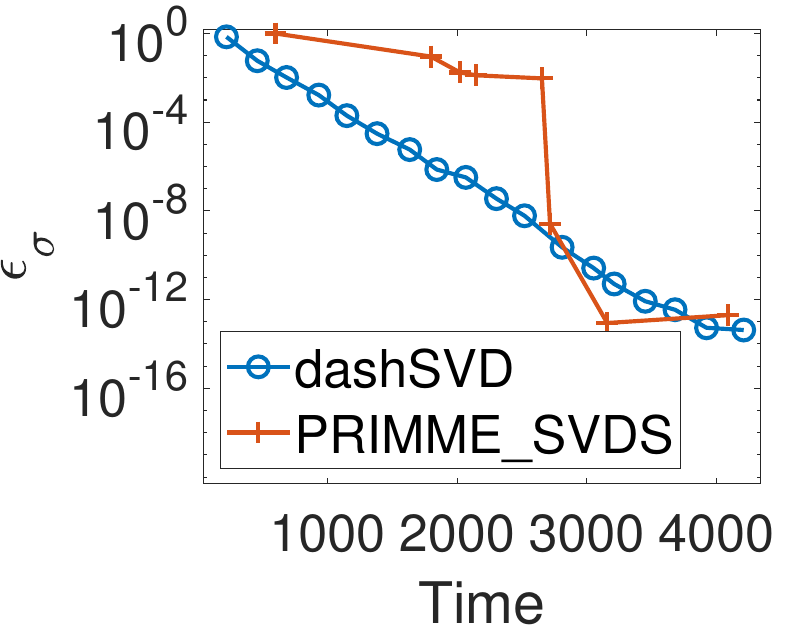} 
		\end{minipage}
	}
	\caption{\notice{The error vs. runtime curves of dashSVD and PRIMME\_SVDS with 8-thread computing ($k=100$). The unit of time is second. }}
	\label{fig:3} 
	\centering
\end{figure} 

And, we compare the C version of dashSVD with PRIMME\_SVDS for 8-thread parallel computing (as parallel \texttt{svds} is not efficient). To avoid the NUMA complications, we stick with one NUMA domain for comparison in one socket with 8 threads by \notice{ setting the environment variables ``\texttt{OMP\_PROC\_BIND=close}'' and ``\texttt{OMP\_PLACES=cores}''}. The error vs. runtime curves of SNAP and uk-2005 are plotted in Fig.~\ref{fig:3}. The results show that the proposed dashSVD runs remarkably faster than PRIMME\_SVDS for producing the results with same accuracy. 
For example, to achieve the results with $10^{-1}$ relative residual error ($\epsilon_\textrm{res}$) dashSVD runs \atn{3.2X} and \atn{1.3X} 
faster than  PRIMME\_SVDS for SNAP and uk-2005, respectively. With metric $\epsilon_\textrm{PVE}$, the corresponding speedup ratios become \atn{4.0X} and \atn{3.9X}. Comparing the results in Fig. \ref{fig:2} and Fig. \ref{fig:3}, we can also find out that the parallel speedup of dashSVD is about \atn{5.5 for 8-thread} computing. 
More error vs. runtime curves of MovieLens, Rucci1, Aminer and sk-2005 are plotted in Appendix \ref{secA3}. 
The memory \atnn{costs} of dashSVD, \texttt{LanczosBD} and PRIMME\_SVDS \atnn{are} listed in Table \ref{table:mem}. Table \ref{table:mem} shows that our dashSVD algorithm costs the same or comparable memory compared with \texttt{LanczosBD} and PRIMME\_SVDS.

\begin{table}[b]
	\setlength{\abovecaptionskip}{0.1 cm}
	\setlength{\belowcaptionskip}{0.1 cm}
	\caption{The memory cost (in unit of GB) of the dashSVD, \texttt{LanczosBD} and PRIMME\_SVDS.}
	\label{table:mem}
	\centering
			{
				\begin{spacing}{1}
					\renewcommand{\multirowsetup}{\centering}
					{\renewcommand{\arraystretch}{0.95}
						\begin{tabular}{cccc}
							\toprule
							Matrix & dashSVD & \texttt{LanczosBD} & PRIMME\_SVDS \\
							\midrule
							SNAP & 0.30 & 0.29 & \atn{0.50}\\
							MovieLens & 0.96 & 1.54 & \atn{1.00}\\
							Rucci1 & 4.66 & 4.01 & \atn{3.62}\\
							Aminer & 31.5 & 27.7 & \atn{23.1}\\
							uk-2005 & 147 & 135 & 133\\
							sk-2005 & 200 & 184 & 179\\
							\bottomrule 
						\end{tabular}
					}
				\end{spacing}
			}
		\end{table}

To validate dashSVD for \atn{matrices} with multiple singular values, a matrix with many multiple singular values called LargeRegFile obtained from SuiteSparse matrix collection~\citep{davis2011university} is tested in addition. 
This is a matrix in size 2,111,154$\times$801,374 with 2.3 nonzero elements per row on average. Firstly, we plot the error vs. runtime curves \atnn{for} LargeRegFile when $k=100$ in Fig. \ref{fig:err_lrf}\atnn{, as} those in Fig. \ref{fig:2} and \ref{fig:3}. Fig. \ref{fig:err_lrf_1} shows that \texttt{LanczosBD} needs \atn{a long time to converge to accurate result for LargeRegFile, which corresponds to a considerable number of restarting},
 and our dashSVD has remarkable advantages on computing time for the four metrics. From Fig. \ref{fig:err_lrf_2}, although PRIMME\_SVDS reaches lower $\epsilon_{\textrm{res}}$ within the same runtime compared with our dashSVD, the other error metrics are really large, which shows the singular triplets computed by PRIMME\_SVDS are wrong. The memory cost of dashSVD, \texttt{LanczosBD} and PRIMME\_SVDS is 5.78 GB, 5.03 GB and 5.09 GB, respectively. This also shows that our dashSVD algorithm costs comparable memory \atnn{to} \texttt{LanczosBD} and PRIMME\_SVDS\atnn{.} 
		
		\begin{figure}[t!]
			\setlength{\abovecaptionskip}{0.1 cm}
			\setlength{\belowcaptionskip}{0.1 cm}
			\centering
			\subfigure[The error vs. time curves of dashSVD and \atn{\texttt{LanczosBD}} in \texttt{svds} with single-thread computing] {
				\label{fig:err_lrf_1} 
				\begin{minipage}{14cm}
					\centering
					\includegraphics[width=3.4cm, trim=103 265 115 273,clip]{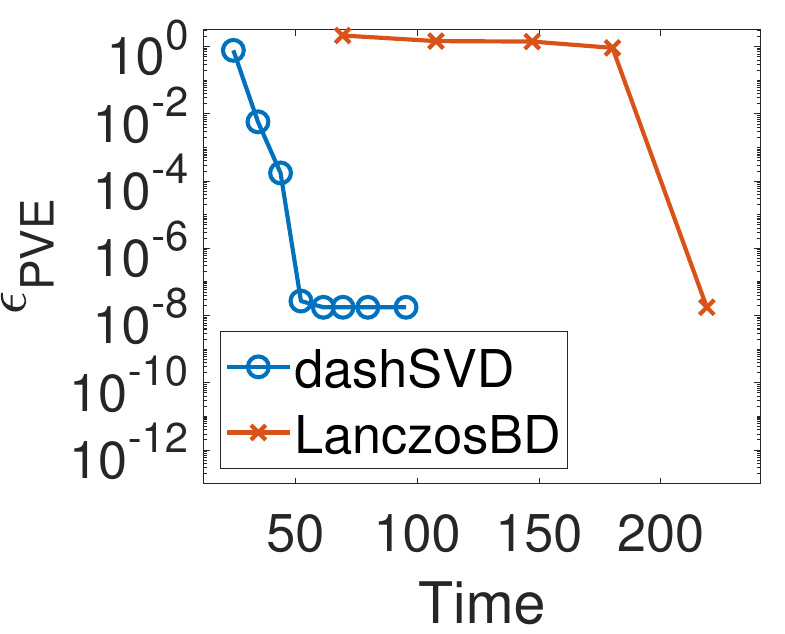}
					\includegraphics[width=3.4cm, trim=103 265 115 273,clip]{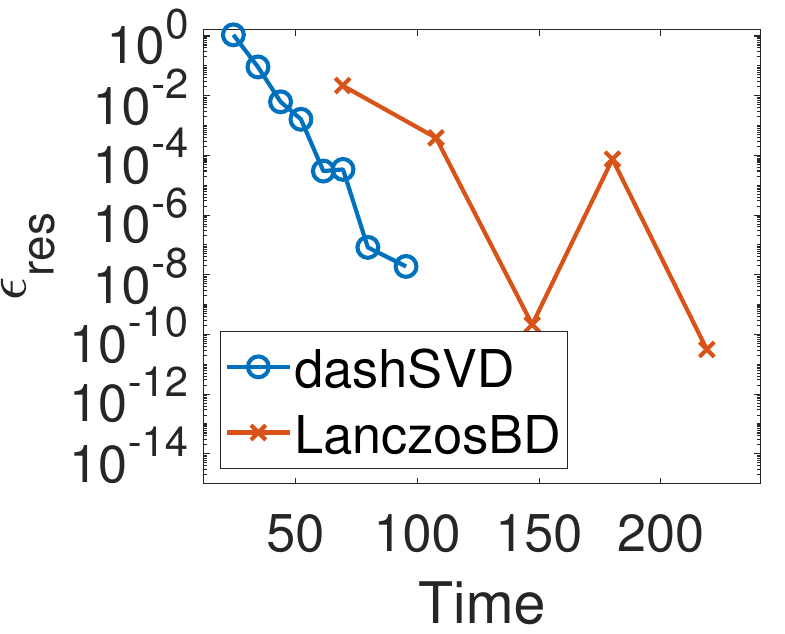}
					\includegraphics[width=3.4cm, trim=103 265 115 273,clip]{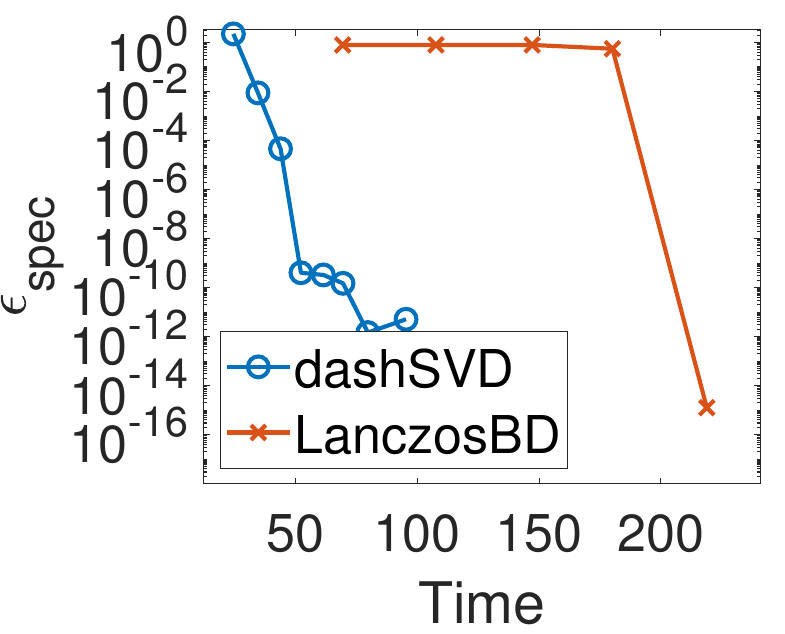}
					\includegraphics[width=3.4cm, trim=103 265 115 273,clip]{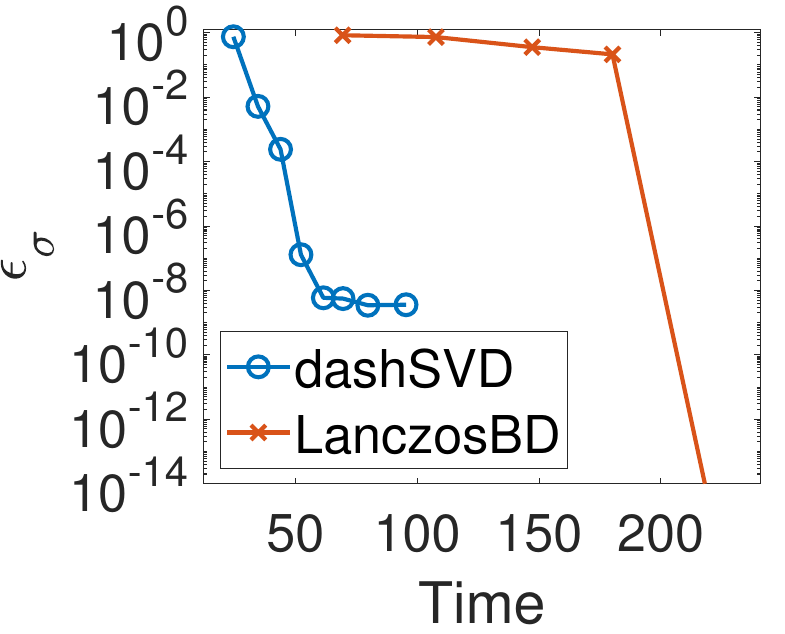}
				\end{minipage}
			}\\[-1ex]
			\subfigure[\atn{The error vs. time curves of dashSVD and PRIMME\_SVDS with 8-thread computing}] {
				\label{fig:err_lrf_2} 
				\begin{minipage} {14cm}
					\centering
					\includegraphics[width=3.4cm, trim=103 265 115 273,clip]{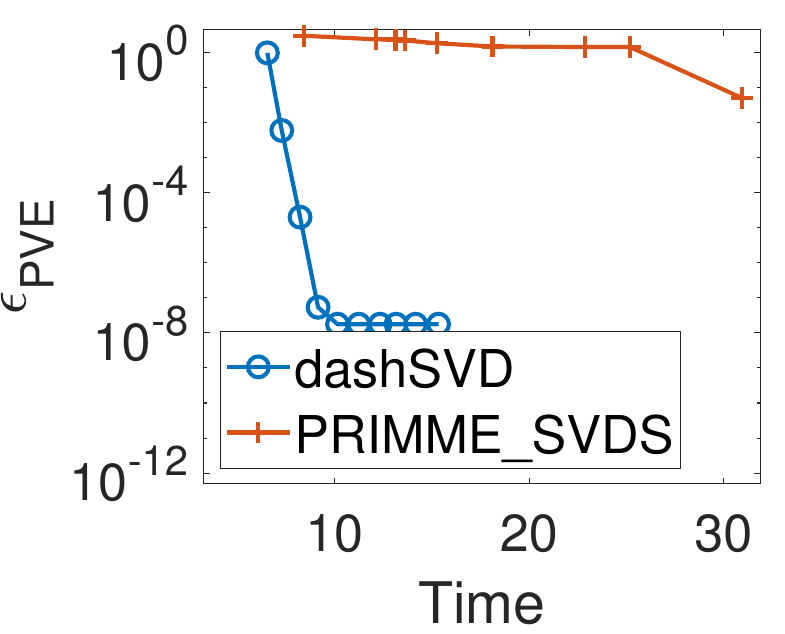} 
					\includegraphics[width=3.4cm, trim=103 265 115 273,clip]{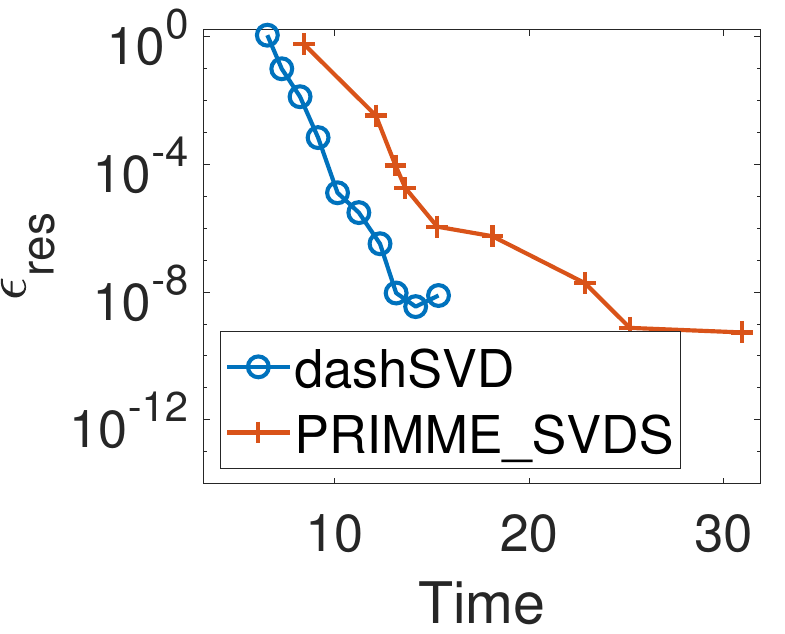} 
					\includegraphics[width=3.4cm, trim=103 265 115 273,clip]{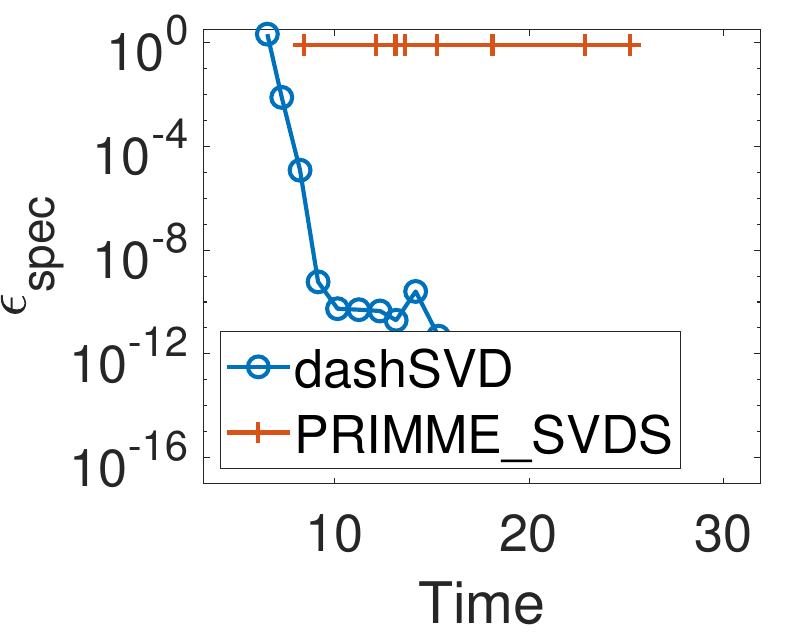} 
					\includegraphics[width=3.4cm, trim=103 265 115 273,clip]{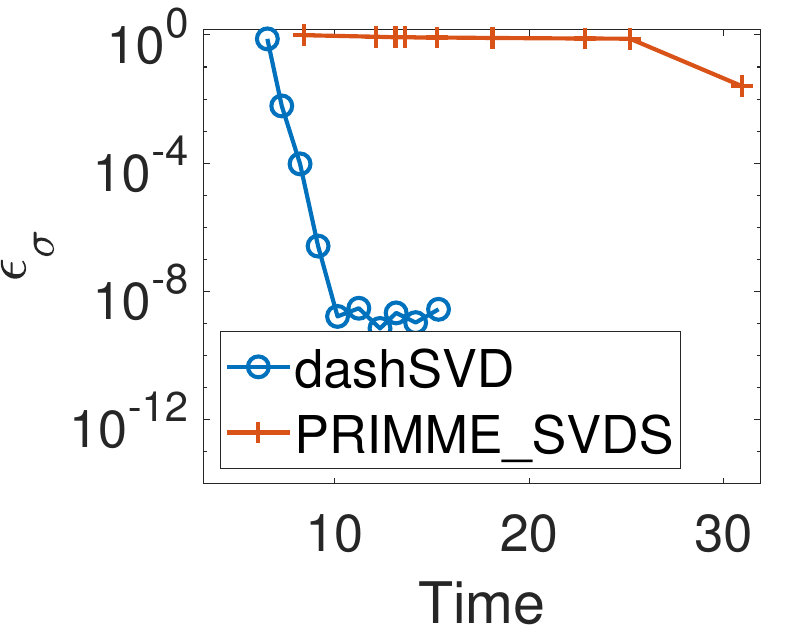} 
				\end{minipage}
			}
			\caption{\notice{The error vs. time curves of dashSVD compared with \texttt{LanczosBD} in \texttt{svds} and PRIMME\_SVDS for LargeRegFile ($k=100$). The unit of time is second. }}
			\label{fig:err_lrf}
			\centering
		\end{figure}

\begin{figure}[b]
\centering
\subfigure[The error curve of three randomized SVD algorithms (Alg.~1, Alg.~2$^*$, Alg.~2) with different values of power parameter] {
	\includegraphics[width=5.2cm,trim=93 265 90 275,clip]{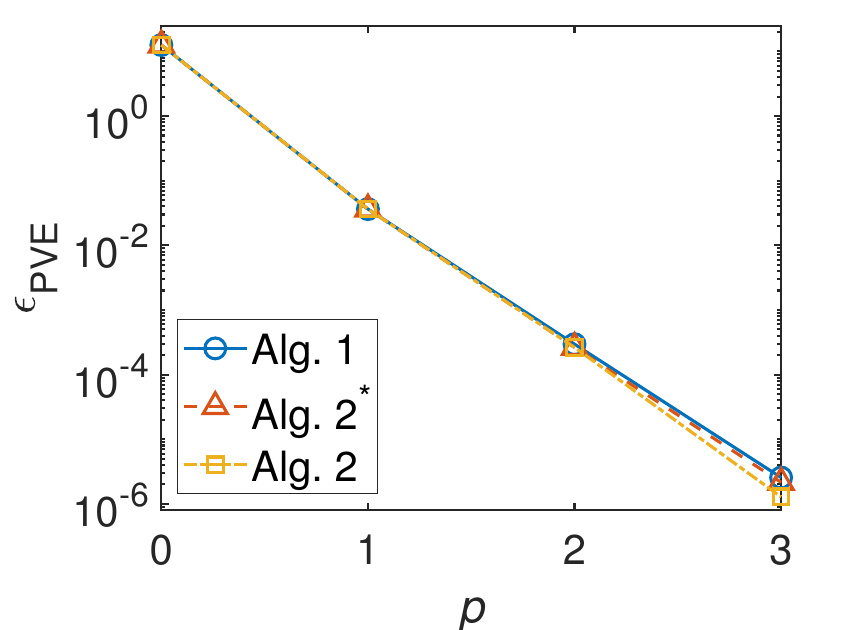}
	\label{lrf:1}
}
\hspace{.1in}
\subfigure[Singular values computed by different algorithms on LargeRegFile] {
	\includegraphics[width=4.2cm,trim=93 270 90 275,clip]{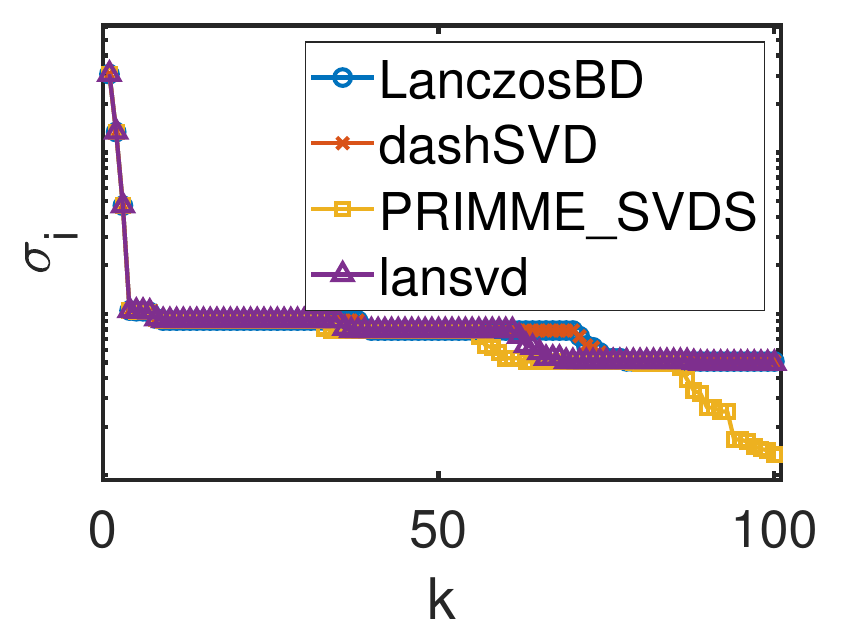}
	\label{lrf:2}
}
\caption{\notice{Experimental results on LargeRegFile ($k=100$).} }
\label{fig:lrf}
\centering
\end{figure}

Then, we draw the error curve of Alg.~1, Alg.~2$^*$ and Alg.~2 with different values of power parameter for $k = 100$ in Fig. \ref{lrf:1}. From \atnn{it} we see \atnn{the shifted power iteration improves the accuracy}. And, we test dashSVD \atnn{(Alg. 5)}, \texttt{LanczosBD}, PRIMME\_SVDS and lansvd \citep{propack} for computing the truncated SVD with $k=100$\atnn{. They are run} with 16 threads. \atnn{The} subspace dimensionalities of \texttt{LanczosBD}, lansvd and PRIMME\_SVDS are \atnn{set} to 150 for fair comparison. lansvd in PROPACK is an algorithm \atnn{for} accelerating \texttt{svds}. We test its Fortran version V2.14 \atn{parallelized} with OpenMP. \notice{Notice that lansvd is faster than \texttt{svds} in single-thread computing, but the parallel performance of lansvd is poor. Besides, the resulted singular vectors computed by lansvd occasionally contains NaN. This makes difficulty of comparison.}
The convergence tolerance parameter in \texttt{LanczosBD}, PRIMME\_SVDS is \atnn{set} to $10^{-10}$ and the convergence tolerance parameter in lansvd is \atnn{set} to $10^{-6}$ to make these algorithms produce results with low relative residual error, while the tol in dashSVD is \atnn{set} to $10^{-3}$ for producing results with low PVE.
The computed singular \atnn{values are} plotted in Fig. \ref{lrf:2}, and more results are listed in Table \ref{table:lrf}. \atnn{These results}  show that accuracy issue occurs for both lansvd and PRIMME\_SVDS on this matrix with multiple singular values. \atnn{With} the results in Fig. \ref{fig:err_lrf_2} \atnn{together we can see that} PRIMME\_SVDS is not robust for matrix with multiple singular values. On contrary, dashSVD performs very \atnn{well. Its} speedup to \texttt{LanczosBD} is \textbf{6X}, with \atnn{a small} $\epsilon_{\textrm{PVE}}$ of $1.8\times 10^{-8}$.

\begin{table}[h]
	\setlength{\abovecaptionskip}{0.1 cm}
	\setlength{\belowcaptionskip}{0.1 cm}
	\caption{Computational results of \texttt{LanczosBD}, lansvd, PRIMME\_SVDS and dashSVD (tol=$10^{-3}$) on LargeRegFile ($k$=100).
	}
	\label{table:lrf}
	\centering
	\begin{spacing}{0.8}
		\renewcommand{\multirowsetup}{\centering}
		\begin{tabular}{@{\,}c@{\,}c@{\,}c@{\,}c@{\,}c@{\,}c@{\,}c@{\,}} 
			\toprule
			Method & Time (s) & Mem (GB) & $\epsilon_{\textrm{PVE}}$ & $\epsilon_{\textrm{res}}$ & $\epsilon_{\textrm{spec}}$ & $\epsilon_{\sigma}$\\
			\midrule
                \atn{\texttt{LanczosBD}} & \atn{77.3} & \atn{5.02} & \atn{-} &  \atn{-} &  \atn{-} &  \atn{-}\\
			lansvd & 516 & 6.42 & NA$^1$ & NA$^1$ & NA$^1$ & 0.36\\
			PRIMME\_SVDS & 16.3 & 5.09 & 1.42 & 3.4E-9 & 0.79 & 0.74\\
			dashSVD & 12.7 & 5.78 & 1.8E-8 & 2.0E-5 & 1.1E-10 & 4.0E-9\\
			\bottomrule 
		\end{tabular}
		\vspace{-1pt}
		\begin{tablenotes}
			\item[] NA$^1$ represents the singular vectors computed by lansvd contains NaN, which makes some metrics of error criterion unavailable. 
		\end{tablenotes}
	\end{spacing}
\end{table}

\subsection{Validation of the \atnn{Accuracy-Control Scheme Based on PVE Criterion}}\label{sec4.2}

In this subsection, we validate the effectiveness of the proposed PVE accuracy control, and the efficiency brought by the 
surrogate strategy, i.e. using  $\sigma_i(\mathbf{A}^\mathrm{T}\mathbf{A}\mathbf{Q}-\alpha\mathbf{Q})+\alpha$ to monitor $\hat{\sigma_i}(\mathbf{A})^2$. We set $\mathrm{tol}=10^{-2}$ and compare the results of Alg.~5 with those obtained without the surrogate strategy. The latter realizes the accuracy control through iteratively computing SVD of matrix $\mathbf{B}=\mathbf{AQ}$ and checking the termination criterion (\ref{err:pve_approx}). The results are listed in Table \ref{table:1}. $N_{p}$ denotes the number of power iteration steps executed  when the algorithm terminates with satisfied criterion. If without the surrogate strategy, the resulted $N_{p}$ has same value or just smaller for 1. 
We can see that the dashSVD with accuracy control obtains results with good accuracy in the four metrics, and the result's $\epsilon_\textrm{PVE}$ is about or less than $10^{-2}$. The surrogate strategy enables convenient monitor of singular value, without which the runtime is increased by 4\% to 94\%.

\begin{table}[h]
\setlength{\abovecaptionskip}{0.1 cm}
\setlength{\belowcaptionskip}{0.1 cm}\caption{The results of the dashSVD with PVE accuracy control in 16-thread computing ($k$=100, tol=$10^{-2}$\atn{, and Inc. represents the increment of time).}}
\label{table:1}
\centering
		{
			\begin{spacing}{1}
				\renewcommand{\multirowsetup}{\centering}
				{\renewcommand{\arraystretch}{0.95}
					\begin{tabular}{@{\,}c@{\,}c@{\,}c@{\,}c@{\,}c@{\,}c@{\,}c@{\,}c@{\,}c@{\,}c@{\,}} 
						\toprule
						\multirow{2}{*}{Matrix} & \multicolumn{6}{c}{Alg. 5 } & \multicolumn{3}{c}{Alg. 5 w/o surrogate}\\
						\cmidrule(r){2-7}
						\cmidrule(r){8-10}
						& Time (s) & $N_{p}$ & $\epsilon_{\textrm{PVE}}$ & $\epsilon_{\textrm{res}}$ & $\epsilon_{\textrm{spec}}$ & $\epsilon_{\sigma}$ & Time (s) & \atn{Inc.} & $\epsilon_{\textrm{PVE}}$\\
						\midrule
						SNAP & 1.93 & 7 & 5.7E-3 & 3.1E-2 & 4.3E-4 & 3.3E-3 & 2.31 & 20\% & 5.3E-3 \\
						MovieLens & 14.4 & 6 & 2.6E-3 & 2.8E-2 & 6.0E-5 & 1.7E-3 & 15.4 & 7\% & 3.0E-3 \\
						Rucci1 & 13.6 & 9 & 1.9E-2 & 3.9E-2 & 9.9E-3 & 1.0E-2 & 26.4 & 94\% & 2.3E-3\\
						Aminer & 341 & 7 & 3.1E-3 & 2.7E-2 & 3.6E-4 & 1.8E-3 & 353 & 4\% & 6.6E-3\\
						uk-2005 & 877 & 6 & 2.4E-3 & 2.6E-2 & 2.8E-4 & 1.5E-3 & 986 & 12\% & 2.2E-3\\
						sk-2005 & 1385 & 5 & 8.4E-4 & 2.0E-2 & 5.2E-5 & 6.3E-4 & 1525 & 10\% & 1.2E-3\\
						\bottomrule 
					\end{tabular}
				}
			\end{spacing}
		}
	\end{table}	 

	\section{Conclusions}
	
	We have developed \atnn{a  randomized algorithm named dashSVD to efficiently compute the truncated SVD in the regime of low accuracy}. dashSVD is built on a shifted power iteration scheme with dynamic shifts and an accuracy-control mechanism monitoring the per vector error \atnn{based accuracy criterion}. Inheriting the parallelizability of randomized SVD algorithm, it exhibits much shorter runtime than the state-of-the-art SVD algorithms for \atn{sparse matrices}. \atnn{dashSVD is expected to be useful in the application scenarios where low-accuracy SVD is allowed}.
	
	\atnn{In the future, efficient implementation of the randomized SVD algorithm on distributed-memory parallel architecture can be explored.}
	
\appendix
		
		\section{The proof of Theorem 1}\label{secA1}
		We first give the following two Lemmas~\citep{rokhlin2010randomized}. \cref{lemma:4} is equivalent to Lemma A.1 in \cite{rokhlin2010randomized} and \cref{lemma:5} contains the equivalent forms of (A.18) and (A.40) in \cite{rokhlin2010randomized}. $\Vert\cdot\Vert$ denotes the spectral norm. 
		\begin{lemma}
			\label{lemma:4}
			Suppose that $i, k, l, m,$ and $n$ are positive integers with $k\le l\le n \le m$. Suppose further that $\mathbf{A}$ is a real $m\times n$ matrix, $\mathbf{Q}$ is a real $m\times k$ matrix whose columns are orthonormal, $\mathbf{R}$ is a real $k\times l$ matrix, $\mathbf{F}$ is a real $n\times l$ matrix, and $\mathbf{G}$ is a real $l\times n$ matrix.
			Then,
			\begin{equation}\label{lemma4:1}
				\begin{aligned}
					\Vert\mathbf{Q}\mathbf{Q}^\mathrm{T}\mathbf{A}-\mathbf{A}\Vert \le 2\Vert\mathbf{F}\mathbf{G}(\mathbf{A}^\mathrm{T}\mathbf{A})^i\mathbf{A}^\mathrm{T}-\mathbf{A}^\mathrm{T}\Vert+2\Vert\mathbf{F}\Vert\Vert\mathbf{QR}-(\mathbf{G}(\mathbf{A}^\mathrm{T}\mathbf{A})^i\mathbf{A}^\mathrm{T})^\mathrm{T}\Vert.
				\end{aligned}
			\end{equation}
		\end{lemma}
		
		\begin{lemma}
			\label{lemma:5}
			Suppose that $i, k, l, m,$ and $n$ are positive integers with $n\le m$ and $l=k+s\le n-k$. Suppose further that $\mathbf{A}$ is a real $m\times n$ matrix,  $\mathbf{G}$ is a real $l\times n$ matrix whose entries are i.i.d. Gaussian random variables of zero mean and unit variance, positive integer $j<k$ such that the $j$-th largest singular value $\sigma_j$ of $\mathbf{A}$ is positive. If real numbers $\beta$, $\gamma>1$ satisfying $\phi\!=\!$ $\frac{1}{\sqrt{2\pi(l\!-\!j\!+\!1)}}(\frac{e}{(l\!-\!j\!+\!1)\beta})^{l\!-\!j\!+\!1} \!+\!\frac{1}{4\gamma(\gamma^2\!-\!1)} (\frac{1}{\sqrt{\pi(n\!-\!k)}}(\frac{2\gamma^2}{e^{\gamma^2\!-\!1}})^{n\!-\!k}\! +\!\frac{1}{\sqrt{\pi l}}(\frac{2\gamma^2}{e^{\gamma^2\!-\!1}})^{l}) \!\le \! 1$. Then, there exists a real ${n\times l}$ matrix $\mathbf{F}$ such that
			\begin{equation}\label{lemma5:1}
				\begin{aligned}
					&\Vert\mathbf{F}\mathbf{G}(\mathbf{A}^\mathrm{T}\mathbf{A})^i\mathbf{A}^\mathrm{T}-\mathbf{A}^\mathrm{T}\Vert^2 \le \left(2l^2\beta^2\gamma^2\left(\frac{\sigma_{j+1}}{\sigma_j}\right)^{4i}+1\right)\sigma_{j+1}^2+\left(2l(n-k)\beta^2\gamma^2\left(\frac{\sigma_{k+1}}{\sigma_{j}}\right)^{4i}+1\right)\sigma_{k+1}^2,
				\end{aligned}
			\end{equation}
			and
			\begin{equation}\label{lemma5:2}
				\Vert\mathbf{F}\Vert\le \frac{\sqrt{l}\beta}{\sigma_j^{2i}},
			\end{equation}
			with probability not less than $1-\phi$.
		\end{lemma}

		According to \cref{lemma:4,lemma:5}, the following Lemma is derived, which bounds the approximation error of Alg.~1's result \citep{rokhlin2010randomized}. Notice that we assume $l=k+s\le n-k$ as it always holds in practical scenarios. 
		
		\begin{lemma}
			\label{lemma:6}
			Suppose $\mathbf{A}\in\mathbb{R}^{m\times n}~(m\ge n)$, and $k$, $s$ and $p$ are the parameters in Alg.~1. $l=k+s\le n-k$, and $\mathbf{Q}$ in size $m\times l$ is the obtained orthonormal matrix with Alg.~1. If positive integer $j<k$, and real numbers $\beta$, $\gamma>1$ satisfying $\phi\!=\!$ $\frac{1}{\sqrt{2\pi(l\!-\!j\!+\!1)}}(\frac{e}{(l\!-\!j\!+\!1)\beta})^{l\!-\!j\!+\!1} \!+\!\frac{1}{4\gamma(\gamma^2\!-\!1)} (\frac{1}{\sqrt{\pi(n\!-\!k)}}(\frac{2\gamma^2}{e^{\gamma^2\!-\!1}})^{n\!-\!k}\! +\!\frac{1}{\sqrt{\pi l}}(\frac{2\gamma^2}{e^{\gamma^2\!-\!1}})^{l}) \!\le \! 1$, then
			\begin{equation}\label{lemma6:1}
				\begin{aligned}
					\Vert\mathbf{Q}\mathbf{Q}^\mathrm{T}\mathbf{A}-\mathbf{A}\Vert \le &2\sqrt{\left(2l^2\beta^2\gamma^2\left(\frac{\sigma_{j+1}}{\sigma_j}\right)^{4p}+1\right)\sigma_{j+1}^2+
					\left(2l(n-k)\beta^2\gamma^2\left(\frac{\sigma_{k+1}}{\sigma_{j}}\right)^{4p}+1\right)\sigma_{k+1}^2},
				\end{aligned}
			\end{equation}
			with probability not less than $1-\phi$, where $\sigma_j$ denotes the $j$-th \atn{largest} singular value of $\mathbf{A}$.
		\end{lemma}
		
		\begin{proof}
			\atnn{In Alg.~1, the result matrix $\mathbf{Q}$ is formed by a set of orthonormal basis of $range(\mathbf{A}(\mathbf{A}^\mathrm{T}\mathbf{A})^p\mathbf{\Omega})$, where $\mathbf{\Omega}$ in size $n\times l$ is a Gaussian random matrix and $p$ is the power parameter. Therefore, each column of $\mathbf{A}(\mathbf{A}^\mathrm{T}\mathbf{A})^p\mathbf{\Omega}= (\mathbf{\Omega}^\mathrm{T}(\mathbf{A}^\mathrm{T}\mathbf{A})^p\mathbf{A}^\mathrm{T})^\mathrm{T}$ can be expressed as a linear combination of $\mathbf{Q}$'s columns. This means there is a matrix $\mathbf{R}$ in $l\times l$ satisfying $(\mathbf{\Omega}^\mathrm{T}(\mathbf{A}^\mathrm{T}\mathbf{A})^p\mathbf{A}^\mathrm{T})^\mathrm{T}= \mathbf{QR}$. So,
			\begin{equation}\label{l6p:1}
				\begin{aligned}
				    \Vert\mathbf{QR}-(\mathbf{\Omega}^\mathrm{T}(\mathbf{A}^\mathrm{T}\mathbf{A})^p\mathbf{A}^\mathrm{T})^\mathrm{T}\Vert
				    = 0.
				\end{aligned}
			\end{equation}
			Suppose $\mathbf{F}$ is a real $m\times l$ matrix, $k=l, p=i$, $\mathbf{G}=\mathbf{\Omega}^\mathrm{T}$ and $\mathbf{R}$ is the one in (\ref{l6p:1}).
			According to \cref{lemma:4}, we have
			\begin{equation}\label{l6p:2}
				\begin{aligned}
					\Vert\mathbf{Q}\mathbf{Q}^\mathrm{T}\mathbf{A}-\mathbf{A}\Vert &\le 2\Vert\mathbf{F}\mathbf{\Omega}^\mathrm{T}(\mathbf{A}^\mathrm{T}\mathbf{A})^p\mathbf{A}^\mathrm{T}-\mathbf{A}^\mathrm{T}\Vert + 2 \Vert\mathbf{F}\Vert\Vert\mathbf{QR}-(\mathbf{\Omega}^\mathrm{T}(\mathbf{A}^\mathrm{T}\mathbf{A})^p\mathbf{A}^\mathrm{T})^\mathrm{T}\Vert \\
					&= 2\Vert\mathbf{F}\mathbf{\Omega}^\mathrm{T}(\mathbf{A}^\mathrm{T}\mathbf{A})^p\mathbf{A}^\mathrm{T}-\mathbf{A}^\mathrm{T}\Vert. \\
				\end{aligned}
			\end{equation}
		Then, according to \cref{lemma:5} there exists $\mathbf{F}$ such that (\ref{lemma5:1}) holds. Combining it and (\ref{l6p:2}) we derive (\ref{lemma6:1}).
}
			
		\end{proof}
		
		Before proving \cref{theorem:1}, we prove \cref{lemma:7} which is similar to the proof of Lemma~A.2 in \cite{rokhlin2010randomized}.

		\begin{lemma}
			\label{lemma:7}
			Suppose that $i, k, l, m,$ and $n$ are positive integers with $n\le m$ and $l=k+s\le n-k$. Suppose further that $\mathbf{A}$ is a real $m\times n$ matrix,  $\mathbf{G}$ is a real $l\times n$ matrix whose entries are i.i.d. Gaussian random variables of zero mean and unit variance, positive integer $j<k$ such that the $j$-th \atn{largest} singular value $\sigma_j$ of $\mathbf{A}$ is positive. If real numbers $\beta$, $\gamma>1$ satisfying $\phi\!=\!$ $\frac{1}{\sqrt{2\pi(l\!-\!j\!+\!1)}}(\frac{e}{(l\!-\!j\!+\!1)\beta})^{l\!-\!j\!+\!1} \!+\!\frac{1}{4\gamma(\gamma^2\!-\!1)} (\frac{1}{\sqrt{\pi(n\!-\!k)}}(\frac{2\gamma^2}{e^{\gamma^2\!-\!1}})^{n\!-\!k}\! +\!\frac{1}{\sqrt{\pi l}}(\frac{2\gamma^2}{e^{\gamma^2\!-\!1}})^{l}) \!\le \! 1$.
			Then, there exists a real ${n\times l}$ matrix $\mathbf{F}$ such that
			\begin{equation}\label{lemma7:1}
				\begin{aligned}
					&\Vert\mathbf{F}\mathbf{G}\mathbf{A}^\mathrm{T}\prod_{c=1}^{i}(\mathbf{A}\mathbf{A}^\mathrm{T}-\alpha_c\mathbf{I})-\mathbf{A}^\mathrm{T}\Vert^2\\
					& \le \left(2l^2\beta^2\gamma^2\prod_{c=1}^{i}\left(\frac{\sigma_{j+1}^2-\alpha_c}{\sigma_j^2-\alpha_c}\right)^{2}\!+\!1\right)\sigma_{j+1}^2+
					\left(2l(n-k)\beta^2\gamma^2\prod_{c=1}^{i}\left(\frac{\sigma_{k+1}^2-\alpha_c}{\sigma_{j}^2-\alpha_c}\right)^{2}\!+\!1\right)\sigma_{k+1}^2,
				\end{aligned}
			\end{equation}
			and
			\begin{equation}\label{lemma7:2}
				\Vert\mathbf{F}\Vert\le \frac{\sqrt{l}\beta}{\prod_{c=1}^{i}(\sigma_j^2-\alpha_c)},
			\end{equation}
			with probability not less than $1-\phi$, where $\mathbf{I}$ is the identity matrix and $0\le\alpha_c\le\sigma_l^2/2$ for any $c\le i$.
		\end{lemma}
		
		\begin{proof}
			We following the proof of Lemma~A.2 in \cite{rokhlin2010randomized} to construct $\mathbf{F}$ satisfying (\ref{lemma7:1}) and (\ref{lemma7:2}). The formulation of SVD of $\mathbf{A}$ is $\mathbf{A}=\mathbf{U\Sigma V}^\mathrm{T}$, where $\mathbf{U}$ is a real $m\times n$ matrix whose columns are orthonormal, $\mathbf{\Sigma}$ is a real diagonal $n\times n$ matrix, and $\mathbf{V}$ is a real unitary $n\times n$ matrix, such that $\mathbf{\Sigma}(y,y) =\sigma_y$ for $y=1,2,\cdots,n$, and $\sigma_y$ is the $y$-th largest singular value of $\mathbf{A}$. Then, several auxiliary matrices $\mathbf{H}$, $\mathbf{\Lambda}$, $\mathbf{\Gamma}$, $\mathbf{S}_1$, $\mathbf{S}_2$ and $\mathbf{S}_3$ are defined to assist the proof. Suppose $\mathbf{H}$ is the leftmost $l\times j$ block of the $l\times n$ matrix $\mathbf{GV}$, $\mathbf{\Lambda}$ is the $l\times (k-j)$ block of $\mathbf{GV}$ whose first column is the $(k+1)$-th column of $\mathbf{GV}$, and $\mathbf{\Gamma}$ is the rightmost $l\times (n-k)$ block of $\mathbf{GV}$, so that
			\begin{equation}\label{p5:1}
				\mathbf{GV} = [\mathbf{H}, \mathbf{\Lambda}, \mathbf{\Gamma}].
			\end{equation}
			Combing the fact that $\mathbf{V}$ is real and unitary, and the fact that the entries of $\mathbf{G}$ are i.i.d. Gaussian random variables of zero mean and unit variance. Then, the entries of $\mathbf{H}$ are also i.i.d. Gaussian random variables of zero mean and unit variance, as are the entries of $\mathbf{\Lambda}$, and as are the entries of $\mathbf{\Gamma}$. Suppose $\mathbf{H}^\dagger$ is the real $j\times l$ matrix given by the formula
			\begin{equation}\label{p5:2}
				\mathbf{H}^\dagger=(\mathbf{H}^\mathrm{T}\mathbf{H})^{-1}\mathbf{H}^\mathrm{T},
			\end{equation} 
			where $\mathbf{H}^\mathrm{T}\mathbf{H}$ is invertible with high probability due to Lemma 2.7 in \cite{rokhlin2010randomized}. Suppose $\mathbf{S}_1$ to be the leftmost uppermost $j\times j$ block of $\mathbf{\Sigma}$, $\mathbf{S}_2$ is the leftmost uppermost $(k-j)\times (k-j)$ block of $\mathbf{\Sigma}$ whose leftmost uppermost entry is the entry in the $(j+1)$-th row and $(j+1)$-th column of $\mathbf{\Sigma}$, and $\mathbf{S}_3$ is the rightmost lowermost $(n-k)\times (n-k)$ block of $\mathbf{\Sigma}$, so that
			\begin{equation}\label{p5:3}
				\mathbf{\Sigma} = \left(\begin{matrix}
					\mathbf{S}_1 & \mathbf{0} & \mathbf{0}\\
					\mathbf{0} & \mathbf{S}_2 & \mathbf{0}\\
					\mathbf{0} & \mathbf{0} & \mathbf{S}_3\\
				\end{matrix}\right).
			\end{equation}
			Suppose $\mathbf{P}$ is \atn{a real $n\times l$ matrix given by}
			\begin{equation}\label{p5:4}
				\mathbf{P}=\left(\begin{matrix} \prod_{c=1}^{i}(\mathbf{S}_1^2-\alpha_{i-c+1}\mathbf{I})^{-1}\mathbf{H}^\dagger \\ \mathbf{0} \\ \mathbf{0} \\ \end{matrix}\right).
			\end{equation}		
			Finally,  $\mathbf{F}$ is constructed as an $n\times l$ matrix given by
			\begin{equation}\label{p5:5}
				\mathbf{F} = \mathbf{VP} = \mathbf{V}\left(\begin{matrix} \prod_{c=1}^{i}(\mathbf{S}_1^2-\alpha_{i-c+1}\mathbf{I})^{-1}\mathbf{H}^\dagger \\ \mathbf{0} \\ \mathbf{0} \\ \end{matrix}\right).
			\end{equation}
			According to the proof of Lemma~A.2 in \cite{rokhlin2010randomized},
			\begin{equation}\label{p5:6}
				\Vert\mathbf{H}^{\dagger}\Vert\le \sqrt{l}\beta,
			\end{equation}
			with probability not less than
			\begin{equation}\label{p5:7}
				1 -\frac{1}{\sqrt{2\pi(l-j+1)}}\left(\frac{e}{(l-j+1)\beta}\right)^{l-j+1}.
			\end{equation}
			Combining (\ref{p5:3}), (\ref{p5:5}), (\ref{p5:6}), the fact $\mathbf{\Sigma}$ is zero off its main diagonal, and the fact that $\mathbf{V}$ is unitary yields (\ref{lemma7:2}).
			
			We now show that $\mathbf{F}$ defined in (\ref{p5:5}) satisfies (\ref{lemma7:1}).
			
			Suppose $\hat{\mathbf{S}}_1=\prod_{c=1}^{i}(\mathbf{S}_1^2-\alpha_{i-c+1}\mathbf{I})^{-1}$. Combing (\ref{p5:1}) and (\ref{p5:5}) yields

			\begin{equation}\label{p5:8}
				\begin{aligned}
					&\mathbf{F}\mathbf{G}\mathbf{A}^\mathrm{T}\prod_{c=1}^{i}(\mathbf{A}\mathbf{A}^\mathrm{T}-\alpha_c\mathbf{I}) - \mathbf{A}^\mathrm{T}=\\
					&\mathbf{V}\left(\begin{matrix} \hat{\mathbf{S}}_1\mathbf{H}^\dagger \\ \mathbf{0} \\ \mathbf{0} \\ \end{matrix}\right) [\mathbf{H}, \mathbf{\Lambda}, \mathbf{\Gamma}]\mathbf{\Sigma}\prod_{c=1}^{i}(\mathbf{\Sigma}^2-\alpha_c\mathbf{I})\mathbf{U}^{\mathrm{T}}-\mathbf{V\Sigma U}^\mathrm{T}=\\
					&\mathbf{V}\left(\begin{matrix} \hat{\mathbf{S}}_1 & \hat{\mathbf{S}}_1\mathbf{H}^\dagger\mathbf{\Lambda} & \hat{\mathbf{S}}_1\mathbf{H}^\dagger\mathbf{\Gamma} \\ \mathbf{0} &\mathbf{0}&\mathbf{0}\\ \mathbf{0} & \mathbf{0} & \mathbf{0} \\ \end{matrix}\right)\prod_{c=1}^{i}(\mathbf{\Sigma}^2-\alpha_c\mathbf{I})\mathbf{\Sigma}\mathbf{U}^\mathrm{T}-\mathbf{V\Sigma U}^\mathrm{T}=\\
					&\mathbf{V}\left(\left(\begin{matrix} \mathbf{I} & \hat{\mathbf{S}}_1\mathbf{H}^\dagger\mathbf{\Lambda}\prod_{c=1}^{i}(\mathbf{S}_2^2-\alpha_c\mathbf{I}) & \hat{\mathbf{S}}_1\mathbf{H}^\dagger\mathbf{\Gamma}\prod_{c=1}^{i}(\mathbf{S}_3^2-\alpha_c\mathbf{I}) \\ \mathbf{0} & \mathbf{0}&\mathbf{0}\\ \mathbf{0} & \mathbf{0} & \mathbf{0} \\ \end{matrix}\right)-\mathbf{I}\right)\mathbf{\Sigma}\mathbf{U}^\mathrm{T}=\\
					&\mathbf{V}\left(\begin{matrix} \mathbf{0} & \hat{\mathbf{S}}_1\mathbf{H}^\dagger\mathbf{\Lambda}\prod_{c=1}^{i}(\mathbf{S}_2^2-\alpha_c\mathbf{I})\mathbf{S}_2 & \hat{\mathbf{S}}_1\mathbf{H}^\dagger\mathbf{\Gamma}\prod_{c=1}^{i}(\mathbf{S}_3^2-\alpha_c\mathbf{I})\mathbf{S}_3 \\ \mathbf{0} & -\mathbf{S}_2&\mathbf{0}\\ \mathbf{0} & \mathbf{0} & -\mathbf{S}_3 \\ \end{matrix}\right)\mathbf{U}^\mathrm{T}.
				\end{aligned}
			\end{equation}			
			Furthermore,
			\begin{equation}\label{p5:9}
				\begin{aligned}
					&\Bigg\Vert\left(\begin{matrix} \mathbf{0} & \hat{\mathbf{S}}_1\mathbf{H}^\dagger\mathbf{\Lambda}\prod_{c=1}^{i}(\mathbf{S}_2^2-\alpha_c\mathbf{I})\mathbf{S}_2 & \hat{\mathbf{S}}_1\mathbf{H}^\dagger\mathbf{\Gamma}\prod_{c=1}^{i}(\mathbf{S}_3^2-\alpha_c\mathbf{I})\mathbf{S}_3 \\ \mathbf{0} & -\mathbf{S}_2&\mathbf{0}\\ \mathbf{0} & \mathbf{0} & -\mathbf{S}_3 \\ \end{matrix}\right)\Bigg\Vert^2\le\\
					& \Big\Vert\hat{\mathbf{S}}_1\mathbf{H}^\dagger\mathbf{\Lambda}\prod_{c=1}^{i}(\mathbf{S}_2^2-\alpha_c\mathbf{I})\mathbf{S}_2\Big\Vert^2+\Big\Vert\hat{\mathbf{S}}_1\mathbf{H}^\dagger\mathbf{\Gamma}\prod_{c=1}^{i}(\mathbf{S}_3^2-\alpha_c\mathbf{I})\mathbf{S}_3\Big\Vert^2+\Vert\mathbf{S}_2\Vert^2 + \Vert\mathbf{S}_3\Vert^2.
				\end{aligned}
			\end{equation}		
			Moreover, 
			\begin{equation}\label{p5:10}
				\Big\Vert\hat{\mathbf{S}}_1\mathbf{H}^\dagger\mathbf{\Lambda}\prod_{c=1}^{i}(\mathbf{S}_2^2-\alpha_c\mathbf{I})\mathbf{S}_2\Big\Vert \le \Vert\hat{\mathbf{S}}_1\Vert\Vert\mathbf{H}^\dagger\Vert\Vert\mathbf{\Lambda}\Vert\Vert\prod_{c=1}^{i}(\mathbf{S}_2^2-\alpha_c\mathbf{I})\Vert\Vert\mathbf{S}_2\Vert,
			\end{equation}
			and
			\begin{equation}\label{p5:11}
				\Big\Vert\hat{\mathbf{S}}_1\mathbf{H}^\dagger\mathbf{\Gamma}\prod_{c=1}^{i}(\mathbf{S}_3^2-\alpha_c\mathbf{I})\mathbf{S}_3\Big\Vert \le \Vert\hat{\mathbf{S}}_1\Vert\Vert\mathbf{H}^\dagger\Vert\Vert\mathbf{\Gamma}\Vert\Vert\prod_{c=1}^{i}(\mathbf{S}_3^2-\alpha_c\mathbf{I})\Vert\Vert\mathbf{S}_3\Vert.
			\end{equation}		
			Then, according to \cref{proposition:1} and the fact $0\le\alpha_c\le\sigma_l^2/2$ for any $c\le i$ we can get
			\begin{equation}\label{p5:12}
				\Vert\hat{\mathbf{S}}_1\Vert \le \prod_{c=1}^i\frac{1}{\sigma_j^2-\alpha_c},
			\end{equation}
			\begin{equation}\label{p5:13}
				\Vert\prod_{c=1}^{i}(\mathbf{S}_2^2-\alpha_c\mathbf{I})\Vert \le \prod_{c=1}^i(\sigma_{j+1}^2-\alpha_c),
			\end{equation}
			\begin{equation}\label{p5:14}
				\Vert\prod_{c=1}^{i}(\mathbf{S}_3^2-\alpha_c\mathbf{I})\Vert \le \prod_{c=1}^i(\sigma_{k+1}^2-\alpha_c),
			\end{equation}
			\begin{equation}\label{p5:15}
				\Vert\mathbf{S}_2\Vert \le \sigma_{j+1},
			\end{equation}
			\begin{equation}\label{p5:16}
				\Vert\mathbf{S}_3\Vert \le \sigma_{k+1}.
			\end{equation}		
			Then, combining (\ref{p5:9})$\sim$(\ref{p5:16}) and the fact that the columns of $\mathbf{U}$ and $\mathbf{V}$ are orthonormal, yields
			\begin{equation}\label{p5:17}
				\begin{aligned}
					&\Vert\mathbf{F}\mathbf{G}\mathbf{A}^\mathrm{T}\prod_{c=1}^{i}(\mathbf{A}\mathbf{A}^\mathrm{T}-\alpha_c\mathbf{I}) - \mathbf{A}^\mathrm{T}\Vert^2\le\\
					&\left(\Vert\mathbf{H}^\dagger\Vert^2\Vert\mathbf{\Lambda}\Vert^2\prod_{c=1}^i(\frac{\sigma_{j+1}^2-\alpha_c}{\sigma_j^2-\alpha_c})^2+1\right)\sigma_{j+1}^2 +\left(\Vert\mathbf{H}^\dagger\Vert^2\Vert\mathbf{\Gamma}\Vert^2\prod_{c=1}^i(\frac{\sigma_{k+1}^2-\alpha_c}{\sigma_j^2-\alpha_c})^2+1\right)\sigma_{k+1}^2.
				\end{aligned}
			\end{equation}	
			Combing Lemma 2.6 in \cite{rokhlin2010randomized} and the fact that the entries of $\mathbf{\Lambda}$ are i.i.d. Gaussian random variables of zero mean and unit variance, as are the entries of $\mathbf{\Gamma}$, \atn{yields}
			\begin{equation}\label{p5:18}
				\Vert\mathbf{\Lambda}\Vert\le \sqrt{2l}\gamma,~\Vert\mathbf{\Gamma}\Vert\le \sqrt{2(n-k)}\gamma,
			\end{equation}
			with probability not less than
			\begin{equation}\label{p5:19}
				\begin{aligned}
					&1-\frac{1}{4(\gamma^2-1)\sqrt{\pi(n-k)\gamma^2}}\left(\frac{2\gamma^2}{e^{\gamma^2-1}}\right)^{(n-k)}
					-\frac{1}{4(\gamma^2-1)\sqrt{\pi l\gamma^2}}\left(\frac{2\gamma^2}{e^{\gamma^2-1}}\right)^{l}.
				\end{aligned}
			\end{equation}	
			Finally, combining (\ref{p5:6}), (\ref{p5:17}) and (\ref{p5:18}) yields
			\begin{equation}
				\begin{aligned}
					&\Vert\mathbf{F}\mathbf{G}\mathbf{A}^\mathrm{T}\prod_{c=1}^{i}(\mathbf{A}\mathbf{A}^\mathrm{T}-\alpha_c\mathbf{I})-\mathbf{A}^\mathrm{T}\Vert^2\\ 
					&\le \left(2l^2\beta^2\gamma^2\prod_{c=1}^{i}\left(\frac{\sigma_{j+1}^2-\alpha_c}{\sigma_j^2-\alpha_c}\right)^{2}\!+\!1\right)\sigma_{j+1}^2+
					\left(2l(n-k)\beta^2\gamma^2\prod_{c=1}^{i}\left(\frac{\sigma_{k+1}^2-\alpha_c}{\sigma_{j}^2-\alpha_c}\right)^{2}\!+\!1\right)\sigma_{k+1}^2,
				\end{aligned}
			\end{equation}
			with a probability not less than $1-\phi$.
			
		\end{proof}
		
		Now, we can prove \cref{theorem:1}. 
		\begin{proof}
			Suppose $\mathbf{I}$ is an identity matrix. In Alg.~2, \atnn{the result matrix $\mathbf{Q}$ is formed by a set of orthonormal basis of $range(\prod_{c=1}^{p}(\mathbf{A}\mathbf{A}^\mathrm{T}-\alpha_{p-c+1}\mathbf{I})\mathbf{A}\mathbf{\Omega})$}, where $\mathbf{\Omega}$ in size $n\times l$ is a Gaussian random matrix and $p$ is the power parameter. 
			\atnn{Therefore, each column of $\prod_{c=1}^{p}(\mathbf{A}\mathbf{A}^\mathrm{T}-\alpha_{p-c+1}\mathbf{I})\mathbf{A}\mathbf{\Omega}=(\mathbf{\Omega}^\mathrm{T}\mathbf{A}^\mathrm{T}\prod_{c=1}^{p}(\mathbf{A}\mathbf{A}^\mathrm{T}-\alpha_c\mathbf{I}))^\mathrm{T}$ can be expressed as a linear combination of $\mathbf{Q}$'s columns. This means there is a matrix $\mathbf{R}$ in $l\times l$ satisfying $(\mathbf{\Omega}^\mathrm{T}\mathbf{A}^\mathrm{T}\prod_{c=1}^{p}(\mathbf{A}\mathbf{A}^\mathrm{T}-\alpha_c\mathbf{I}))^\mathrm{T}=\mathbf{QR}$. So,
			\begin{equation}\label{p6:1}
				\begin{aligned}
					\Vert\mathbf{QR}-(\mathbf{\Omega}^\mathrm{T}\mathbf{A}^\mathrm{T}\prod_{c=1}^{p}(\mathbf{A}\mathbf{A}^\mathrm{T}-\alpha_c\mathbf{I}))^\mathrm{T}\Vert 
					=0.
				\end{aligned}
			\end{equation}
			Suppose $\mathbf{F}$ is a real $m\times l$ matrix, $k=l$, $p=i$, $\mathbf{G}=\mathbf{\Omega}^\mathrm{T}$ and $\mathbf{R}$ is the one in (\ref{p6:1}).} According to the proof of Lemma~A.1 in \cite{rokhlin2010randomized}, we have
			\begin{equation}\label{p6:2}
				\begin{aligned}
					\Vert\mathbf{Q}\mathbf{Q}^\mathrm{T}\mathbf{A}-\mathbf{A}\Vert &\le
					2\Vert\mathbf{F}\mathbf{\Omega}^\mathrm{T}\mathbf{A}^\mathrm{T}\prod_{c=1}^{p}(\mathbf{A}\mathbf{A}^\mathrm{T}-\alpha_c\mathbf{I})-\mathbf{A}^\mathrm{T}\Vert + 2 \Vert\mathbf{F}\Vert\Vert\mathbf{QR}-(\mathbf{\Omega}^\mathrm{T}\mathbf{A}^\mathrm{T}\prod_{c=1}^{p}(\mathbf{A}\mathbf{A}^\mathrm{T}-\alpha_c\mathbf{I}))^\mathrm{T}\Vert\\
					&=2\Vert\mathbf{F}\mathbf{\Omega}^\mathrm{T}\mathbf{A}^\mathrm{T}\prod_{c=1}^{p}(\mathbf{A}\mathbf{A}^\mathrm{T}-\alpha_c\mathbf{I})-\mathbf{A}^\mathrm{T}\Vert.
				\end{aligned}
			\end{equation}
			According to \cref{lemma:7} there exists $\mathbf{F}$ such that \atnn{(\ref{lemma7:1}) holds. Combing it and (\ref{p6:2}) we derive (\ref{theorem1:1}) in} \cref{theorem:1}.
			
		\end{proof}
		
		\begin{corollary}
			If $\mathbf{A}\in\mathbb{R}^{m\times n}~(m\ge n)$, Gaussian random matrix $\mathbf{\Omega}\in\mathbb{R}^{n\times l}$ and parameters $p$, $l$, $j$, $k$, $\beta$ and $\gamma$ are the same, the bound of approximation error of Alg.~2's result given in right hand side of (\ref{theorem1:1}) in \cref{theorem:1} is not larger than that of Alg.~1's result given in right hand side of (\ref{lemma6:1}) in \cref{lemma:6}. Moreover, the equality of the both bounds only occurs when $\sigma_j=\sigma_{j+1}=\cdots=\sigma_{k+1}$.
		\end{corollary}
		
		\begin{proof}
			Because for any $c$, $1<c\le p$, we have positive number $\alpha_c\le\sigma_l^2/2$ according to Alg.~2 ($\alpha_1=0$ in Alg.~2), we can derive
			\begin{equation}
				\frac{\sigma_{j+1}^2-\alpha_c}{\sigma_j^2-\alpha_c} \le \frac{\sigma_{j+1}^2}{\sigma_j^2}, ~ ~ \frac{\sigma_{k+1}^2-\alpha_c}{\sigma_j^2-\alpha_c} \le \frac{\sigma_{k+1}^2}{\sigma_j^2}.
			\end{equation}
			Therefore,
			\begin{equation}
				\prod_{c=1}^{p}\left(\frac{\sigma_{j+1}^2-\alpha_c}{\sigma_j^2-\alpha_c}\right)^2 \le \left(\frac{\sigma_{j+1}}{\sigma_j}\right)^{4p}, ~ ~
				\prod_{c=1}^{p}\left(\frac{\sigma_{k+1}^2-\alpha_c}{\sigma_j^2-\alpha_c}\right)^2 \le \left(\frac{\sigma_{k+1}}{\sigma_j}\right)^{4p},
			\end{equation}
			which suggests that the bound of approximation error of Alg.~2's result given in right hand side of (\ref{theorem1:1}) in Theorem~1 is not larger than that of Alg.~1's result given in right hand side of (\ref{lemma6:1}) in Lemma~6.
			
			Because $j<j+1< k+1$, $\sigma_j =\sigma_{k+1}$ only occurs when $\sigma_j=\sigma_{j+1}=\cdots=\sigma_{k+1}$. Otherwise, $\sigma_j >\sigma_{k+1}$ which derives
			\begin{equation}
				 \frac{\sigma_{k+1}^2-\alpha_c}{\sigma_j^2-\alpha_c} < \frac{\sigma_{k+1}^2}{\sigma_j^2}, ~  ~ \prod_{c=1}^{p}\left(\frac{\sigma_{k+1}^2-\alpha_c}{\sigma_j^2-\alpha_c}\right)^2 < \left(\frac{\sigma_{k+1}}{\sigma_j}\right)^{4p} 
			\end{equation}
		and the bound of approximation error given in (\ref{theorem1:1}) is less than that given in (\ref{lemma6:1}). Therefore, the equality of the both bounds only occurs when $\sigma_j=\sigma_{j+1}=\cdots=\sigma_{k+1}$.
		\end{proof}

    	\section{More comparison results}\label{secA3}

\begin{figure}[t]
	\setlength{\abovecaptionskip}{0.1 cm}
	\setlength{\belowcaptionskip}{0.1 cm}
	\centering
	\subfigure[SNAP \atn{($\Delta p=1$)}] { \label{figfrpca:1} 
		\includegraphics[width=3.5cm,trim=15 5 20 10,clip]{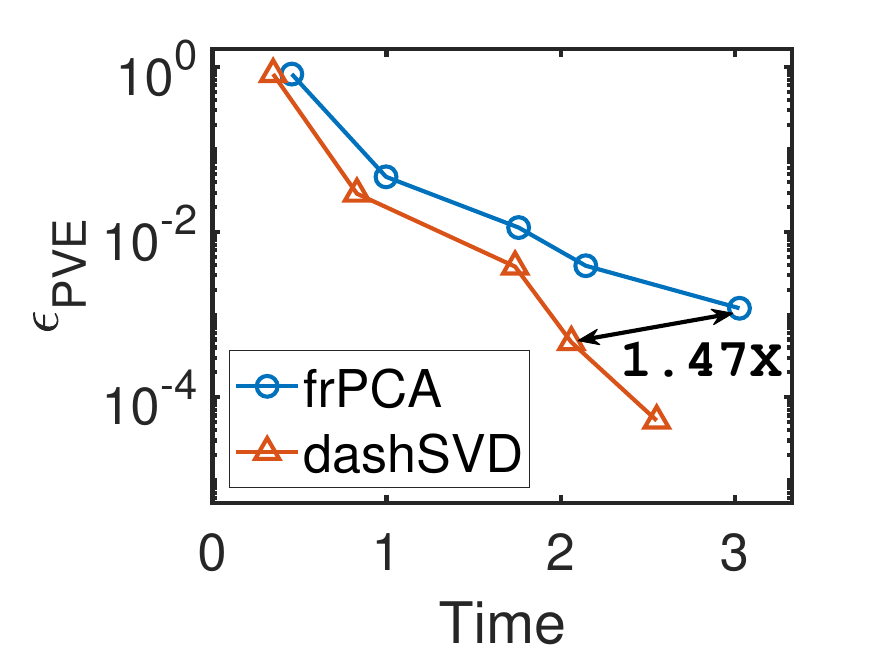} 
	}
	\subfigure[MovieLens \atn{($\Delta p=1$)}] { \label{figfrpca:2} 
		\includegraphics[width=3.5cm,trim=15 5 20 10,clip]{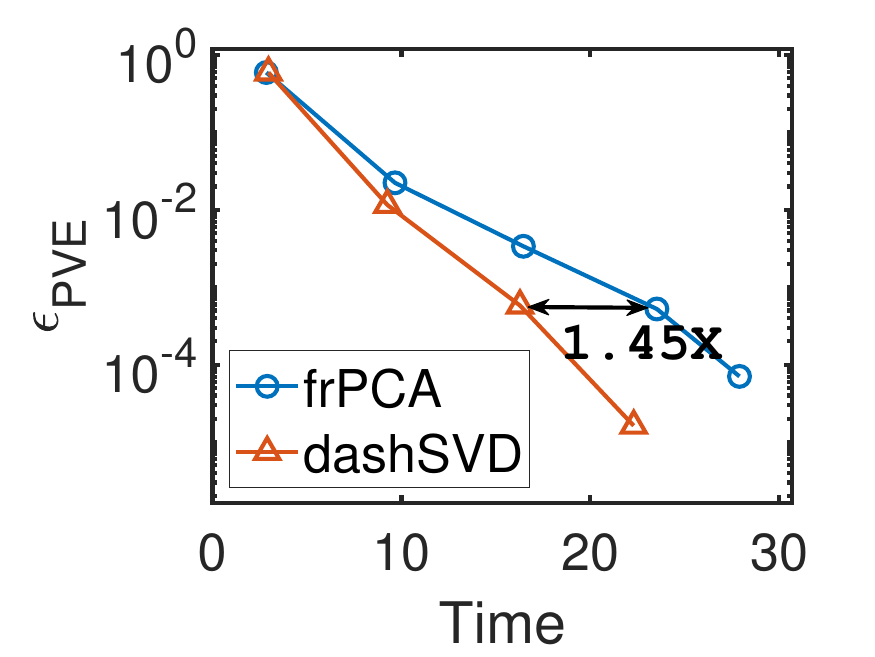} 
	}
	\subfigure[Rucci1 \atn{($\Delta p=2$)}] { \label{figfrpca:3} 
		\includegraphics[width=3.5cm,trim=15 5 20 10,clip]{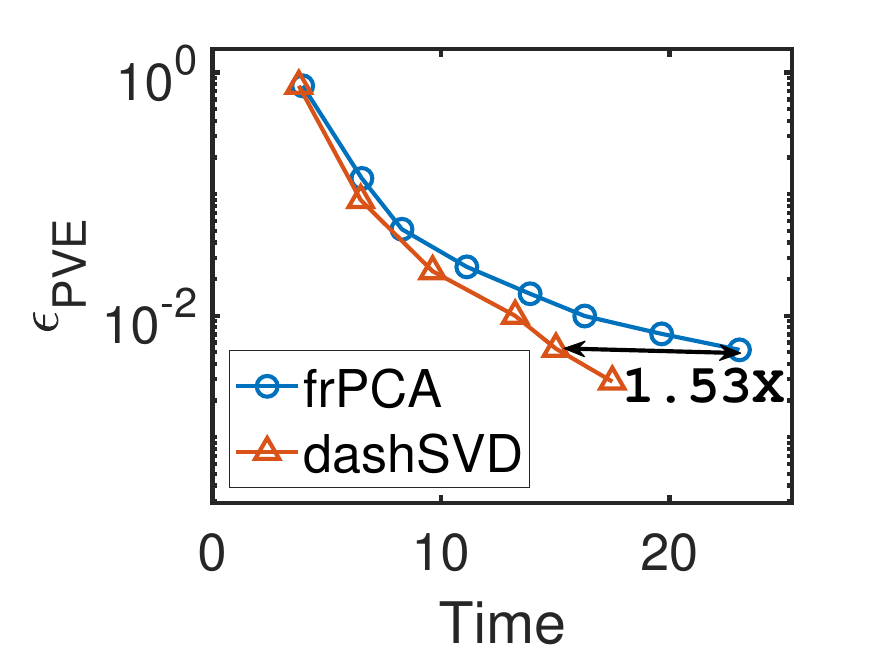} 
	}
	\subfigure[Aminer \atn{($\Delta p=1$)}] { \label{figfrpca:4} 
		\includegraphics[width=3.5cm,trim=15 5 20 10,clip]{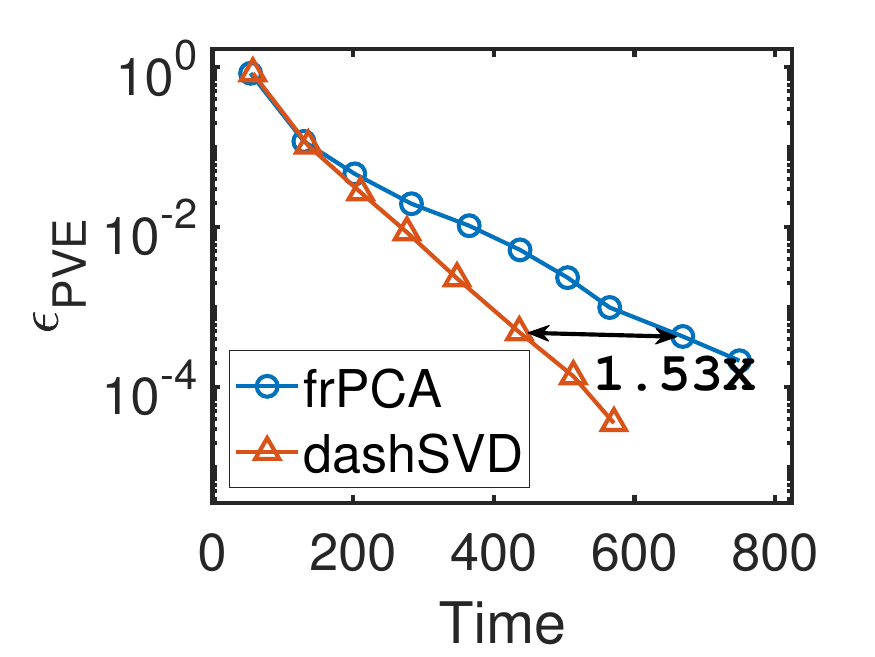} 
	}
	\subfigure[uk-2005 \atn{($\Delta p=1$)}] { \label{figfrpca:5} 
		\includegraphics[width=3.5cm,trim=15 5 20 10,clip]{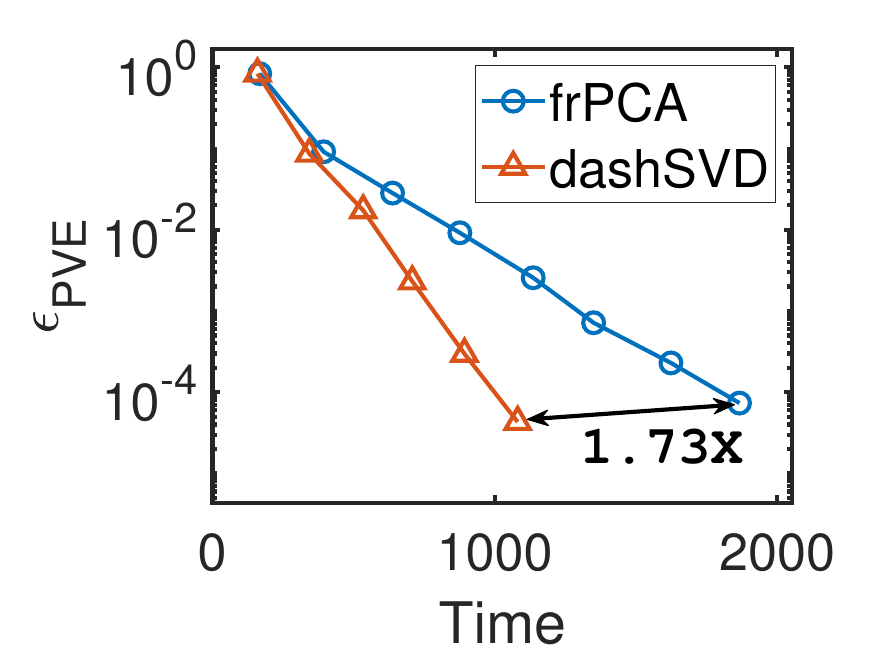} 
	}
	\subfigure[sk-2005 \atn{($\Delta p=1$)}] { \label{figfrpca:6} 
		\includegraphics[width=3.5cm,trim=15 5 20 10,clip]{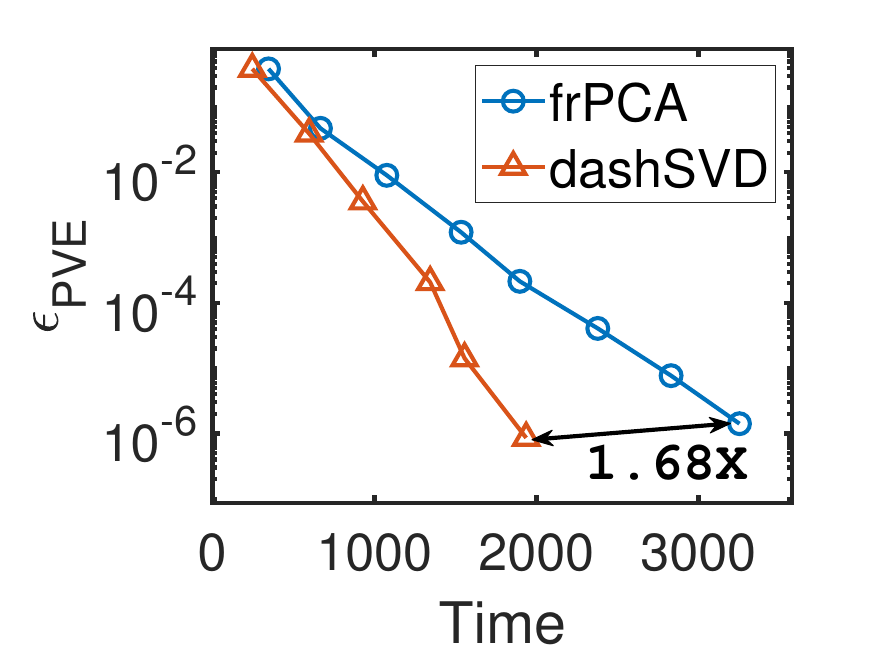} 
	}
	\caption{\notice{The $\epsilon_{\textrm{PVE}}$ vs. runtime curves of dashSVD (Alg. 4) and frPCA with the number of power iterations $p$ increased from 0 and in  a step $\Delta p$. The unit of time is second.}
	}
	\label{fig:frpca}
	\centering
\end{figure} 

\atnn{ Firstly, we present the comparison results of our dashSVD program in C with frPCA \citep{pmlr-v95-feng18a} to further validate the shifted power iteration for the sparse matrices. The both programs are parallely executed with 16 threads. 
 Because for frPCA the  number of power iterations ($p$) should be specified, we use the dashSVD without accuracy control, i.e. Alg. 4, in this experiment. It is actually an improved version of frPCA, as they use the same accelerating skills for handling sparse matrices.
By setting different power parameter $p$, we do the comprehensive investigation on the relationship of runtime and $\epsilon_{\textrm{PVE}}$ for the two algorithms when producing results with similar $\epsilon_{\textrm{PVE}}$.
 $p$ is increased with a step $\Delta p$.
 The results for the sparse matrices SNAP, MovieLens, Rucci1, Aminer, uk-2005 and sk-2005 are shown in Fig.~\ref{fig:frpca}. The curves in Fig.~\ref{fig:frpca} show that dashSVD can reach more accurate results during the same runtime compared with frPCA for all tested matrices, and the speed-up ratio is from 1.45X to 1.73X. Besides, the memory cost of dashSVD on SNAP, MovieLens, rucci1, Aminer, uk-2005 and sk-2005 is 0.30 GB, 0.96 GB, 4.66 GB, 31.5 GB, 147 GB and 200 GB, which is smaller than 0.35 GB, 0.98 GB, 4.69 GB, 31.5 GB, 162 GB and 220 GB of frPCA on these matrices.}

		\begin{figure}[t]
			\setlength{\abovecaptionskip}{0.1 cm}
			\setlength{\belowcaptionskip}{0.1 cm}
			\centering
			\subfigure[MovieLens] {
				\begin{minipage}{14cm}
					\centering
					\includegraphics[width=3.4cm, trim=103 265 115 273,clip]{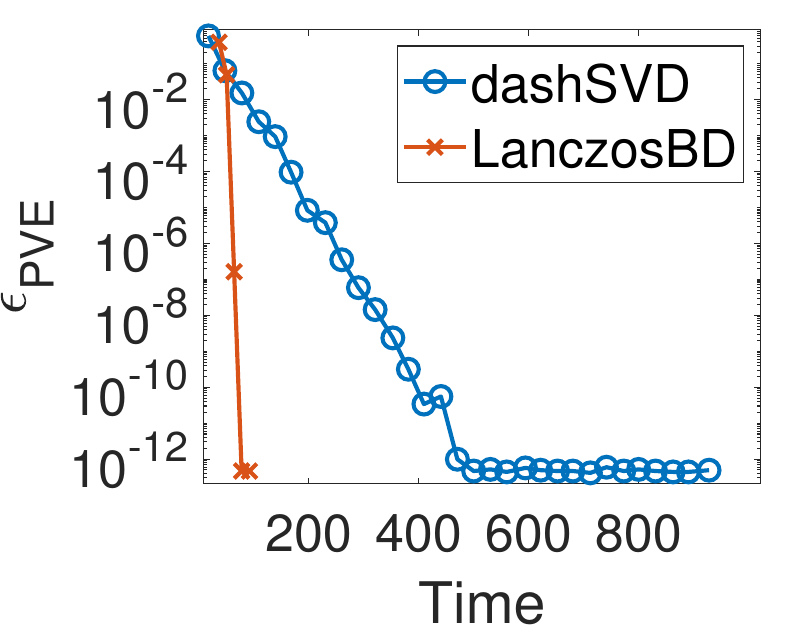}
					\includegraphics[width=3.4cm, trim=103 265 115 273,clip]{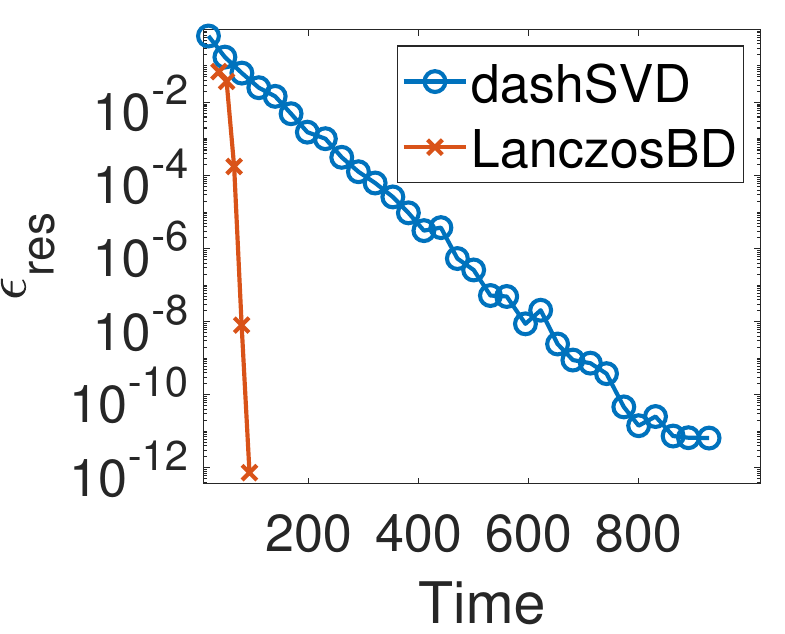}
					\includegraphics[width=3.4cm, trim=103 265 115 273,clip]{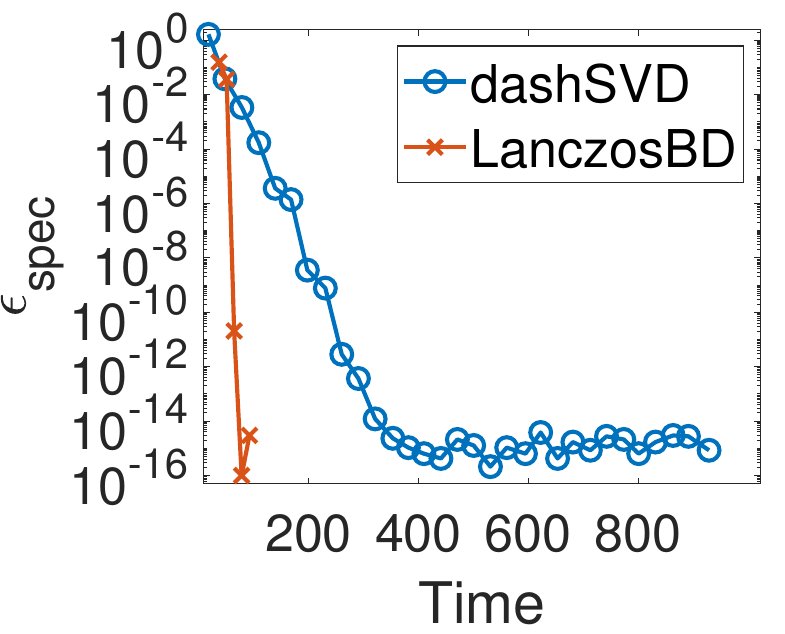}
					\includegraphics[width=3.4cm, trim=103 265 115 273,clip]{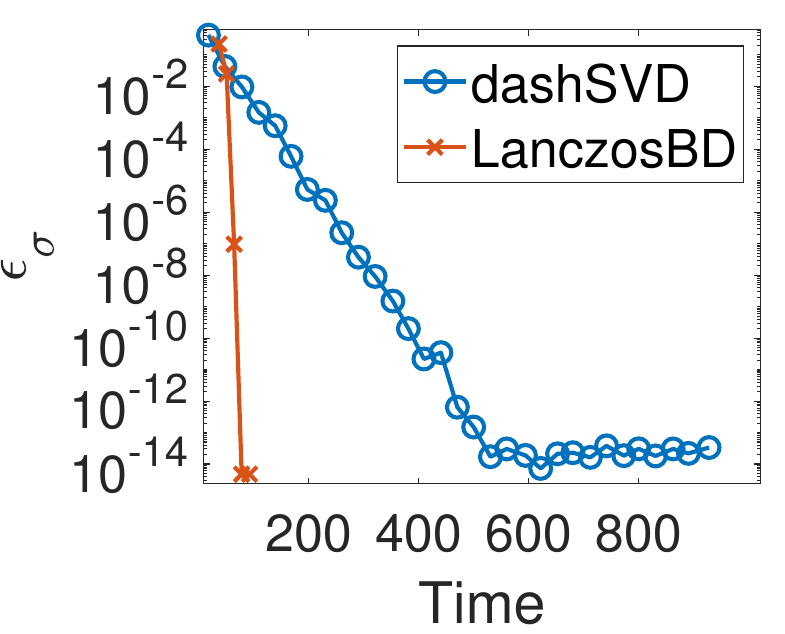}
				\end{minipage}
			}\\[-1ex]
			\subfigure[Rucci1] {
				\begin{minipage}{14cm}
					\centering
					\includegraphics[width=3.4cm, trim=103 265 115 273,clip]{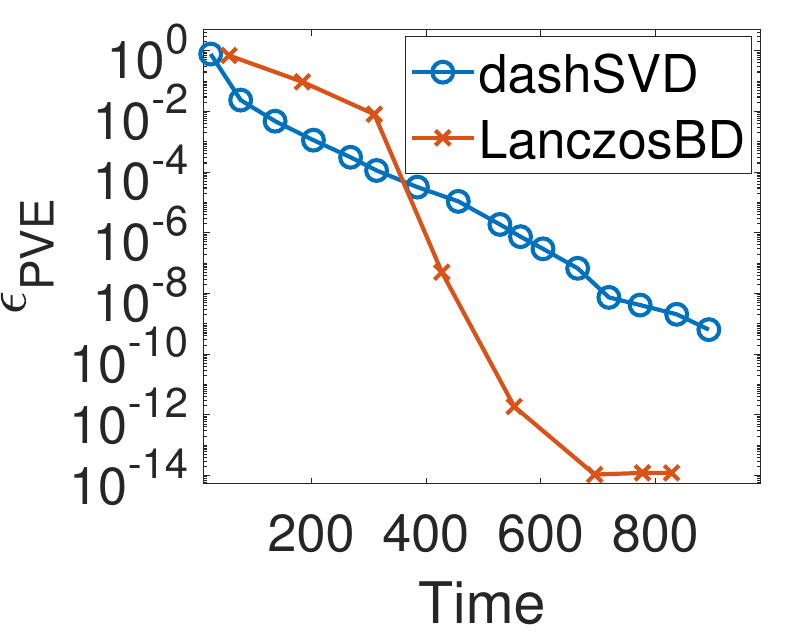}
					\includegraphics[width=3.4cm, trim=103 265 115 273,clip]{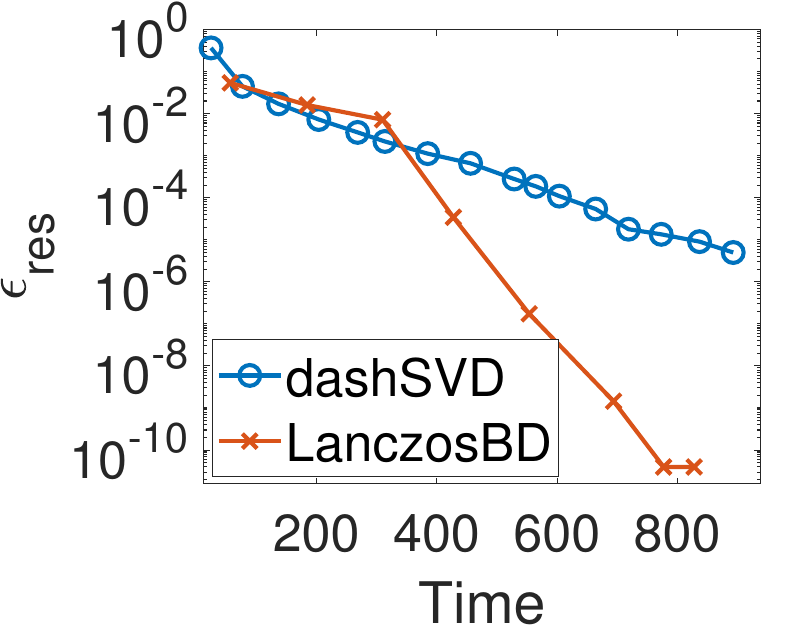}
					\includegraphics[width=3.4cm, trim=103 265 115 273,clip]{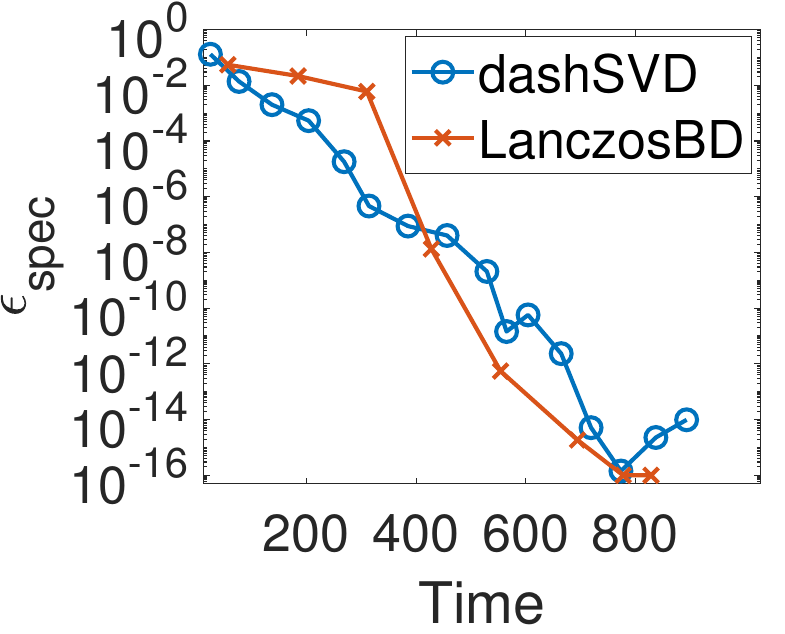}
					\includegraphics[width=3.4cm, trim=103 265 115 273,clip]{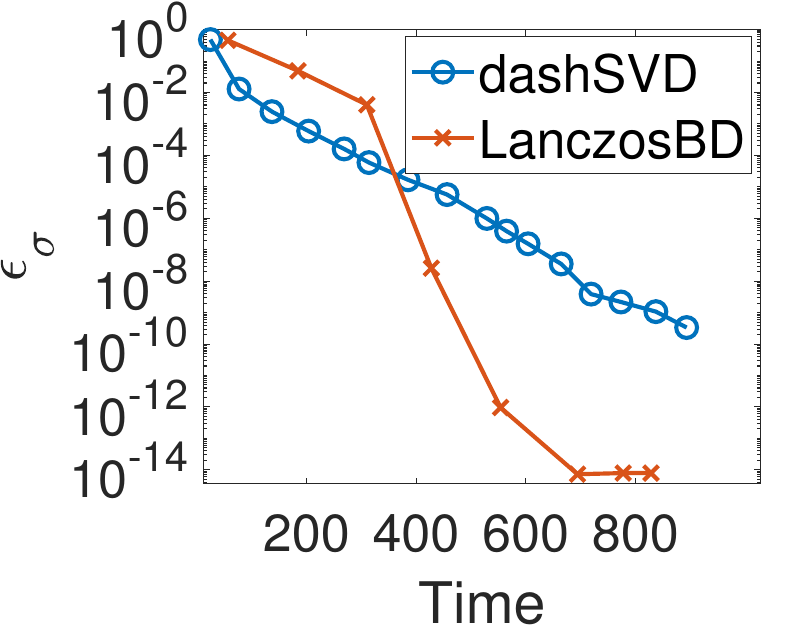}
				\end{minipage}
			}\\[-1ex]
			\subfigure[Aminer] {
				\begin{minipage}{14cm}
					\centering
					\includegraphics[width=3.4cm, trim=103 265 115 273,clip]{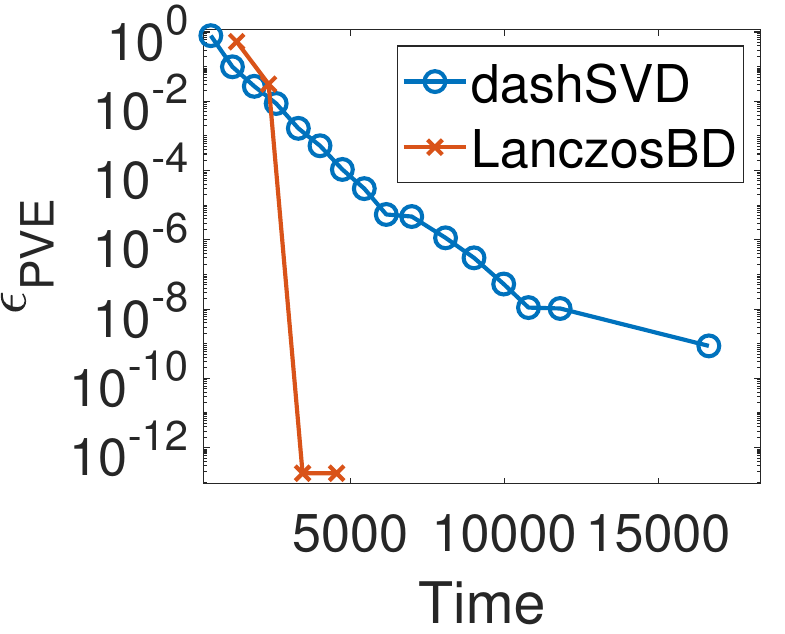}
					\includegraphics[width=3.4cm, trim=103 265 115 273,clip]{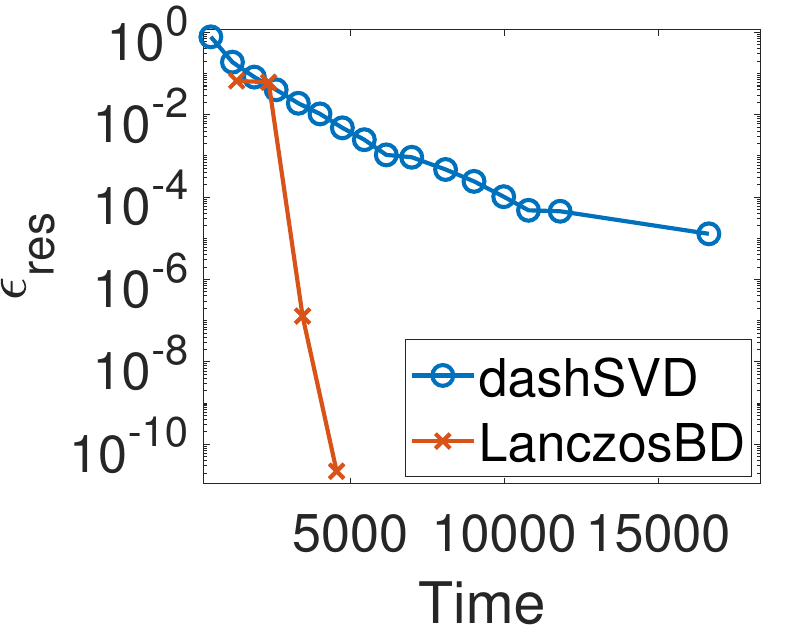}
					\includegraphics[width=3.4cm, trim=103 265 115 273,clip]{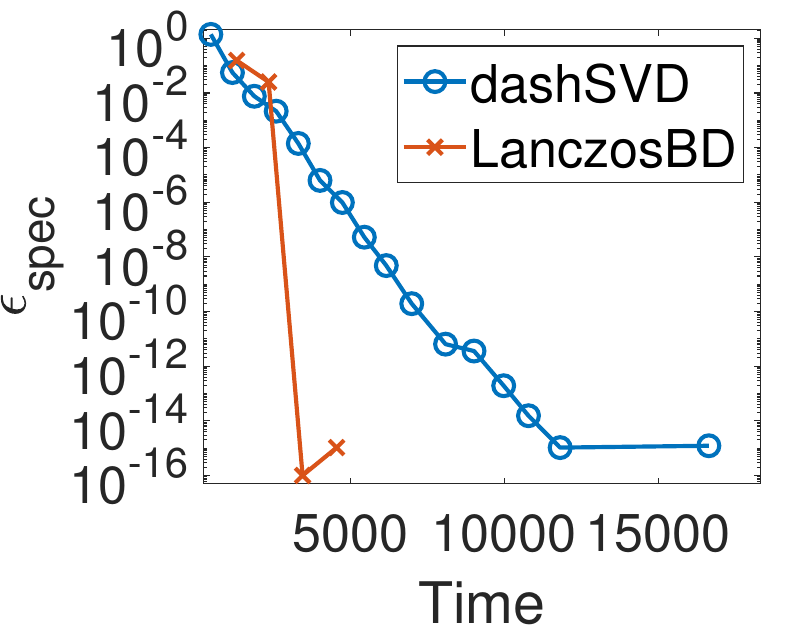}
					\includegraphics[width=3.4cm, trim=103 265 115 273,clip]{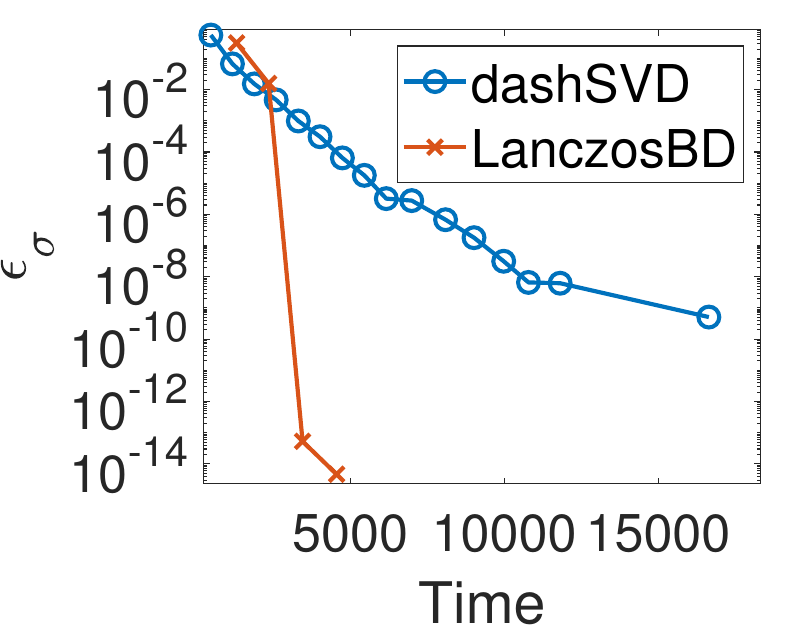}
				\end{minipage}
			}\\[-1ex]
			\subfigure[sk-2005] {
				\begin{minipage}{14cm}
					\centering
					\includegraphics[width=3.4cm, trim=103 265 115 273,clip]{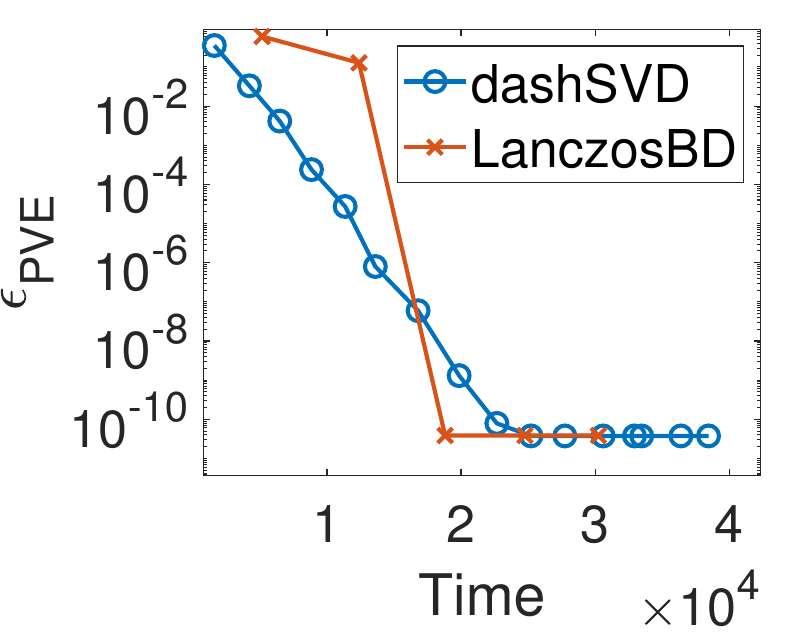}
					\includegraphics[width=3.4cm, trim=103 265 115 273,clip]{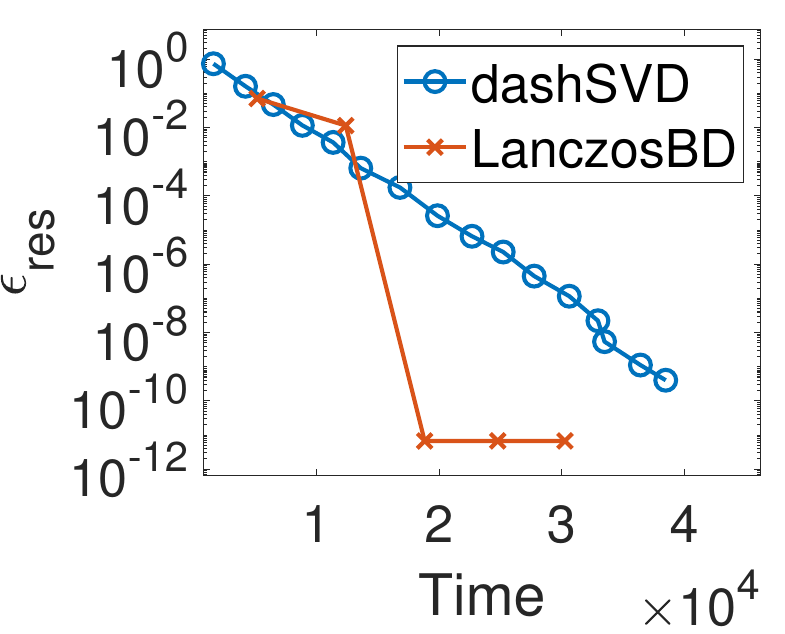}
					\includegraphics[width=3.4cm, trim=103 265 115 273,clip]{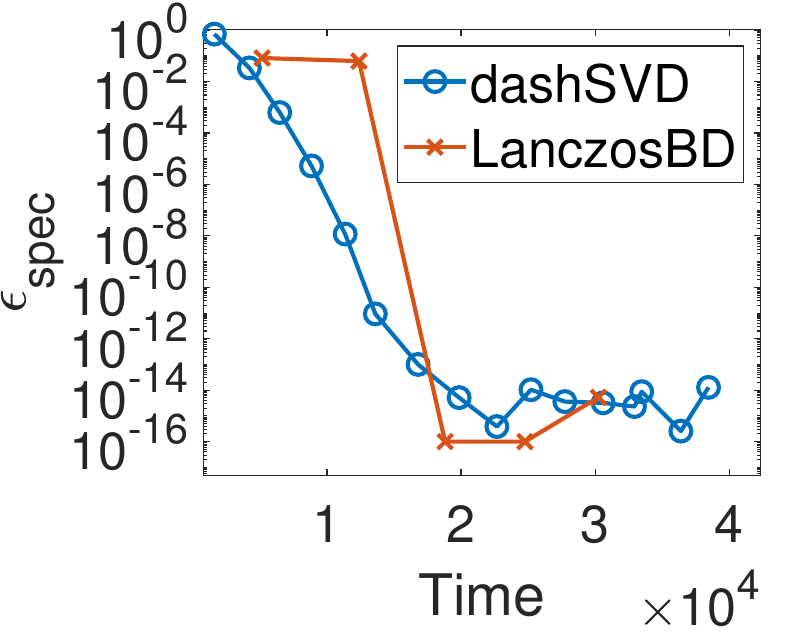}
					\includegraphics[width=3.4cm, trim=103 265 115 273,clip]{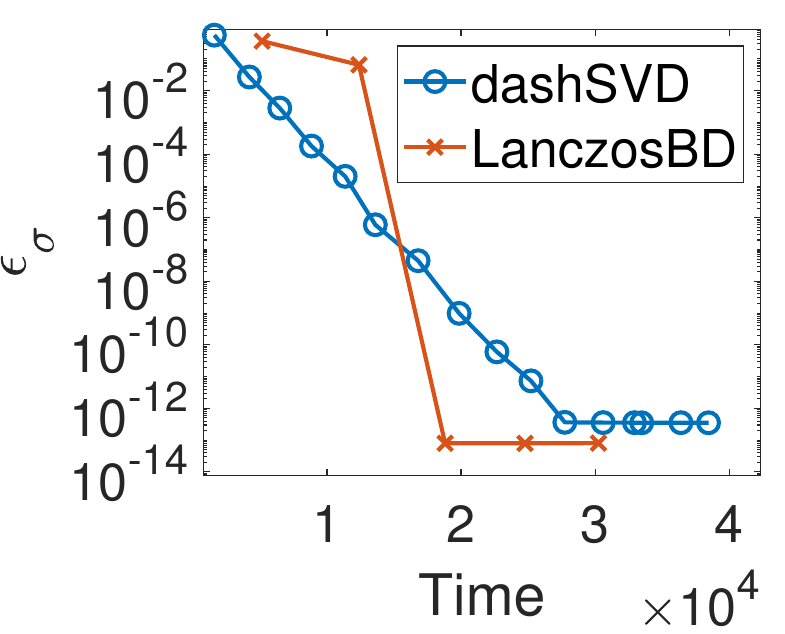}
				\end{minipage}
			}
			\caption{\notice{More error vs. time curves of dashSVD and \texttt{LanczosBD} in \texttt{svds} with single-thread computing ($k=100$). The unit of time is second. }}
			\label{fig:lbd_more}
			\centering
		\end{figure}

				The results of \texttt{LanczosBD} and dashSVD with single-thread computing for MovieLens, Rucci1, Aminer and sk-2005 are plotted in Fig.~\ref{fig:lbd_more}. \atnn{They} show that our algorithm always costs less runtime than \texttt{LanczosBD} when producing results with not high accuracy \atnn{in} error metrics $\epsilon_{\textrm{PVE}}$, $\epsilon_{\textrm{spec}}$ and $\epsilon_{\sigma}$. For the datasets with slower decaying trend of singular values (Rucci1 and sk-2005), our dashSVD gains more acceleration than other datasets. For the matrices with faster decaying trend of singular values, e.g. MovieLens and Aminer, \texttt{LanczosBD} performs better than dashSVD on the $\epsilon_{\textrm{res}}$ criterion.
				
				\begin{figure}[!t]
					\setlength{\abovecaptionskip}{0.01 cm}
					\setlength{\belowcaptionskip}{0.01 cm}
					\centering
					\subfigure[MovieLens] {
						\begin{minipage} {14cm}
							\centering
							\includegraphics[width=3.4cm, trim=103 265 115 273,clip]{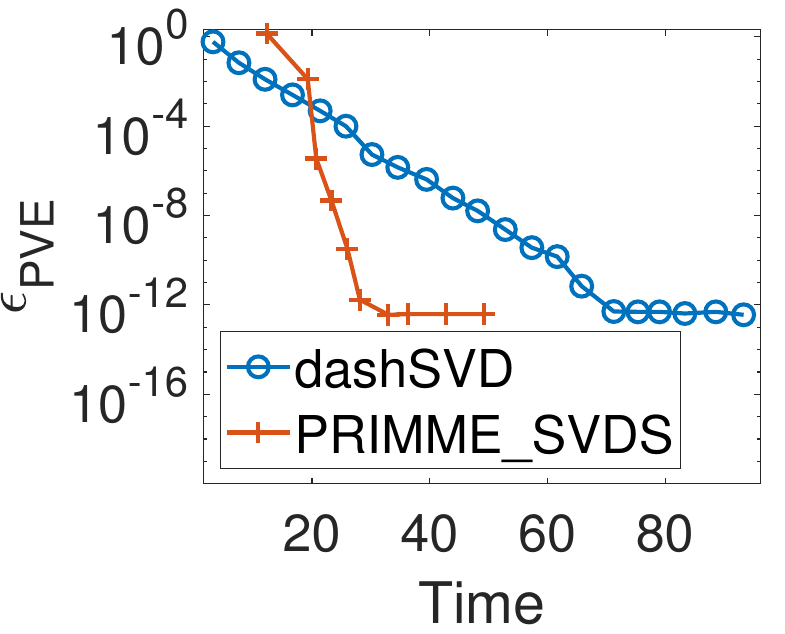} 
							\includegraphics[width=3.4cm, trim=103 265 115 273,clip]{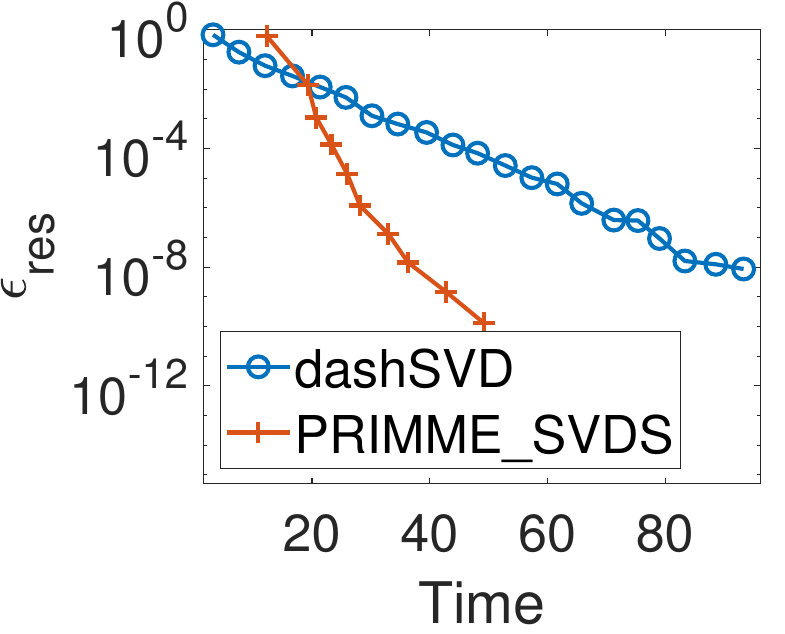} 
							\includegraphics[width=3.4cm, trim=103 265 115 273,clip]{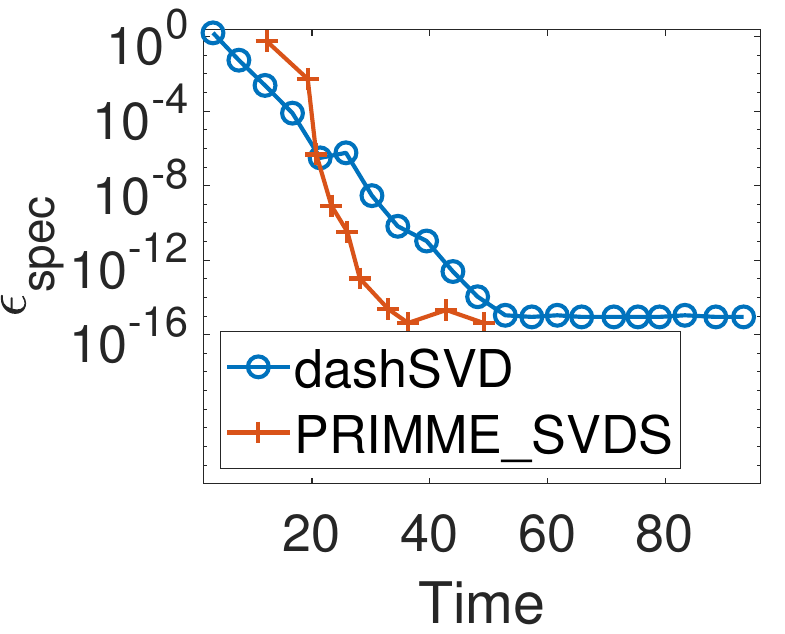} 
							\includegraphics[width=3.4cm, trim=103 265 115 273,clip]{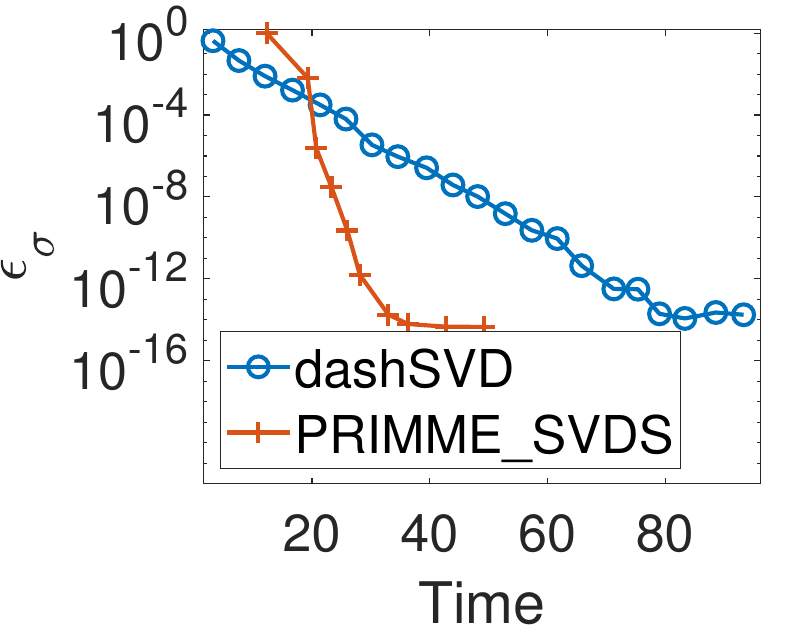} 
						\end{minipage}
					}\\[-1ex]
					\subfigure[Rucci1] {
						\begin{minipage} {14cm}
							\centering
							\includegraphics[width=3.4cm, trim=103 265 115 273,clip]{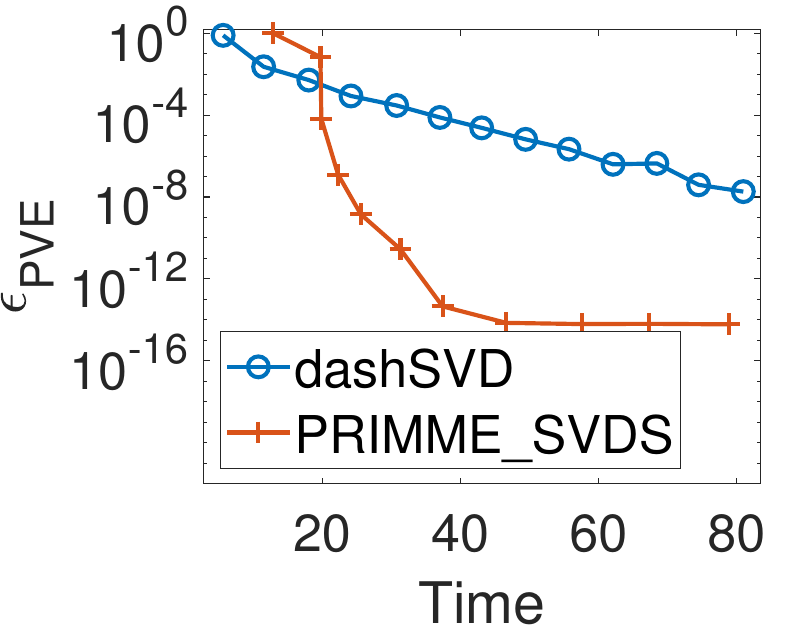} 
							\includegraphics[width=3.4cm, trim=103 265 115 273,clip]{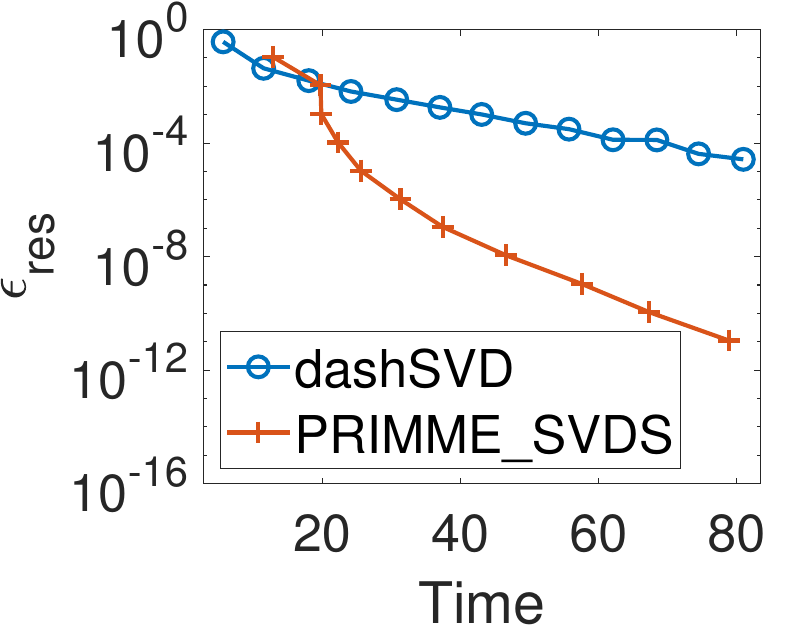} 
							\includegraphics[width=3.4cm, trim=103 265 115 273,clip]{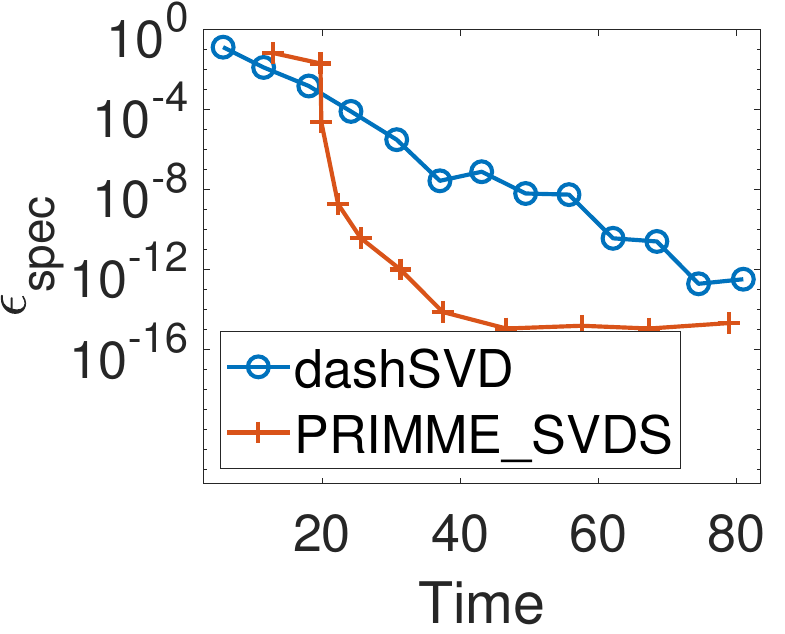} 
							\includegraphics[width=3.4cm, trim=103 265 115 273,clip]{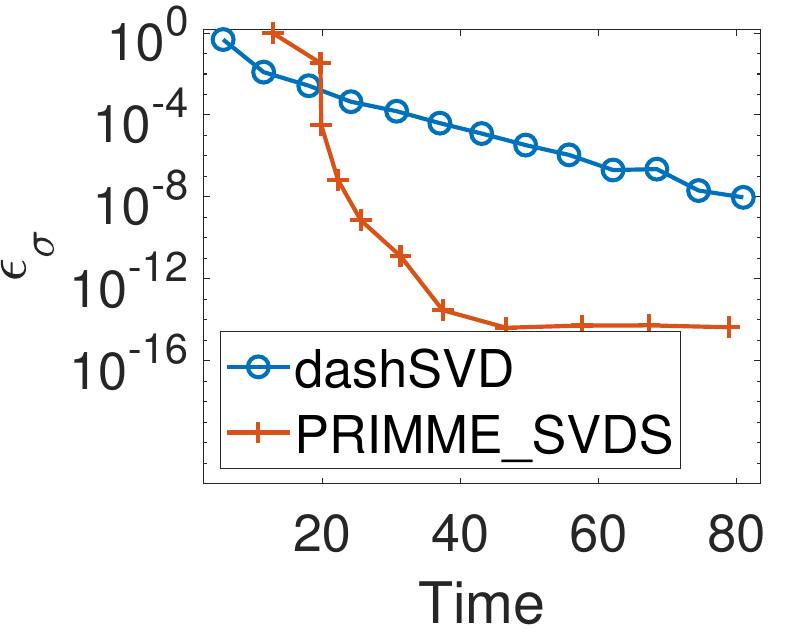} 
						\end{minipage}
					}\\[-1ex]
					\subfigure[Aminer] {
						\begin{minipage} {14cm}
							\centering
							\includegraphics[width=3.4cm, trim=103 265 115 273,clip]{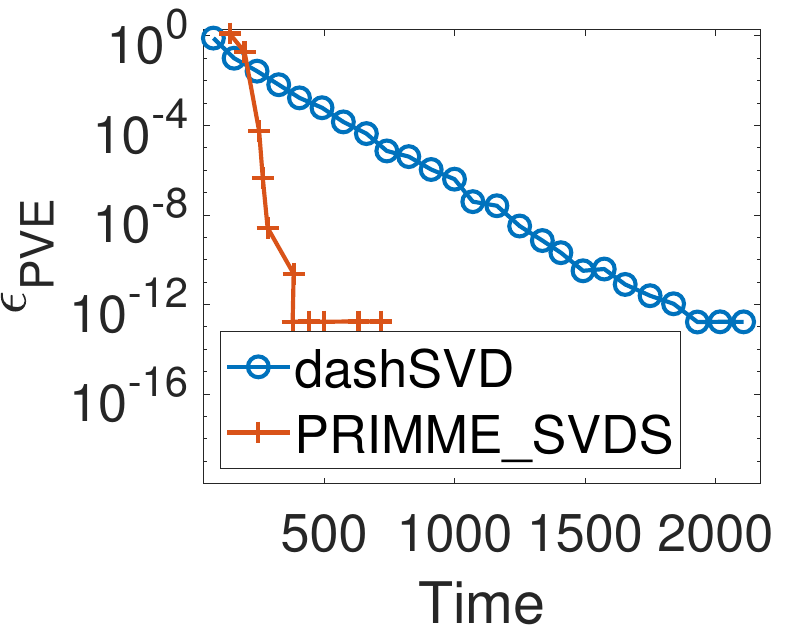} 
							\includegraphics[width=3.4cm, trim=103 265 115 273,clip]{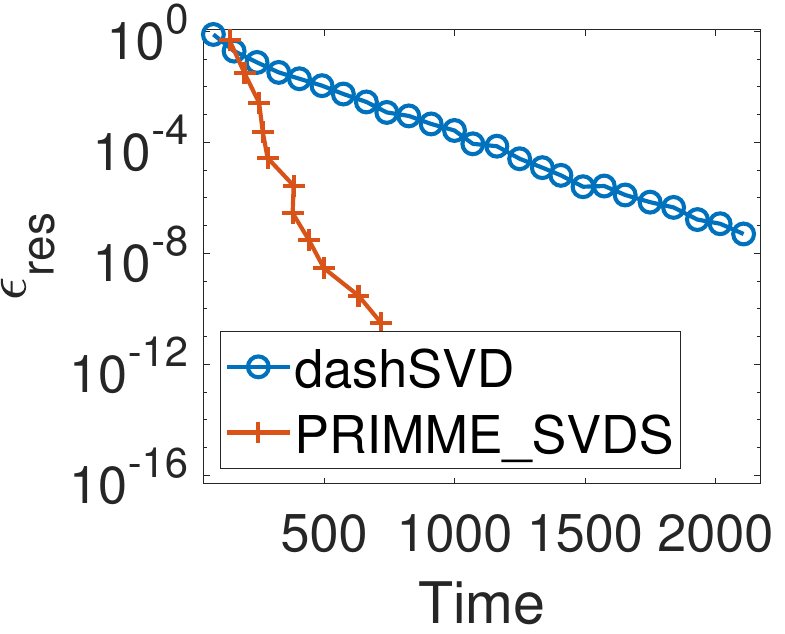} 
							\includegraphics[width=3.4cm, trim=103 265 115 273,clip]{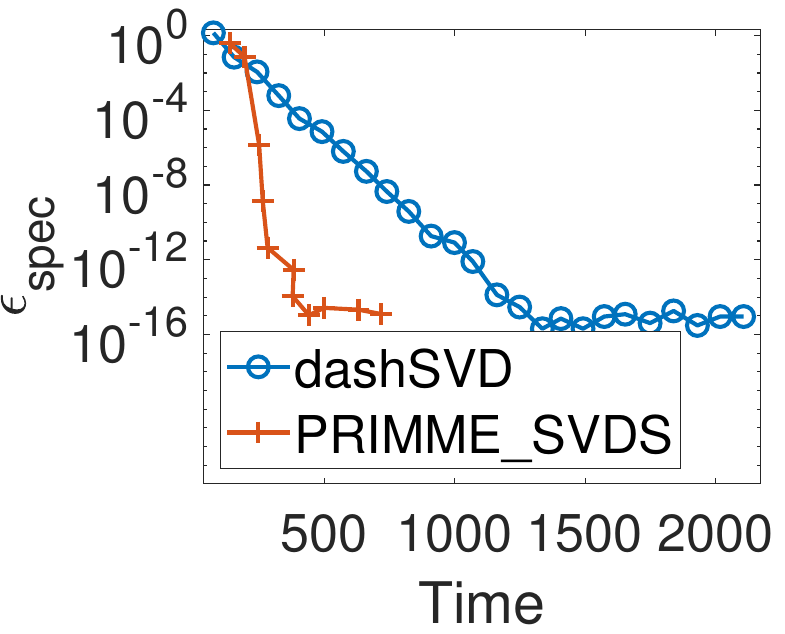} 
							\includegraphics[width=3.4cm, trim=103 265 115 273,clip]{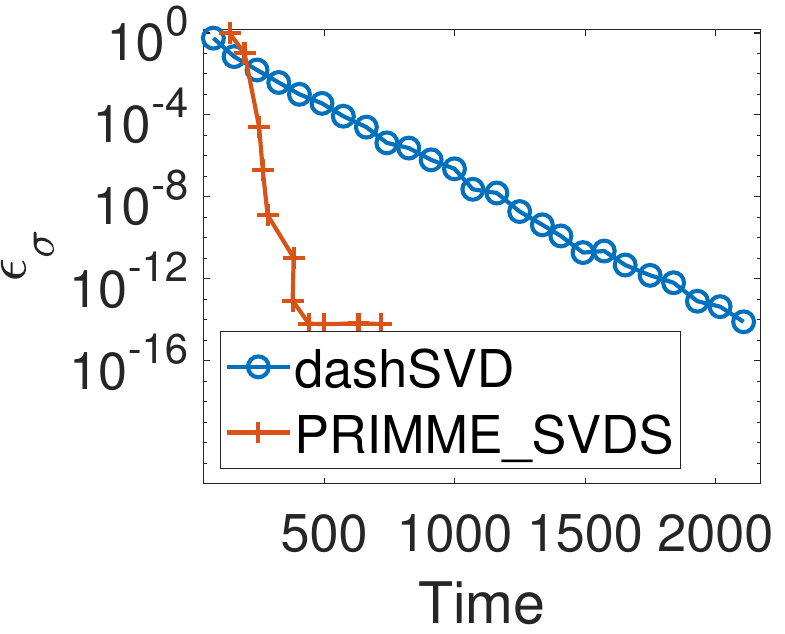} 
						\end{minipage}
					}\\[-1ex]
					\subfigure[sk-2005] { 
						\begin{minipage} {14cm}
							\centering
							\includegraphics[width=3.4cm, trim=103 265 115 273,clip]{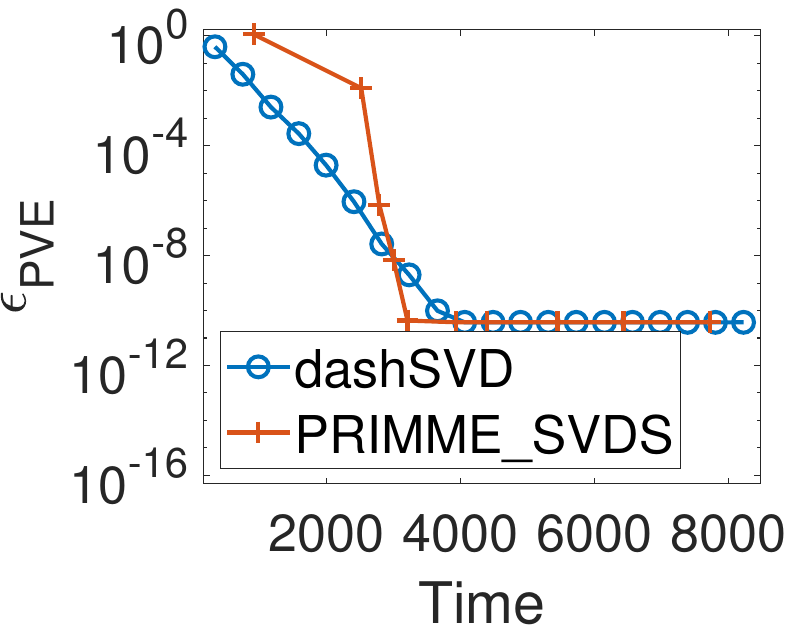} 
							\includegraphics[width=3.4cm, trim=103 265 115 273,clip]{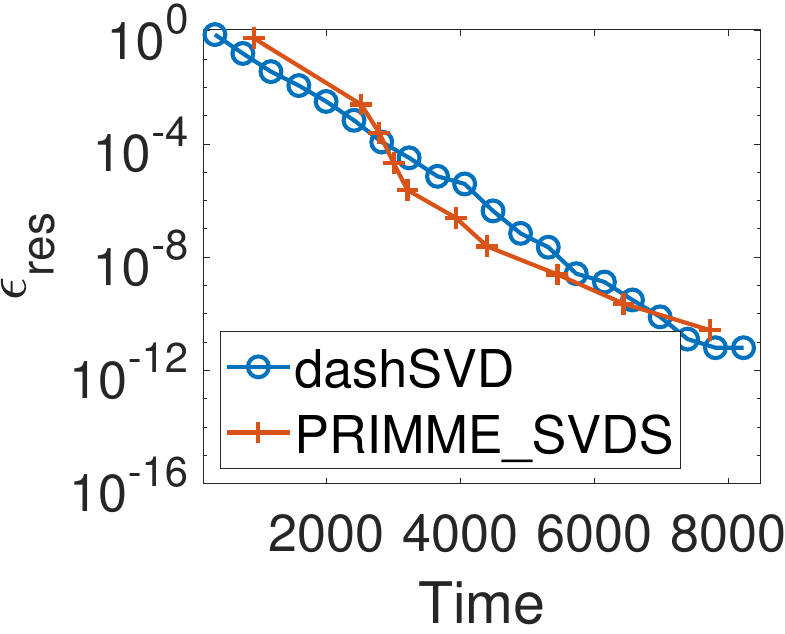} 
							\includegraphics[width=3.4cm, trim=103 265 115 273,clip]{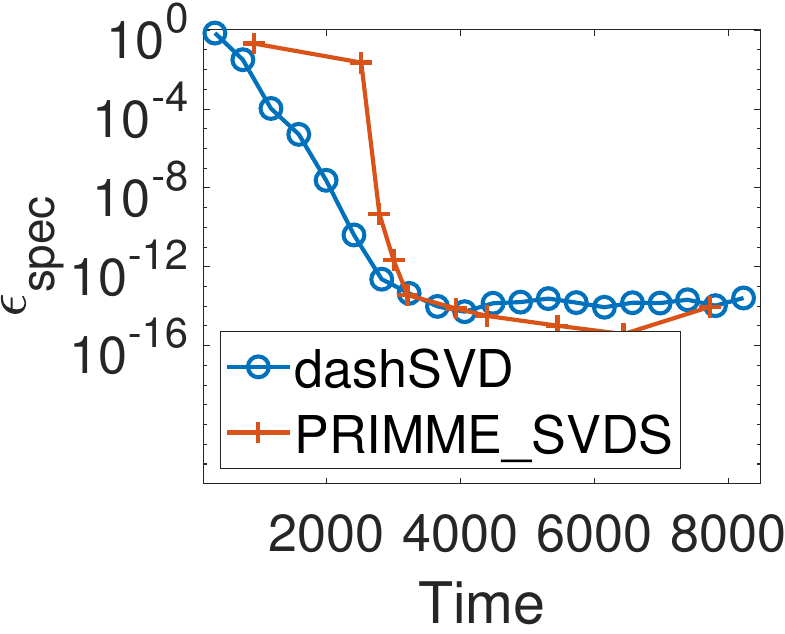} 
							\includegraphics[width=3.4cm, trim=103 265 115 273,clip]{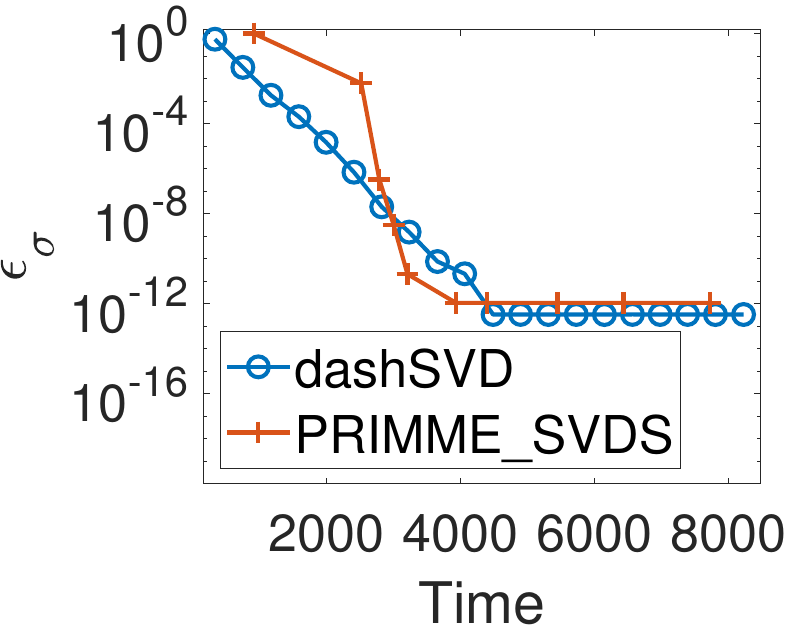} 
						\end{minipage}
					}\\[-1ex]	
					\caption{\notice{More error vs. runtime curves of dashSVD and PRIMME\_SVDS with 8-thread computing ($k=100$). The unit of time is second. }}
					\label{fig:primme_more}
					\centering
				\end{figure}

				The results of PRIMME\_SVDS and dashSVD with 8-thread computing for MovieLens, Rucci1, Aminer and sk-2005 are plotted in Fig.~\ref{fig:primme_more}. Fig.~\ref{fig:primme_more} shows \atn{that} our dashSVD costs less runtime than PRIMME\_SVDS when producing results with not high accuracy. 
				Because PRIMME\_SVDS firstly computes the eigenvalue decomposition of $\mathbf{A}^\mathrm{T}\mathbf{A}$, it performs better on the matrix when $m\gg n$, e.g. MovieLens, Rucci1 and Aminer. Therefore, the speed-up ratios of our dashSVD on the three matrices are smaller than those on SNAP, uk-2005 and sk-2005.

\begin{table}[h]
	\setlength{\abovecaptionskip}{0.01 cm}
	\setlength{\belowcaptionskip}{0.01 cm}
	\caption{Computational results of LazySVD, lansvd and dashSVD (tol=1e-2) with 16-thread computing ($k$=100). 
	}
	\label{table:lazy}
	\centering
	\small
	\begin{spacing}{0.8}
		\renewcommand{\multirowsetup}{\centering}
		\begin{tabular}{@{\,}c@{\,}c@{\,}c@{\,}c@{\,}c@{\,}c@{\,}c@{\,}c@{\,}c@{\,}c@{\,}c@{\,}c@{\,}c@{\,}c@{\,}} 
			\toprule
			\multirow{2}{*}{Matrix} & \multicolumn{3}{c}{LazySVD} & \multicolumn{3}{c}{lansvd} & \multicolumn{6}{c}{dashSVD} &\\
			\cmidrule(r){2-4}
			\cmidrule(r){5-7}
			\cmidrule(r){8-13}
			& Time\atnn{(s)} & $\epsilon_{\textrm{PVE}}$ & $\epsilon_{\sigma}$ & Time\atnn{(s)} & $\epsilon_{\textrm{PVE}}$ & $\epsilon_{\sigma}$ & Time\atnn{(s)} & Mem\atnn{(GB)} & $\epsilon_{\textrm{PVE}}$ & $\epsilon_{\sigma}$ & SP1 & SP2\\
			\midrule
			SNAP & 9.58 & 6.5E-3 & 3.2E-3 & 4.31 & NA$^1$ & 1.9E-2 & 1.93 & 0.30 & 5.7E-3 & 3.3E-3 & 5.0 & 2.2 \\
			MovieLens & 55.5 & 3.3E-3 & 1.6E-3 & 41.6 & 5.0E-2 & 2.6E-2 & 14.4 & 0.96 & 2.6E-3 & 1.7E-3 & 3.9 & 2.9 \\
			Rucci1 & 152 & 2.5E-2 & 1.2E-2 & 49.8 & NA$^1$ & 3.0E-2 & 13.6 & 4.66 & 1.9E-2 & 1.0E-2 & 11 & 3.7 \\
			Aminer & 1961 & 6.8E-3 & 3.4E-3 & 796 & NA$^1$ & 1.6E-2 & 341 & 31.5 & 3.1E-3 & 1.8E-3 & 5.8 & 2.3 \\
			uk-2005 & 7707 & 3.6E-3 & 1.7E-3 & NA$^2$ & NA$^2$ & NA$^2$ & 877 & 147 & 2.4E-3 & 1.5E-3 & 8.8 & NA$^2$\\
			sk-2005 & 13634 & 1.5E-3 & 7.4E-4 & NA$^2$ & NA$^2$ & NA$^2$ & 1385 & 200 & 8.4E-4 & 6.2E-4 & 9.8 & NA$^2$ \\
			\bottomrule 
		\end{tabular}
	\end{spacing}
	\vspace{-0.5pt}
	\begin{tablenotes}
		\item[] Mem \atn{represents} the memory cost. SP1 and SP2 represent the speed-up ratio from dashSVD to LazySVD and lansvd. NA$^1$ represents the singular vectors computed by lansvd contains NaN, which makes some metrics of error criterion unavailable. NA$^2$ represents program breaks down due to segment fault. 
	\end{tablenotes}
\end{table}	
		\begin{figure}[h]
			\setlength{\abovecaptionskip}{0.1 cm}
			\setlength{\belowcaptionskip}{0.1 cm}
			\centering
			\subfigure[\atn{The error vs. time curves of dashSVD and \texttt{LanczosBD} in \texttt{svds} with single-thread computing ($k=50$)}] {
				\label{fig:err_snap_50_1} 
				\begin{minipage}{14cm}
					\centering
					\includegraphics[width=3.4cm, trim=103 265 115 273,clip]{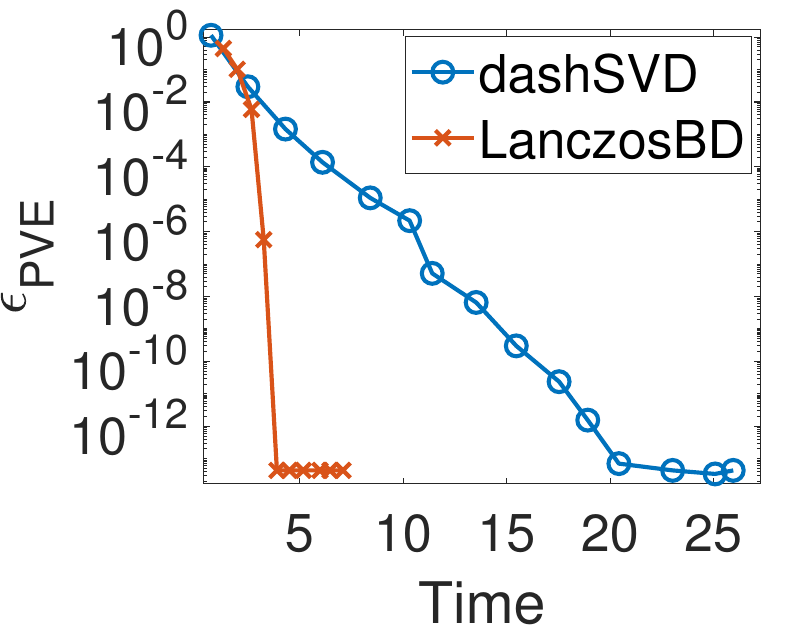}
					\includegraphics[width=3.4cm, trim=103 265 115 273,clip]{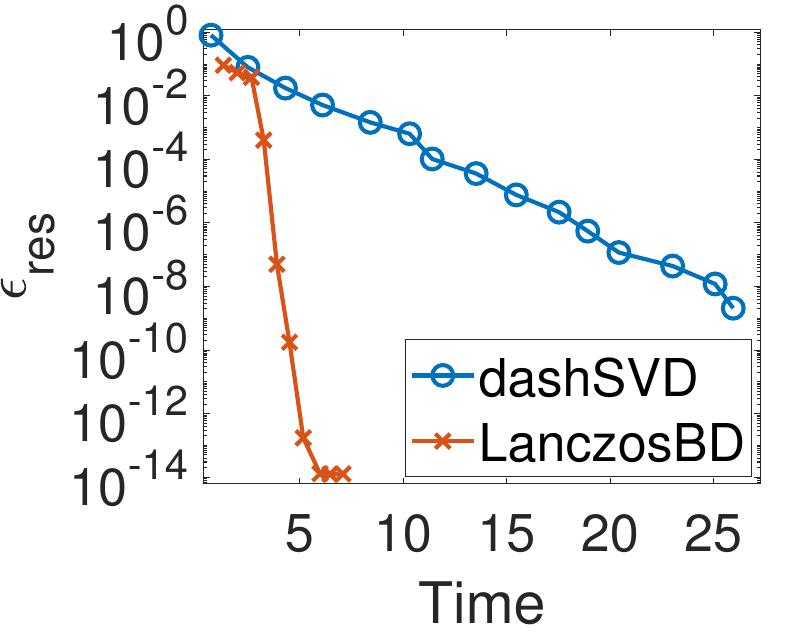}
					\includegraphics[width=3.4cm, trim=103 265 115 273,clip]{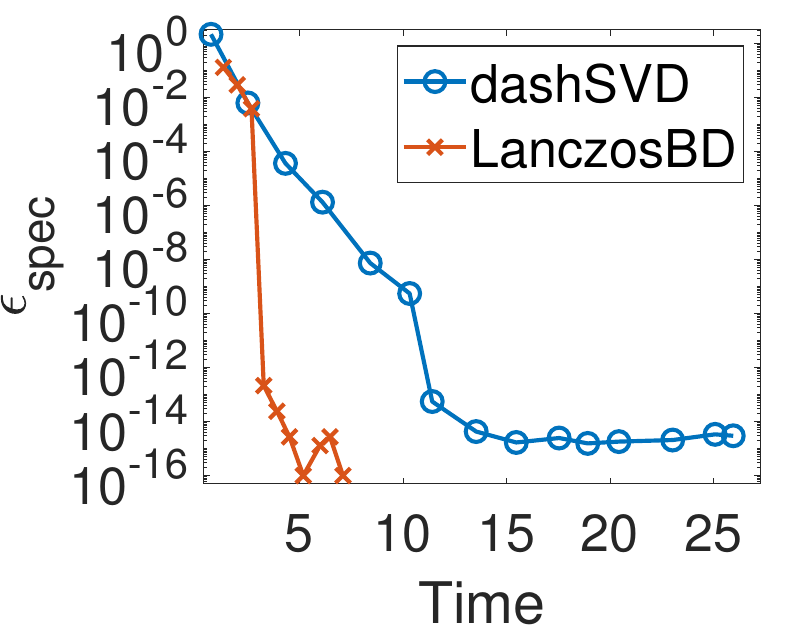}
					\includegraphics[width=3.4cm, trim=103 265 115 273,clip]{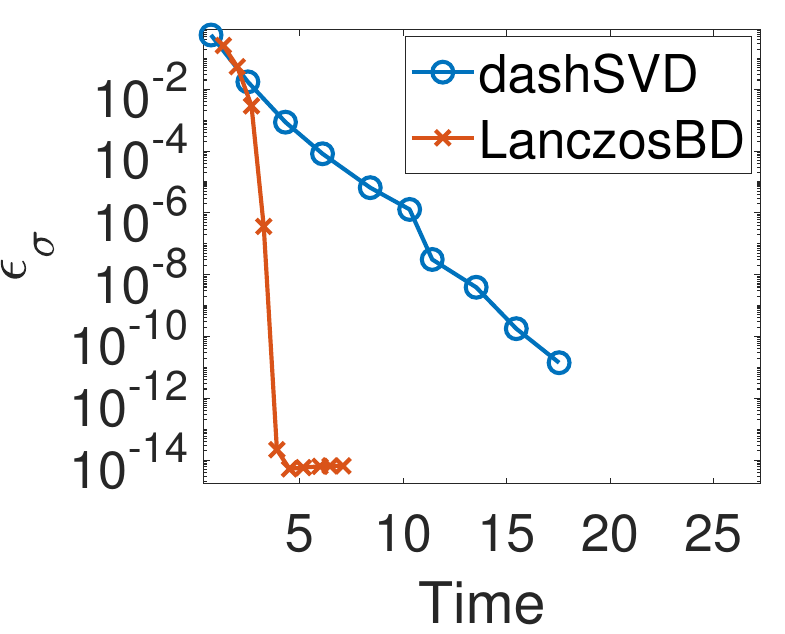}
				\end{minipage}
			}\\[-1ex]
			\subfigure[\atn{The error vs. time curves of dashSVD and PRIMME\_SVDS with 8-thread computing ($k=50$)}] {
				\label{fig:err_snap_50_2} 
				\begin{minipage} {14cm}
					\centering
					\includegraphics[width=3.4cm, trim=103 265 115 273,clip]{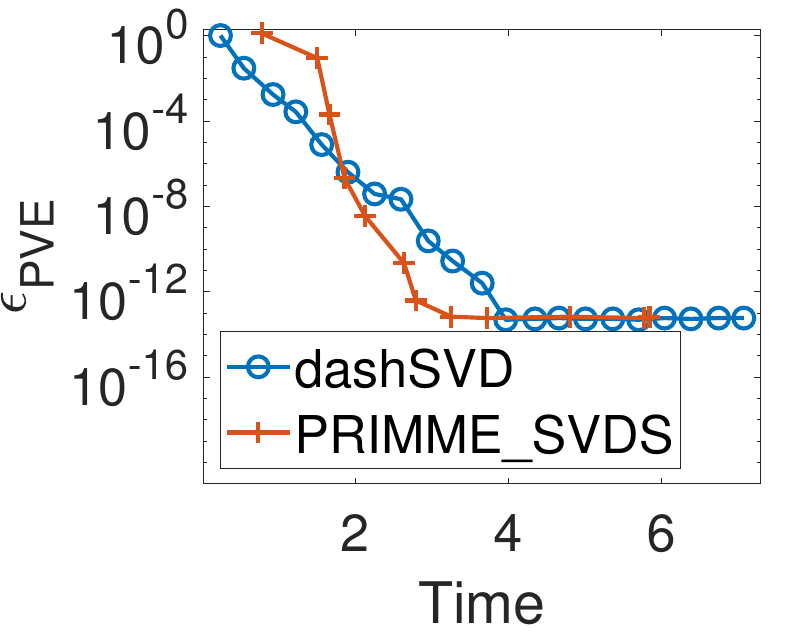} 
					\includegraphics[width=3.4cm, trim=103 265 115 273,clip]{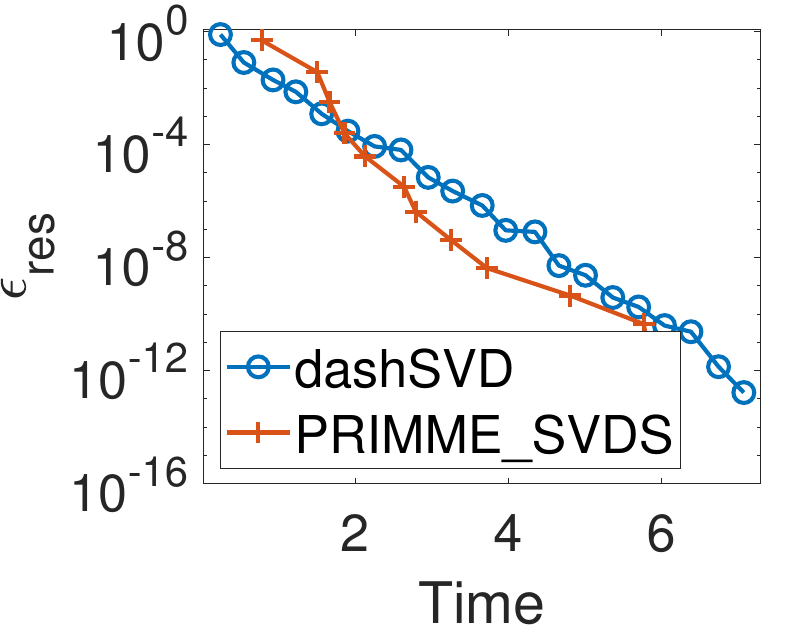} 
					\includegraphics[width=3.4cm, trim=103 265 115 273,clip]{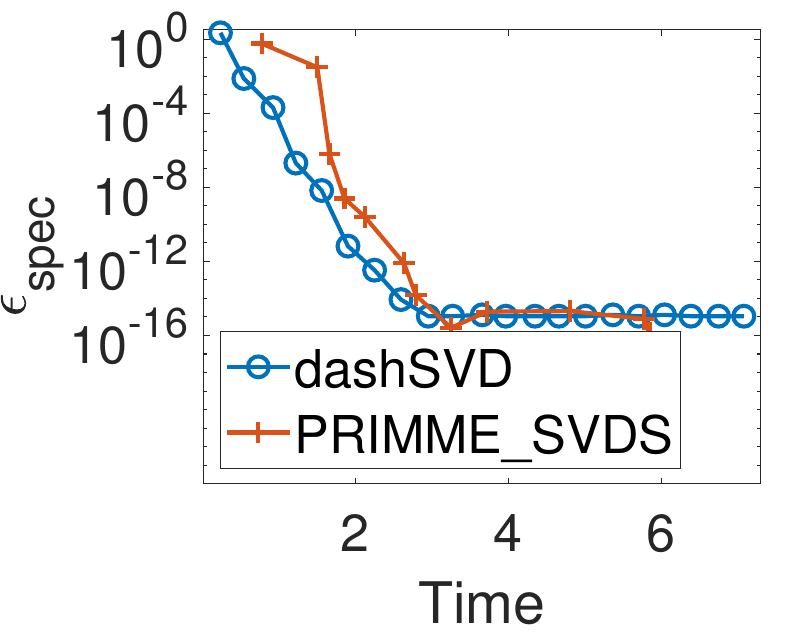} 
					\includegraphics[width=3.4cm, trim=103 265 115 273,clip]{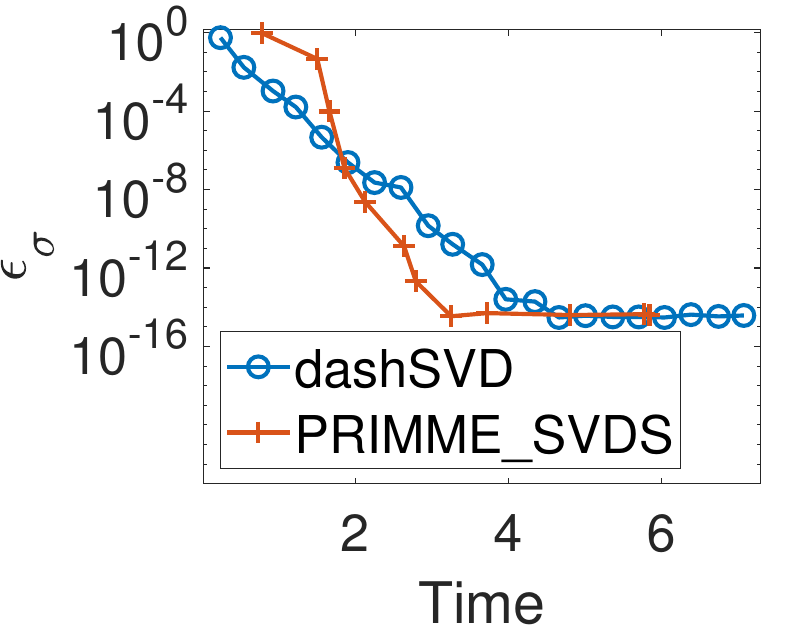} 
				\end{minipage}
			}\\[-1ex]
                \subfigure[\atn{The error vs. time curves of dashSVD and \texttt{LanczosBD} in \texttt{svds} with single-thread computing ($k=200$)}] {
				\label{fig:err_snap_200_1} 
				\begin{minipage}{14cm}
					\centering
					\includegraphics[width=3.4cm, trim=103 265 115 273,clip]{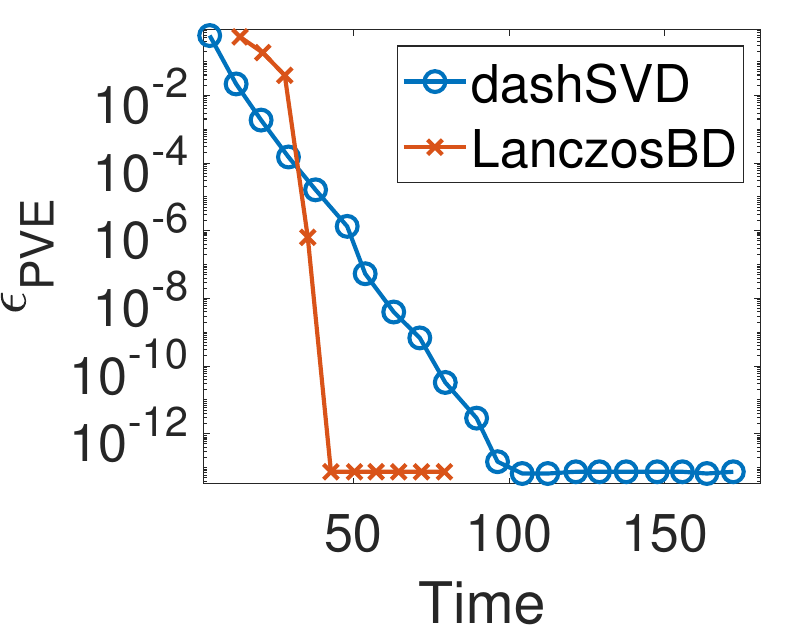}
					\includegraphics[width=3.4cm, trim=103 265 115 273,clip]{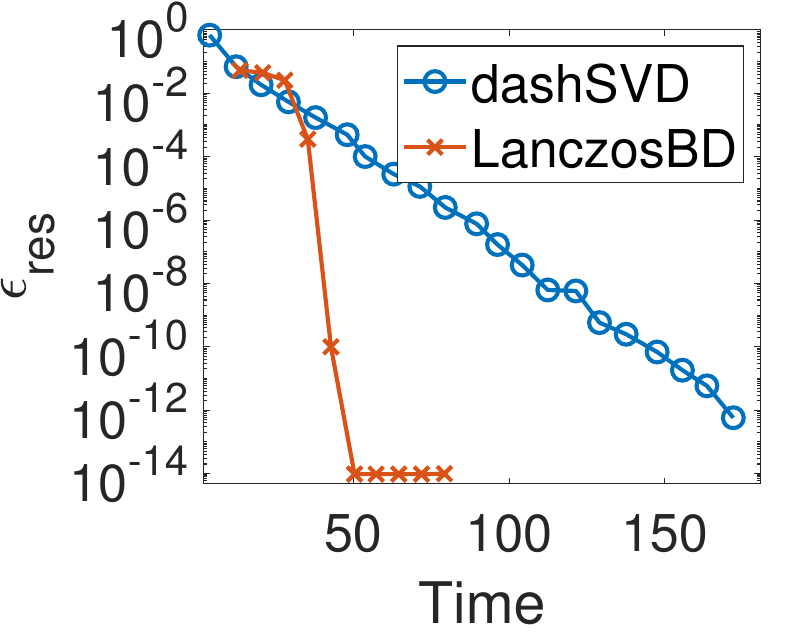}
					\includegraphics[width=3.4cm, trim=103 265 115 273,clip]{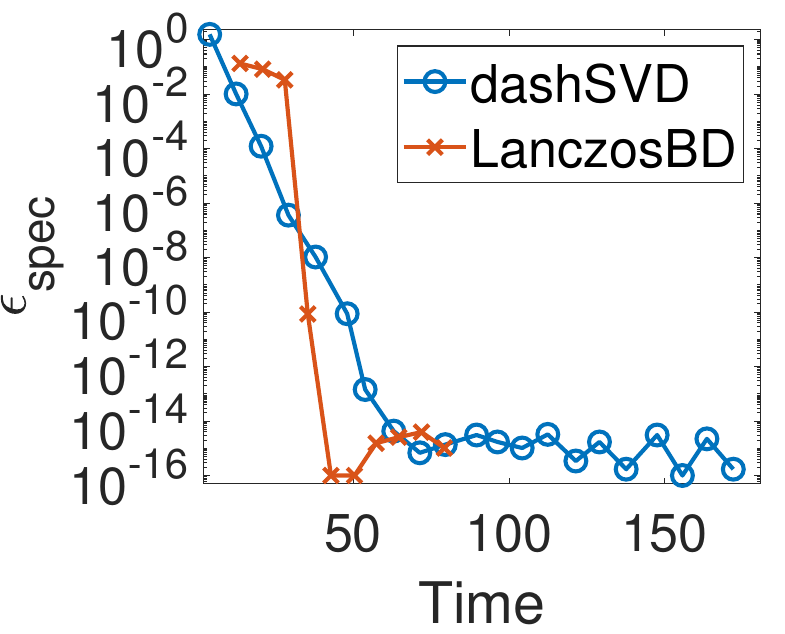}
					\includegraphics[width=3.4cm, trim=103 265 115 273,clip]{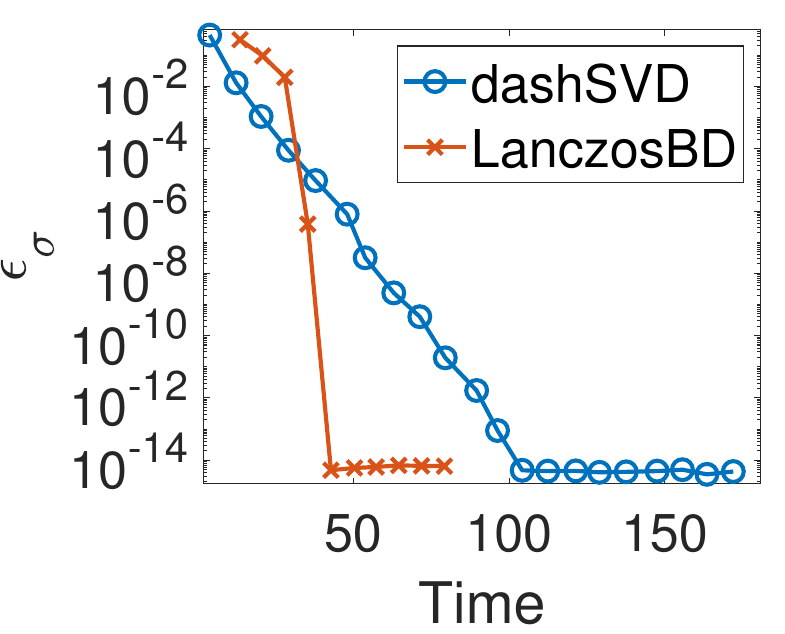}
				\end{minipage}
			}\\[-1ex]
			\subfigure[\atn{The error vs. time curves of dashSVD and PRIMME\_SVDS with 8-thread computing ($k=200$)}] {
				\label{fig:err_snap_200_2} 
				\begin{minipage} {14cm}
					\centering
					\includegraphics[width=3.4cm, trim=103 265 115 273,clip]{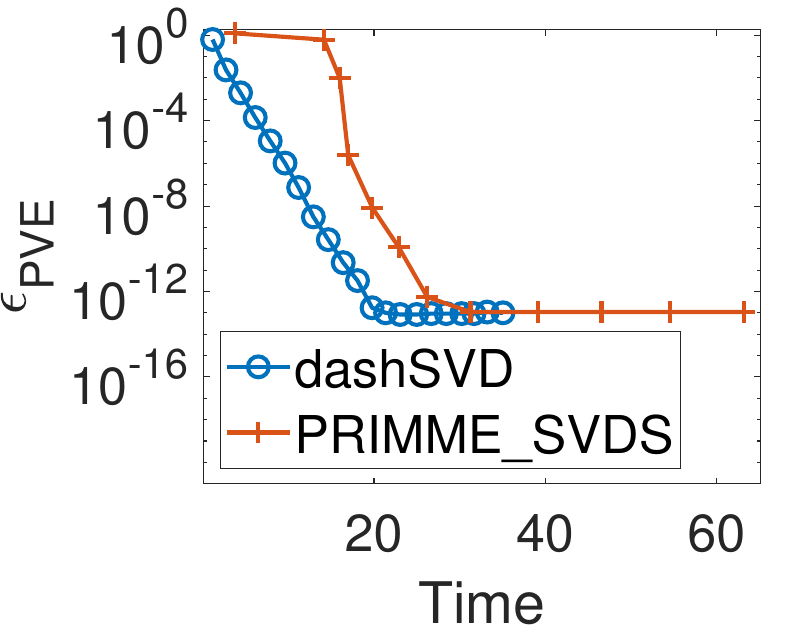} 
					\includegraphics[width=3.4cm, trim=103 265 115 273,clip]{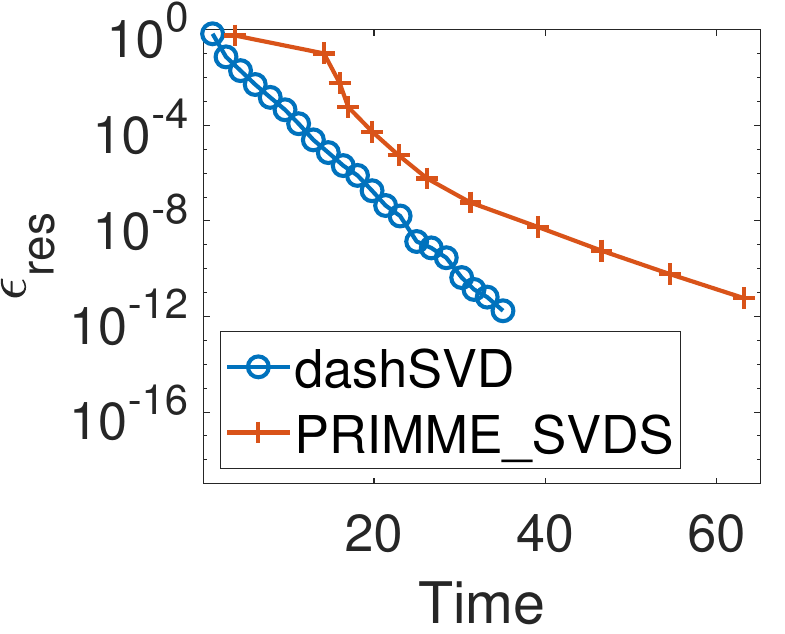} 
					\includegraphics[width=3.4cm, trim=103 265 115 273,clip]{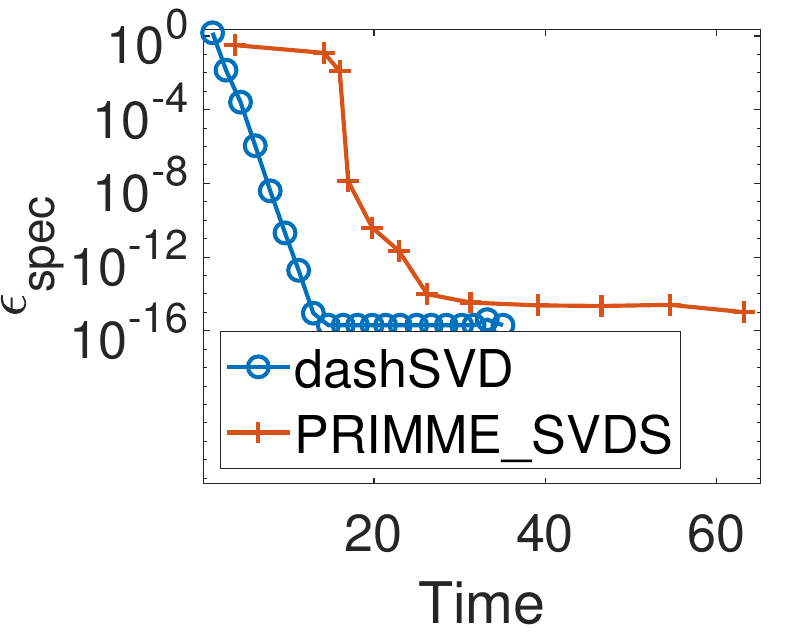} 
					\includegraphics[width=3.4cm, trim=103 265 115 273,clip]{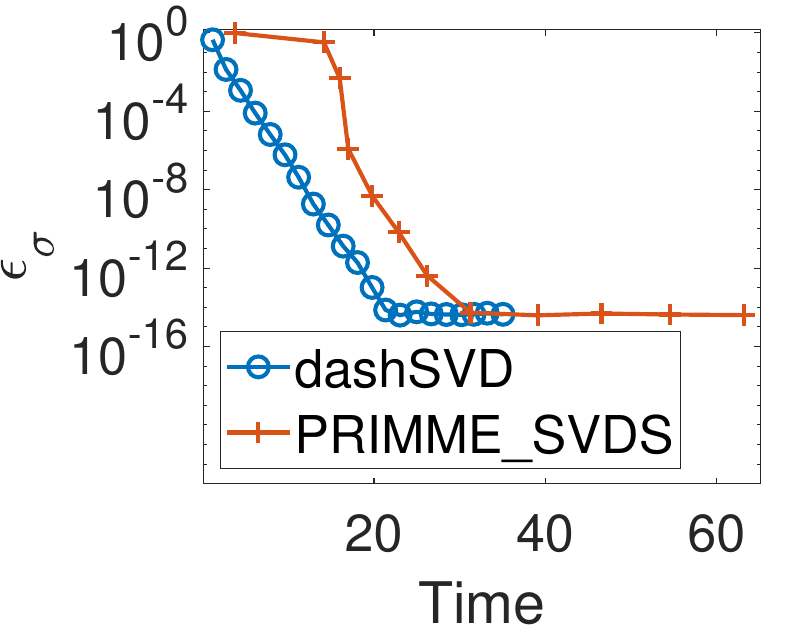} 
				\end{minipage}
			}
			\caption{\notice{The error vs. time curves of dashSVD compared with \texttt{LanczosBD} in \texttt{svds} and PRIMME\_SVDS for SNAP ($k=50$ and $200$). The unit of time is second. }}
			\label{fig:err_snap_vark}
			\centering
		\end{figure}

Then, the computational results with similar accuracy of LazySVD, lansvd and dashSVD \atn{with 16 threads} are listed in Table \ref{table:lazy}. For fair comparison, LazySVD is also implemented in C with MKL and the subspace dimensionality of lansvd is fixed to $1.5k$. We varies the $tol$ in lansvd to get the results with similar accuracy compared with dashSVD. From Table~\ref{table:lazy}, we see that dashSVD is \textbf{3.9X} to \textbf{11X} faster than LazySVD and \textbf{2.2X} to \textbf{3.7X} faster than lansvd for achieving similar accuracy. Besides, LazySVD and lansvd costs less memory than our dashSVD when the dimensionalities of subspace are the same. Notice that lansvd may produce singular vectors with NaN, which implies lansvd is not robust for computing the truncated SVD of large datasets.

\atn{Besides, we set $k=50,200$ on SNAP dataset to validate the efficiency of dashSVD with different rank parameters. The error vs. runtime curves are plotted in Fig.~\ref{fig:err_snap_vark}. The results of dashSVD compared with PRIMME\_SVDS and \texttt{LanczosBD} show that dashSVD performs better than the competitors when $k=200$, rather than when $k=50$ and $k=100$. This implies that the advantage of dashSVD would increase when $k$ increases. Besides, the memory cost of dashSVD, \texttt{LanczosBD} and PRIMME\_SVDS   are 0.16 GB, 0.17 GB, 0.37 GB  respectively for $k=50$, and 0.59 GB, 0.59 GB and 0.74 GB respectively for $k=200$.}

\atn{Finally, to show the variance of dashSVD when computing truncated SVD, we run dashSVD on SNAP with $k=100$ 
for 100 times for each $p=0,4,8,12,16,20$. The whisker plots of the four error metrics are shown in Fig.~\ref{fig:err_snap_var}. It shows that when computing low-accuracy SVD (with $p=0, 4, 8$), dashSVD produces results with little variance. For these cases, the standard deviations of most error metrics are no more than $5\%$ of mean value.}

\begin{figure}[h]
			\setlength{\abovecaptionskip}{0.1 cm}
			\setlength{\belowcaptionskip}{0.1 cm}
			\centering 
				\begin{minipage}{14cm}
					\centering
					\includegraphics[width=3.4cm, trim=103 265 115 273,clip]{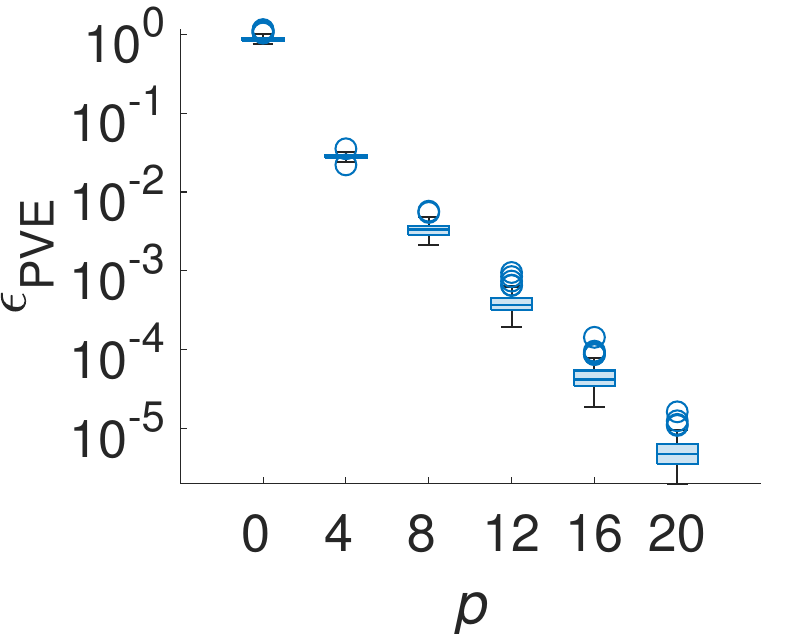}
					\includegraphics[width=3.4cm, trim=103 265 115 273,clip]{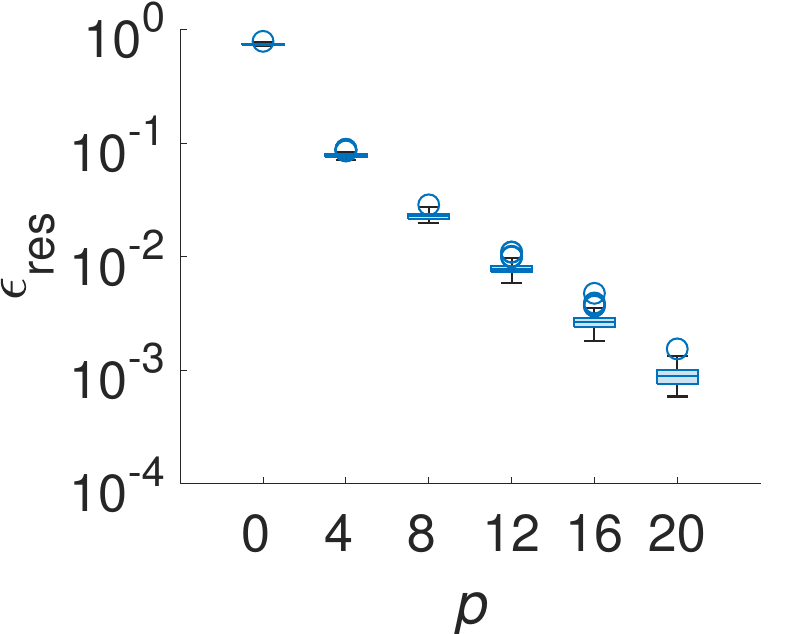}
					\includegraphics[width=3.4cm, trim=103 265 115 273,clip]{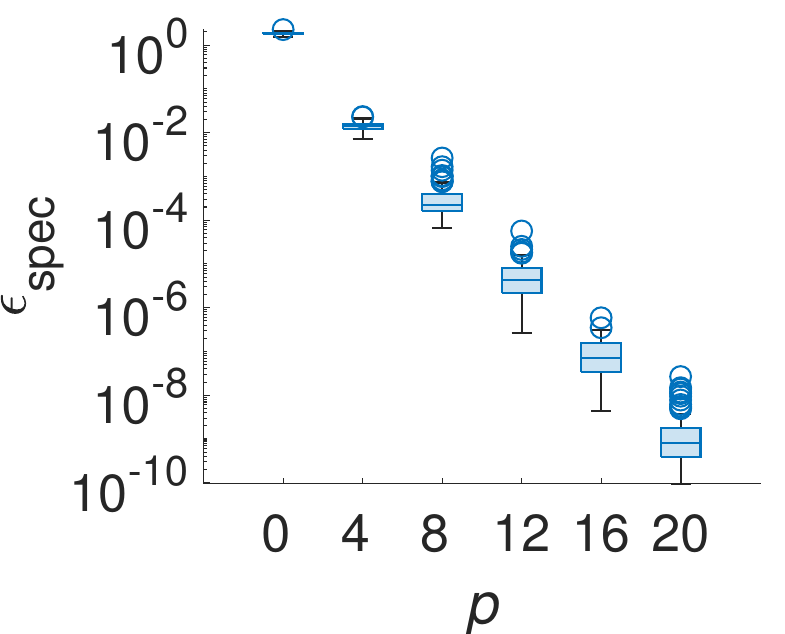}
					\includegraphics[width=3.4cm, trim=103 265 115 273,clip]{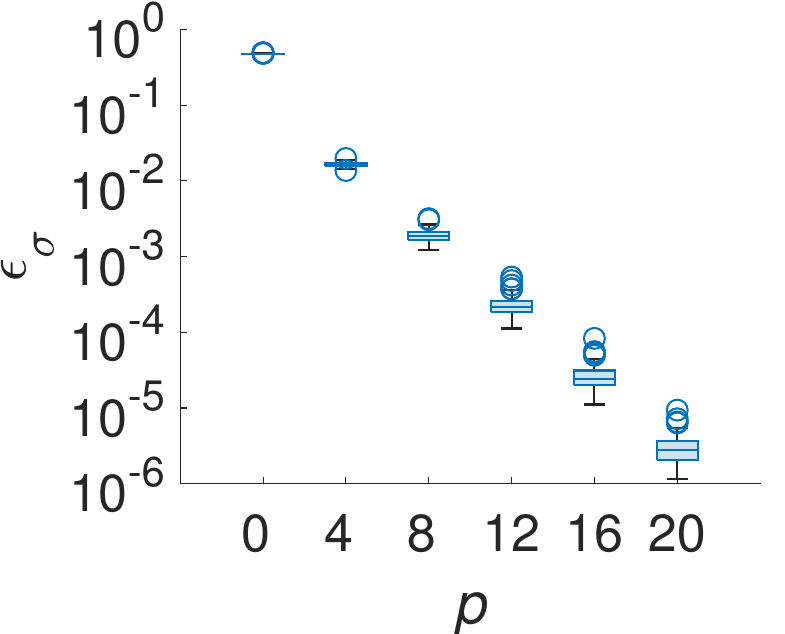}
				\end{minipage}
			\caption{\notice{The whisker plots of four error metrics vs. $p$ of dashSVD for SNAP ($k=100$).}}
			\label{fig:err_snap_var}
			\centering
		\end{figure}
	\bibliographystyle{ACM-Reference-Format}
	\bibliography{icml22}


\begin{thebibliography}{34}


\ifx \showCODEN    \undefined \def \showCODEN     #1{\unskip}     \fi
\ifx \showDOI      \undefined \def \showDOI       #1{#1}\fi
\ifx \showISBNx    \undefined \def \showISBNx     #1{\unskip}     \fi
\ifx \showISBNxiii \undefined \def \showISBNxiii  #1{\unskip}     \fi
\ifx \showISSN     \undefined \def \showISSN      #1{\unskip}     \fi
\ifx \showLCCN     \undefined \def \showLCCN      #1{\unskip}     \fi
\ifx \shownote     \undefined \def \shownote      #1{#1}          \fi
\ifx \showarticletitle \undefined \def \showarticletitle #1{#1}   \fi
\ifx \showURL      \undefined \def \showURL       {\relax}        \fi
\providecommand\bibfield[2]{#2}
\providecommand\bibinfo[2]{#2}
\providecommand\natexlab[1]{#1}
\providecommand\showeprint[2][]{arXiv:#2}

\bibitem[\protect\citeauthoryear{??}{Ami}{2021}]%
        {Aminer}
 \bibinfo{year}{2021}\natexlab{}.
\newblock \bibinfo{title}{Aminer}.
\newblock \bibinfo{howpublished}{\url{https://www.aminer.cn}}.
\newblock


\bibitem[\protect\citeauthoryear{??}{Int}{2021}]%
        {Intel}
 \bibinfo{year}{2021}\natexlab{}.
\newblock \bibinfo{title}{Intel one{API} {M}ath {K}ernel {L}ibrary}.
\newblock
  \bibinfo{howpublished}{\url{https://software.intel.com/content/www/us/en/develop/tools/oneapi/components/onemkl.html}}.
\newblock


\bibitem[\protect\citeauthoryear{??}{ran}{2021}]%
        {randqb-code}
 \bibinfo{year}{2021}\natexlab{}.
\newblock \bibinfo{title}{RandQR}.
\newblock
  \bibinfo{howpublished}{\url{https://users.oden.utexas.edu/~pgm/Codes/randqb_codes_intel_mkl.zip}}.
\newblock


\bibitem[\protect\citeauthoryear{??}{frP}{2022}]%
        {frPCA-code}
 \bibinfo{year}{2022}\natexlab{}.
\newblock \bibinfo{title}{frPCA\_sparse}.
\newblock
  \bibinfo{howpublished}{\url{https://github.com/XuFengthucs/frPCA_sparse}}.
\newblock


\bibitem[\protect\citeauthoryear{??}{pri}{2022}]%
        {primme-code}
 \bibinfo{year}{2022}\natexlab{}.
\newblock \bibinfo{title}{primme}.
\newblock \bibinfo{howpublished}{\url{https://github.com/primme/primme}}.
\newblock


\bibitem[\protect\citeauthoryear{Allen-Zhu and Li}{Allen-Zhu and Li}{2016}]%
        {LazySVD}
\bibfield{author}{\bibinfo{person}{Zeyuan Allen-Zhu} {and}
  \bibinfo{person}{Yuanzhi Li}.} \bibinfo{year}{2016}\natexlab{}.
\newblock \showarticletitle{Lazy{SVD}: Even faster {SVD} decomposition yet
  without agonizing pain}. In \bibinfo{booktitle}{\emph{Advances in Neural
  Information Processing Systems}}. \bibinfo{pages}{974--982}.
\newblock


\bibitem[\protect\citeauthoryear{ARB}{ARB}{2022}]%
        {openMP}
\bibfield{author}{\bibinfo{person}{The~OpenMP ARB}.}
  \bibinfo{year}{2022}\natexlab{}.
\newblock \bibinfo{title}{The OpenMP API specification for parallel
  programming}.
\newblock \bibinfo{howpublished}{\url{https://www.openmp.org/}}.
\newblock


\bibitem[\protect\citeauthoryear{{Baglama} and {Reichel}}{{Baglama} and
  {Reichel}}{2005}]%
        {baglama2005augmented}
\bibfield{author}{\bibinfo{person}{James {Baglama}} {and}
  \bibinfo{person}{Lothar {Reichel}}.} \bibinfo{year}{2005}\natexlab{}.
\newblock \showarticletitle{Augmented implicitly restarted Lanczos
  bidiagonalization methods}.
\newblock \bibinfo{journal}{\emph{SIAM Journal on Scientific Computing}}
  \bibinfo{volume}{27}, \bibinfo{number}{1} (\bibinfo{year}{2005}),
  \bibinfo{pages}{19--42}.
\newblock


\bibitem[\protect\citeauthoryear{Balay, Abhyankar, Adams, Benson, Brown, Brune,
  Buschelman, Constantinescu, Dalcin, Dener, Eijkhout, Gropp, Hapla, Isaac,
  Jolivet, Karpeev, Kaushik, Knepley, Kong, Kruger, May, McInnes, Mills,
  Mitchell, Munson, Roman, Rupp, Sanan, Sarich, Smith, Zampini, Zhang, Zhang,
  and Zhang}{Balay et~al\mbox{.}}{2022}]%
        {petsc-web-page}
\bibfield{author}{\bibinfo{person}{Satish Balay}, \bibinfo{person}{Shrirang
  Abhyankar}, \bibinfo{person}{Mark~F. Adams}, \bibinfo{person}{Steven Benson},
  \bibinfo{person}{Jed Brown}, \bibinfo{person}{Peter Brune},
  \bibinfo{person}{Kris Buschelman}, \bibinfo{person}{Emil~M. Constantinescu},
  \bibinfo{person}{Lisandro Dalcin}, \bibinfo{person}{Alp Dener},
  \bibinfo{person}{Victor Eijkhout}, \bibinfo{person}{William~D. Gropp},
  \bibinfo{person}{V\'{a}clav Hapla}, \bibinfo{person}{Tobin Isaac},
  \bibinfo{person}{Pierre Jolivet}, \bibinfo{person}{Dmitry Karpeev},
  \bibinfo{person}{Dinesh Kaushik}, \bibinfo{person}{Matthew~G. Knepley},
  \bibinfo{person}{Fande Kong}, \bibinfo{person}{Scott Kruger},
  \bibinfo{person}{Dave~A. May}, \bibinfo{person}{Lois~Curfman McInnes},
  \bibinfo{person}{Richard~Tran Mills}, \bibinfo{person}{Lawrence Mitchell},
  \bibinfo{person}{Todd Munson}, \bibinfo{person}{Jose~E. Roman},
  \bibinfo{person}{Karl Rupp}, \bibinfo{person}{Patrick Sanan},
  \bibinfo{person}{Jason Sarich}, \bibinfo{person}{Barry~F. Smith},
  \bibinfo{person}{Stefano Zampini}, \bibinfo{person}{Hong Zhang},
  \bibinfo{person}{Hong Zhang}, {and} \bibinfo{person}{Junchao Zhang}.}
  \bibinfo{year}{2022}\natexlab{}.
\newblock \bibinfo{title}{{PETS}c {W}eb page}.
\newblock \bibinfo{howpublished}{\url{https://petsc.org/}}.
\newblock


\bibitem[\protect\citeauthoryear{Benjamin~Erichson, Brunton, and
  Nathan~Kutz}{Benjamin~Erichson et~al\mbox{.}}{2017}]%
        {Erichson_2017_ICCV}
\bibfield{author}{\bibinfo{person}{N. Benjamin~Erichson},
  \bibinfo{person}{Steven~L. Brunton}, {and} \bibinfo{person}{J. Nathan~Kutz}.}
  \bibinfo{year}{2017}\natexlab{}.
\newblock \showarticletitle{Compressed singular value decomposition for image
  and video processing}. In \bibinfo{booktitle}{\emph{Proc. IEEE International
  Conference on Computer Vision (ICCV)}}. \bibinfo{pages}{1880--1888}.
\newblock


\bibitem[\protect\citeauthoryear{Boldi, Rosa, Santini, and Vigna}{Boldi
  et~al\mbox{.}}{2011}]%
        {BRSLLP}
\bibfield{author}{\bibinfo{person}{Paolo Boldi}, \bibinfo{person}{Marco Rosa},
  \bibinfo{person}{Massimo Santini}, {and} \bibinfo{person}{Sebastiano Vigna}.}
  \bibinfo{year}{2011}\natexlab{}.
\newblock \showarticletitle{Layered Label Propagation: A MultiResolution
  Coordinate-Free Ordering for Compressing Social Networks}. In
  \bibinfo{booktitle}{\emph{Proc. the 20th international conference on World
  Wide Web}}. \bibinfo{publisher}{ACM Press}, \bibinfo{pages}{587--596}.
\newblock


\bibitem[\protect\citeauthoryear{Boldi and Vigna}{Boldi and Vigna}{2004}]%
        {BoVWFI}
\bibfield{author}{\bibinfo{person}{Paolo Boldi} {and}
  \bibinfo{person}{Sebastiano Vigna}.} \bibinfo{year}{2004}\natexlab{}.
\newblock \showarticletitle{The {W}eb{G}raph Framework {I}: {C}ompression
  Techniques}. In \bibinfo{booktitle}{\emph{Proc. the Thirteenth International
  World Wide Web Conference (WWW 2004)}}. \bibinfo{publisher}{ACM Press},
  \bibinfo{address}{Manhattan, USA}, \bibinfo{pages}{595--601}.
\newblock


\bibitem[\protect\citeauthoryear{Davis and Hu}{Davis and Hu}{2011}]%
        {davis2011university}
\bibfield{author}{\bibinfo{person}{Timothy~A Davis} {and}
  \bibinfo{person}{Yifan Hu}.} \bibinfo{year}{2011}\natexlab{}.
\newblock \showarticletitle{The University of Florida sparse matrix
  collection}.
\newblock \bibinfo{journal}{\emph{ACM Transactions on Mathematical Software
  (TOMS)}} \bibinfo{volume}{38}, \bibinfo{number}{1} (\bibinfo{year}{2011}),
  \bibinfo{pages}{1--25}.
\newblock


\bibitem[\protect\citeauthoryear{Ding, Yu, Xie, and Liu}{Ding
  et~al\mbox{.}}{2020}]%
        {ding2020efficient}
\bibfield{author}{\bibinfo{person}{Xiangyun Ding}, \bibinfo{person}{Wenjian
  Yu}, \bibinfo{person}{Yuyang Xie}, {and} \bibinfo{person}{Shenghua Liu}.}
  \bibinfo{year}{2020}\natexlab{}.
\newblock \showarticletitle{Efficient model-based collaborative filtering with
  fast adaptive {PCA}}. In \bibinfo{booktitle}{\emph{2020 IEEE 32nd
  International Conference on Tools with Artificial Intelligence (ICTAI)}}.
  IEEE, \bibinfo{pages}{955--960}.
\newblock


\bibitem[\protect\citeauthoryear{Eckart and Young}{Eckart and Young}{1936}]%
        {eckart1936}
\bibfield{author}{\bibinfo{person}{Carl Eckart} {and} \bibinfo{person}{Gale
  Young}.} \bibinfo{year}{1936}\natexlab{}.
\newblock \showarticletitle{The approximation of one matrix by another of lower
  rank}.
\newblock \bibinfo{journal}{\emph{Psychometrika}} \bibinfo{volume}{1},
  \bibinfo{number}{3} (\bibinfo{year}{1936}), \bibinfo{pages}{211--218}.
\newblock


\bibitem[\protect\citeauthoryear{Feng, Xie, Song, Yu, and Tang}{Feng
  et~al\mbox{.}}{2018}]%
        {pmlr-v95-feng18a}
\bibfield{author}{\bibinfo{person}{Xu Feng}, \bibinfo{person}{Yuyang Xie},
  \bibinfo{person}{Mingye Song}, \bibinfo{person}{Wenjian Yu}, {and}
  \bibinfo{person}{Jie Tang}.} \bibinfo{year}{2018}\natexlab{}.
\newblock \showarticletitle{Fast randomized {PCA} for sparse data}. In
  \bibinfo{booktitle}{\emph{Proc. the 10th Asian Conference on Machine Learning
  (ACML)}}. \bibinfo{pages}{710--725}.
\newblock


\bibitem[\protect\citeauthoryear{Golub and Van~Loan}{Golub and
  Van~Loan}{2012}]%
        {matrix2012}
\bibfield{author}{\bibinfo{person}{Gene~H Golub} {and}
  \bibinfo{person}{Charles~F Van~Loan}.} \bibinfo{year}{2012}\natexlab{}.
\newblock \bibinfo{booktitle}{\emph{Matrix Computations}}.
\newblock \bibinfo{publisher}{JHU Press}.
\newblock


\bibitem[\protect\citeauthoryear{Halko, Martinsson, and Tropp}{Halko
  et~al\mbox{.}}{2011}]%
        {Halko2011Finding}
\bibfield{author}{\bibinfo{person}{N. Halko}, \bibinfo{person}{P.~G.
  Martinsson}, {and} \bibinfo{person}{J.~A. Tropp}.}
  \bibinfo{year}{2011}\natexlab{}.
\newblock \showarticletitle{Finding structure with randomness: Probabilistic
  algorithms for constructing approximate matrix decompositions}.
\newblock \bibinfo{journal}{\emph{SIAM Rev.}} \bibinfo{volume}{53},
  \bibinfo{number}{2} (\bibinfo{year}{2011}), \bibinfo{pages}{217--288}.
\newblock


\bibitem[\protect\citeauthoryear{Harper and Konstan}{Harper and
  Konstan}{2016}]%
        {movielens}
\bibfield{author}{\bibinfo{person}{F.~Maxwell Harper} {and}
  \bibinfo{person}{Joseph~A. Konstan}.} \bibinfo{year}{2016}\natexlab{}.
\newblock \showarticletitle{The {M}ovielens datasets: History and context}.
\newblock \bibinfo{journal}{\emph{{ACM} Transactions on Interactive Intelligent
  Systems (TiiS)}} \bibinfo{volume}{5}, \bibinfo{number}{4}
  (\bibinfo{year}{2016}), \bibinfo{pages}{19}.
\newblock


\bibitem[\protect\citeauthoryear{Heath}{Heath}{2018}]%
        {heath2018scientific}
\bibfield{author}{\bibinfo{person}{Michael~T Heath}.}
  \bibinfo{year}{2018}\natexlab{}.
\newblock \bibinfo{booktitle}{\emph{Scientific Computing: An Introductory
  Survey, Revised Second Edition}}.
\newblock \bibinfo{publisher}{SIAM}.
\newblock


\bibitem[\protect\citeauthoryear{Hernandez, Roman, and Tomas}{Hernandez
  et~al\mbox{.}}{2008}]%
        {andez2008robust}
\bibfield{author}{\bibinfo{person}{Vicente Hernandez},
  \bibinfo{person}{Jos{\'e}~E. Roman}, {and} \bibinfo{person}{Andr{\'e}s
  Tomas}.} \bibinfo{year}{2008}\natexlab{}.
\newblock \showarticletitle{A robust and efficient parallel {SVD} solver based
  on restarted Lanczos bidiagonalization}.
\newblock \bibinfo{journal}{\emph{Electronic Transactions on Numerical
  Analysis}}  \bibinfo{volume}{31} (\bibinfo{year}{2008}),
  \bibinfo{pages}{68--85}.
\newblock


\bibitem[\protect\citeauthoryear{Horn and Johnson}{Horn and Johnson}{1991}]%
        {horn1991topics}
\bibfield{author}{\bibinfo{person}{Roger~A. Horn} {and}
  \bibinfo{person}{Charles~R. Johnson}.} \bibinfo{year}{1991}\natexlab{}.
\newblock \bibinfo{booktitle}{\emph{Topics in Matrix Analysis}}.
\newblock \bibinfo{publisher}{Cambridge University Press}.
\newblock
\urldef\tempurl%
\url{https://doi.org/10.1017/CBO9780511840371}
\showDOI{\tempurl}


\bibitem[\protect\citeauthoryear{Larsen}{Larsen}{2004}]%
        {propack}
\bibfield{author}{\bibinfo{person}{Rasmus~Munk Larsen}.}
  \bibinfo{year}{2004}\natexlab{}.
\newblock \showarticletitle{PROPACK-Software for large and sparse {SVD}
  calculations}.
\newblock \bibinfo{journal}{\emph{Available online.
  \url{http://sun.stanford.edu/~rmunk/PROPACK}}} (\bibinfo{year}{2004}).
\newblock


\bibitem[\protect\citeauthoryear{Leskovec and Krevl}{Leskovec and
  Krevl}{2014}]%
        {snapnets}
\bibfield{author}{\bibinfo{person}{Jure Leskovec} {and} \bibinfo{person}{Andrej
  Krevl}.} \bibinfo{year}{2014}\natexlab{}.
\newblock \bibinfo{title}{{SNAP datasets}: {Stanford} large network dataset
  collection}.
\newblock \bibinfo{howpublished}{\url{http://snap.stanford.edu/data}}.
\newblock


\bibitem[\protect\citeauthoryear{Li, Linderman, Szlam, Stanton, Kluger, and
  Tygert}{Li et~al\mbox{.}}{2017}]%
        {alg971}
\bibfield{author}{\bibinfo{person}{H. Li}, \bibinfo{person}{G.~C. Linderman},
  \bibinfo{person}{A. Szlam}, \bibinfo{person}{K.~P. Stanton},
  \bibinfo{person}{Y. Kluger}, {and} \bibinfo{person}{M. Tygert}.}
  \bibinfo{year}{2017}\natexlab{}.
\newblock \showarticletitle{Algorithm 971: An implementation of a randomized
  algorithm for principal component analysis.}
\newblock \bibinfo{journal}{\emph{ACM Trans. Math. Software}}
  \bibinfo{volume}{43}, \bibinfo{number}{3} (\bibinfo{year}{2017}),
  \bibinfo{pages}{1--14}.
\newblock


\bibitem[\protect\citeauthoryear{Mahoney}{Mahoney}{2011}]%
        {mahoney2011}
\bibfield{author}{\bibinfo{person}{Michael~W. Mahoney}.}
  \bibinfo{year}{2011}\natexlab{}.
\newblock \showarticletitle{Randomized algorithms for matrices and data}.
\newblock \bibinfo{journal}{\emph{Foundations and Trends{\textregistered} in
  Machine Learning}} \bibinfo{volume}{3}, \bibinfo{number}{2}
  (\bibinfo{year}{2011}), \bibinfo{pages}{123--224}.
\newblock


\bibitem[\protect\citeauthoryear{{Martinsson} and {Tropp}}{{Martinsson} and
  {Tropp}}{2020}]%
        {martinsson2020randomized}
\bibfield{author}{\bibinfo{person}{Per-Gunnar {Martinsson}} {and}
  \bibinfo{person}{Joel~A. {Tropp}}.} \bibinfo{year}{2020}\natexlab{}.
\newblock \showarticletitle{Randomized numerical linear algebra: Foundations
  and algorithms}.
\newblock \bibinfo{journal}{\emph{Acta Numerica}}  \bibinfo{volume}{29}
  (\bibinfo{year}{2020}), \bibinfo{pages}{403--572}.
\newblock


\bibitem[\protect\citeauthoryear{Martinsson and Voronin}{Martinsson and
  Voronin}{2016}]%
        {martinsson2016randomized2}
\bibfield{author}{\bibinfo{person}{P.~G. Martinsson} {and} \bibinfo{person}{S.
  Voronin}.} \bibinfo{year}{2016}\natexlab{}.
\newblock \showarticletitle{A randomized blocked algorithm for efficiently
  computing rank-revealing factorizations of matrices}.
\newblock \bibinfo{journal}{\emph{SIAM J. Sci. Comput}}  \bibinfo{volume}{38}
  (\bibinfo{year}{2016}), \bibinfo{pages}{S485--–S507}.
\newblock


\bibitem[\protect\citeauthoryear{Musco and Musco}{Musco and Musco}{2015}]%
        {musco2015}
\bibfield{author}{\bibinfo{person}{Cameron Musco} {and}
  \bibinfo{person}{Christopher Musco}.} \bibinfo{year}{2015}\natexlab{}.
\newblock \showarticletitle{Randomized block {K}rylov methods for stronger and
  faster approximate singular value decomposition}. In
  \bibinfo{booktitle}{\emph{Advances in Neural Information Processing
  Systems}}. \bibinfo{pages}{1396--1404}.
\newblock


\bibitem[\protect\citeauthoryear{Oleg, Matthias, Laura, and Victor}{Oleg
  et~al\mbox{.}}{2023}]%
        {ICML2023}
\bibfield{author}{\bibinfo{person}{Balabanov Oleg}, \bibinfo{person}{Beaupère
  Matthias}, \bibinfo{person}{Grigori Laura}, {and} \bibinfo{person}{Lederer
  Victor}.} \bibinfo{year}{2023}\natexlab{}.
\newblock \showarticletitle{Block subsampled randomized Hadamard transform for
  {Nyström} approximation on distributed architectures}. In
  \bibinfo{booktitle}{\emph{Proc. International Conference on Machine
  Learning}}. \bibinfo{pages}{1564--1576}.
\newblock


\bibitem[\protect\citeauthoryear{Rokhlin, Szlam, and Tygert}{Rokhlin
  et~al\mbox{.}}{2010}]%
        {rokhlin2010randomized}
\bibfield{author}{\bibinfo{person}{Vladimir Rokhlin}, \bibinfo{person}{Arthur
  Szlam}, {and} \bibinfo{person}{Mark Tygert}.}
  \bibinfo{year}{2010}\natexlab{}.
\newblock \showarticletitle{A randomized algorithm for principal component
  analysis}.
\newblock \bibinfo{journal}{\emph{SIAM J. Matrix Anal. Appl.}}
  \bibinfo{volume}{31}, \bibinfo{number}{3} (\bibinfo{year}{2010}),
  \bibinfo{pages}{1100--1124}.
\newblock


\bibitem[\protect\citeauthoryear{Stathopoulos and McCombs}{Stathopoulos and
  McCombs}{2010}]%
        {2010primme}
\bibfield{author}{\bibinfo{person}{Andreas Stathopoulos} {and}
  \bibinfo{person}{James~R McCombs}.} \bibinfo{year}{2010}\natexlab{}.
\newblock \showarticletitle{PRIMME: Preconditioned iterative multimethod
  eigensolver: Methods and software description}.
\newblock \bibinfo{journal}{\emph{ACM Transactions on Mathematical Software
  (TOMS)}} \bibinfo{volume}{37}, \bibinfo{number}{2} (\bibinfo{year}{2010}),
  \bibinfo{pages}{1--30}.
\newblock


\bibitem[\protect\citeauthoryear{Voronin and Martinsson}{Voronin and
  Martinsson}{2015}]%
        {rsvdpack}
\bibfield{author}{\bibinfo{person}{Sergey Voronin} {and}
  \bibinfo{person}{Per-Gunnar Martinsson}.} \bibinfo{year}{2015}\natexlab{}.
\newblock \showarticletitle{{RSVDPACK}: An implementation of randomized
  algorithms for computing the singular value, interpolative, and {CUR}
  decompositions of matrices on multi-core and {GPU} architectures}.
\newblock \bibinfo{journal}{\emph{arXiv preprint arXiv:1502.05366}}
  (\bibinfo{year}{2015}).
\newblock


\bibitem[\protect\citeauthoryear{Wu, Romero, and Stathopoulos}{Wu
  et~al\mbox{.}}{2017}]%
        {wu2017primme_svds}
\bibfield{author}{\bibinfo{person}{Lingfei Wu}, \bibinfo{person}{Eloy Romero},
  {and} \bibinfo{person}{Andreas Stathopoulos}.}
  \bibinfo{year}{2017}\natexlab{}.
\newblock \showarticletitle{{PRIMME\_SVDS}: A high-performance preconditioned
  {SVD} solver for accurate large-scale computations}.
\newblock \bibinfo{journal}{\emph{SIAM Journal on Scientific Computing}}
  \bibinfo{volume}{39}, \bibinfo{number}{5} (\bibinfo{year}{2017}),
  \bibinfo{pages}{S248--S271}.
\newblock


\end{thebibliography}

\end{document}